\documentclass[a4paper,11pt,english,titlepage]{book}

\usepackage{tamasty}

\usepackage{a4wide}
\usepackage{setspace}
\usepackage{amsmath,amsthm,amssymb}
\usepackage{makeidx}
\usepackage[numbers]{natbib}
\usepackage{enumerate}
\usepackage{fancyvrb}
\usepackage{graphicx,subfigure}
\usepackage{babel}
\usepackage{hyperref}
\usepackage[font={small,it}]{caption}
\makeindex

\newtheorem{theorem}{Theorem}[section]

\newtheorem{lemma}[theorem]{Lemma}
\newtheorem{proposition}{Proposition}

\newtheorem{observation}{Observation}

\begin{document}
\newcommand{\ee}{\`e\ }
\renewcommand{\aa}{\`a\ }
\renewcommand{\AA}{\`A\ }
\newcommand{\oo}{\`o\ }
\newcommand{\uu}{\`u\ }
\newcommand{\ket}[1]{\ensuremath {|\: #1 \: \rangle}}
\newcommand{\bra}[1]{\ensuremath{\langle \: #1 \:|}}
\newcommand{\braket}[2]{\ensuremath{\langle \: #1 \: | \: #2 \: \rangle}}
\newcommand{\ketbra}[2]{\ensuremath{| \: #1 \:\rangle \langle \: #2 \:  |}}
\newcommand{\ves}[2]{\ensuremath{#1_1,#1_2, \ldots, #1_{#2}}}
\newcommand{\ve}[2]{\ensuremath{#1(1),#1(2) \ldots, #1(#2)}}
\newcommand{\Hc}{\ensuremath{\mathcal{H}_{cursor\;}}}
\newcommand{\Hr}{\ensuremath{\mathcal{H}_{register\;}}}
\newcommand{\Hcr}{\ensuremath{\mathcal{K}_{cur\;}}}
\newcommand{\Hrr}{\ensuremath{\mathcal{K}_{reg\;}}}
\newcommand{\Hm}{\ensuremath{\mathcal{H}_{machine\;}}}
\newcommand{\eref}[1]{(\ref{#1})}
\newcommand{\sref}[1]{section~\ref{#1}}
\newcommand{\fref}[1]{figure~\ref{#1}}
\newcommand{\tref}[1]{table~\ref{#1}}
\newcommand{\Eref}[1]{Equation (\ref{#1})}
\newcommand{\Sref}[1]{Section~\ref{#1}}
\newcommand{\Fref}[1]{Figure~\ref{#1}}
\newcommand{\Tref}[1]{Table~\ref{#1}}
\renewcommand{\theenumi}{(\roman{enumi})}   
\renewcommand{\labelenumi}{\theenumi}
\frontmatter
\thispagestyle{empty}
%
%
\begin{center}
\Large UNIVERSIT\AA DEGLI STUDI DI MILANO \par 
\normalsize Facolt\aa di Scienze Matematiche Fisiche e Naturali \par 
\normalsize Dipartimento di Matematica ``F. Enriques''\par \bigskip
{\scshape Corso di Dottorato di Ricerca in Matematica e Statistica\par
 per le Scienze Computazionali} -Ciclo XIX-\par \smallskip
{\scshape Tesi di Dottorato di Ricerca}\par\vspace{2cm}
{\LARGE \textsc INTERACTING QUANTUM WALKS} 
\par
\bigskip
{\scshape Settori disciplinari: INF/01, MAT/06}\par
\end{center}
\vspace{2cm}
\begin{flushright}
by \par \smallskip
 Dario \textsc{Tamascelli} \par 
\vspace{1cm}
\end{flushright}
\begin{flushleft}
Advisor: \par \smallskip
  Prof. Diego \textsc{de Falco} \par \smallskip
  Dipartimento di Scienze dell'Informazione \par \smallskip
  Universit\aa degli Studi di Milano\par \bigskip
\vspace{0.5cm}
MaSSC Ph.D. Coordinator: \par \smallskip
  Prof. Giovanni \textsc{Naldi} \par \smallskip
  Dipartimento di Matematica ``F. Enriques'' \par \smallskip
  Universit\aa degli Studi di Milano
\end{flushleft}
\vspace{1.4cm}
\begin{center}
\Large \textsc{Anno accademico 2006/07}
\end{center}
\newpage
\mbox{   }
\vspace{5cm}
\begin{flushright}
\itshape To my family.
\end{flushright}
\newpage
\section*{Acknowledgements}
There are a few people I wish to thank. First of all, I wish to express all my gratitude to prof. Diego de Falco, my advisor, for his constant encouragement, human support, constructive criticism and infinite patience.\\
I wish also thank prof. Alberto Bertoni: his ability to stimulate my curiosity has played a fundamental role since when I was a freshmen.\\
I also thank prof. Bruno Apolloni for hosting me at the LaReN laboratory of the Information Science Department of the University of Milano. There I have found a very stimulating environment and collegues  which, by the way, heppened to become very good friends.\\[5pt]
A special thank goes to prof. Giovanni Naldi, coordinator of the MaSSC doctoral supervisory committee.\\[5pt]
Finally, I wish to thank my family and my friends for their daily patience and support.
\tableofcontents
%
\chapter{Introduction}
Computation is a physical process \cite{landauer82} and the notion of a \emph{computable function}, say $f$, relies on the possibility of implementing a physical process transforming the input state $S_0$ in the desired output state $f(S_0)$ \cite{bennett82,toffoli82,vonNeumann76}.
\footnote{In \cite{toffoli82} Toffoli illustrated the equivalence between a physical experiment and a computing process by means of the following suggestive example:
\begin{quote}
Let us suppose that intelligent beings are observing us from a far away star. How could they understand whether we are carrying on a computation or a physical experiment? They could not understand it from what we are doing since there is no objective difference. The difference is in our intentions, in our knowledge, in our expectations.
\end{quote}
}\\
The Church-Turing hypothesis \cite{church36,turing36,kleene43}:
\begin{quote}
Every function that would \emph{naturally be regarded as computable} can be computed by a universal Turing machine
\end{quote}
makes itself an implicit physical assertion which is explicitly stated in the \emph{Church-Turing-Deutsch} hypothesis \cite{deutsch85}:
\begin{quote}
Every finitely realizable physical system can be perfectly simulated by a universal model computing machine operating by finite means.
\end{quote}
Following this principle, a universal computing machine operating by finite means is able to simulate the evolution of a bundle of, say, $n$ interacting electrons. Feynman showed \cite{feyn82} that the complexity of the simulation, on a classical computer, of a quantum mechanical system scales exponentially with its dimension; but Feynman himself pointed out that quantum systems are more ``suited'' for the simulation of other quantum systems;  by ``suited'' he meant that a  logarithmic reduction of the complexity of the simulation is achievable by means of a universal model quantum computing machine.\\
\emph{Quantum computation} origins from a question which is strongly suggested by Feynman's considerations: are there other \emph{hard} computational problems which can exploit the features of quantum mechanical systems to be efficiently solved?\\
Since the seminal work of Feynman  \emph{quantum computation} has known an enormous growth and nowadays it is a mature and vast research field in between physics and computer science. We refer to \cite{nielsen} for an exhaustive introduction to the field of quantum computation, quantum information and quantum communication.\\[5pt]
Another problem motivating the investigation of quantum computing devices is the technological advancement of semiconductor industry. A state of the art MOSFET (Metal Oxided Semiconductor Field Effect Transistors) has dimension of order $10^{-8}$m, the Bohr radius is approximately $10^{-10}$m; following Moore's law \footnote{Moore's Law is the empirical observation that the transistor density of integrated circuits, with respect to minimum component cost, doubles every 24 months. It is attributed to Gordon E. Moore, a co-founder of Intel.}, which has been, up to now, very accurately verified, in a few years the dimension of a MOSFET will reduce to fractions of a nanometer. At this scale, the evolution of the gates will be described by quantum mechanics. Indeed, quantum corrected diffusion models have already been introduced in the design process of nanoscale semiconductor devices (see \cite{carlodefalco} for an updated list of references).\\[10pt]
In this work we present the quantum mechanical computer proposed by Feynman in 1985 and, since then, widely cited but seldom used. The main feature of the model is the presence of a builtin clocking mechanism managing for the ordered application of the computational primitives to the input/output register.\\
In fact, given a transformation $A$ to be applied to the input/output register, quantum computation starts from the decomposition of $A$ into the sequential application of simpler computational primitives, $ U_1, U_2,\ldots,U_n$, which are unitary operators  acting on few qubits of the input/output register at a time. The evolution of the input state into the output state is then seen as a discrete, stepwise, process: at any time step an operation is performed on the input/output register.\\
In general, in quantum mechanics, the outgoing state at time $t$ for a system with \emph{time independent} Hamiltonian $H$ is $e^{-i H t} \ket{\psi_0}$ , where \ket{\psi_0} is the input state; in other words, the state of the system evolves under the action of the unitary group generated by the Hamiltonian $H$. It appears to be very difficult to find, for a given special time $\bar{t}$, the Hamiltonian which will produce $A = e^{-i H \bar{t}}$ when $A$ is a product of \emph{noncommuting} matrices $U_1,U_2,\ldots, U_n$, from some simple property of the matrices themselves \footnote{For example, following the method proposed by Benioff \cite{benioff82} it is possible to define a time independent Hamiltonian guiding the desired evolution from any input to any output state. The problem is that the explicit construction of such a Hamiltonian requires the prior knowledge of every step in the solution of every problem which the computer can solve}.\\
If one accepts time dependent Hamiltonians, it is fairly straightforward to write a Hamiltonian for the evolution of the input state into the output state via the indicated intermediate steps; keeping in mind the mechanism of a synchronous system, one can imagine an internal  clock turning the interactions on and off. This idealized clocking mechanism is however not satisfactory: a classical macroscopic clock would destroy the coherence of the quantum system; so, to be consistent, the clock itself should be quantized; but as soon as we do it, it becomes clear that this clock would be affected by whatever it interacts with \cite{peres80} and time steps would be blurred.\\
The problem of explicitly defining a time independent Hamiltonian driving the input state to the output state through the intermediate states determined by the \emph{ordered application} of the computational primitives was overcome by Feynman in 1985 \cite{feyn86}.\\[5pt]
There are other interesting aspects of Feynman's proposal for a quantum computer. For example, it has been observed by Margolus \cite{margolus90} that Feynman, in his model, `managed to arrange for all the quantum uncertainty[...] to be concentrated in the time taken for the computation to be completed, rather than in the correctness of the answer'.\\ 
More recently,  Levitin and Margolus \cite{levitin98} related the maximum rate of information processing by a quantum computer to the available, \emph{conserved}, energy. It is therefore of some theoretical interest to revisit a model, such as Feynman's, based on a closed system, evolving according to a time independent and, therefore, conserved Hamiltonian.\\
Moreover, the doubt has been raised by Alicki \cite{alicki00} that `the idea that the physical time[...] of computation is proportional to the complexity [...] is only true of the existing digital computers which are ensembles of controlled bistable elements which [...] can literally mimic logical operations'. Feynman's model provides an ideal context for the study of this issue: timing is modeled by a cursor, which jumps along a sequence of sites, indicating that the corresponding discrete operations should be applied. \\[10pt]
%
%
The thesis is organized as follows.\\
In Chapter \ref{chap:fqc} we present the model, the basic clocking mechanism and establish our notation.\\
In Chapter \ref{chap:grover} we use Grover's algorithm as a case study to introduce the full model, which we call an \emph{interacting} $XY$ \emph{system}; particular attention is paid to the role of additional controlling spins in implementing successive visits to selected parts of the flow chart, or graph, of the algorithm in iterated computations (\emph{quantum subroutines}) and to the \emph{locality} of implementation.\\
In Chapter \ref{chap:quantumtiming} we study the dynamics of Feynman's quantum computer, the timing and synchronization problems related to a rescaling of the clock to the quantum regime and  propose a measurement scheme for  the storage of results of computation via \emph{telomeric chains}.\\
In Chapter \ref{chap:speedandentropy} we  pay specific attention to non-positional observables of the system: our main concerns will be speed (of computation), entropy (of controlled and/or controlling subsystem) and energy (of the system). In particular  we relate the speed of computation to the group velocity of the cursor wave packet along the graph and  discuss the buildup of entropy  in the \emph{clocked} subsystem caused by the spreading of the wave packet of the \emph{clocking agent}. An outline of possible choices of the initial form of the wave packet bringing the entropy buildup close to a minimum is given.\\ 
In Chapter \ref{chap:imperfections} we  consider the observable \emph{number of particles} (agents performing a \emph{quantum walk} along an $XY$ spin chain), discuss the interest and limitations of the proposal of a multi-hand quantum clock (or multi-agent spin networks) as a substitute for the loops implementing iterated applications of quantum subroutines.\\
The Conclusions and Outlook  chapter is devoted to an exposition of open problems and future line of research.\\

%
\mainmatter
%
%
\chapter{The Feynman machine} \label{chap:fqc}
We present Feynman's model of a quantum computer stressing the role of the clocking mechanism. We introduce the notion of logical successor of a given state and the related notion of Peres' constants of motion. We furthermore set the notation used throughout this work.
\section{The basic model} \label{sec:themodel}
It has been shown by Feynman \cite{feyn86} that it is possible to implement the sequential application, in the desired order, of the sequence 
\[
 U_{s-1} \cdot \ldots \cdot U_2 \cdot U_1=A
\]
of unitary operators to an \emph{input/output register} by using \emph{s} additional degrees of freedom: the \emph{program counter sites}.\\
For the sake of definiteness, we will think of each program counter  site $j=1,2, \dots, s$ as occupied by a spin-1/2 system $ \underline{\tau}(j)=(\tau_1(j),\tau_2(j),\tau_3(j))$. We will refer to the collection of such spins, which act in effect as a quantum clocking mechanism, as to a \emph{program line}.\\ 
The \emph{input/output register} will be, similarly, implemented by a collection of a certain number $\mu$ \ of spin-1/2 systems $\underline{\sigma}(i)=(\sigma_1(i),\sigma_2(i),\sigma_3(i)),\; i=1,2, \dots,\mu $.\\
We remind that the  angular momentum operators $\underline{\tau}(\cdot)$ satisfy the commutation rules of the Lie algebra on $SU_2$
\begin{equation} \label{eq:ccr}
 \left[ \tau_h(x),\tau_j(y) \right ]  =  \delta_{x,y} \epsilon_{hjk} \frac{i}{2} \tau_k(x)
\end{equation}
$\epsilon_{hjk}$ being the Levi-Civita symbol and $\delta_{x,y}$ the Kronecker delta. Obviously the same conditions are satisfied by the $\underline{\sigma}$'s.\\ 
We reflect the functional separation of the subsystems by calling  $\Hc = \mathcal{H}^{\otimes s}$, with $\mathcal{H}=\mathbb{C}^2$, the Hilbert space in which the collection of the spins making up the cursor are defined; analogously we refer to the Hilbert space the register is defined in as $\Hr = \mathcal{H}^{\otimes \mu}$. The overall system, \emph{register} + \emph{program line}  evolves in the space $\Hm =\Hr \otimes \Hc$ under the action of the Hamiltonian 
\begin{equation} \label{eq:hamfeynman}
H=-\frac{\lambda}{2}\Bigl ( \sum_{j=1}^{s-1}U_j \otimes\tau_+(j+1) \tau_-(j)  + U_j^{-1} \otimes\tau_+(j)\tau_-(j+1) \Bigr ).
\end{equation}
where
\begin{equation}
	\tau_{\pm}(j)=\frac{\tau_1(j)\pm i \tau_2(j)}{2}
\end{equation}
are respectively the \emph{raising} and \emph{lowering} (or \emph{excitation creation} and \emph{annihilation}) operators acting on the $j$-th spin of the cursor and $\lambda$ is a scalar coupling constant (for notational convenience, we will set $\lambda=1$ unless otherwise specified). For the sake of definiteness, for every spin of the system we will take the eigenstates of the $\tau_3$ (respectively $\sigma_3$) operator as the basis for the Hilbert space $\mathcal{H}$ of a single spin and indicate them by the eigenvalue $\pm 1$ of $\tau_3$ they belong to. This basis is conventionally referred to as \emph{computational basis}\index{computational basis} and the \emph{Pauli operators}\index{Pauli operators} $\tau_1,\ \tau_2,\ \tau_3$ have the usual matrix representation:
\begin{equation}
\tau_1 = \left (
\begin{array}{c c}
	0 & 1 \\
	1 & 0
\end{array}\right ), \;
\tau_2 = \left (
\begin{array}{c c}
	0 & -i \\
	i & 0
\end{array}\right ),\;
\tau_3 = \left (
\begin{array}{c c}
	1 & 0 \\
	0 & -1
\end{array}\right ).
\end{equation}
The evolution of the computing device is given by the solution of the Cauchy problem:
\begin{equation}
\begin{cases}
i \frac{d}{dt} | \psi(t) \rangle = H | \psi(t) \rangle \\
\ket{\psi(0)}=\ket{\psi_1}
\label{eq 1}
\end{cases}
\end{equation}
where we have set $\hbar = 1$ for notational convenience.\\
We define the operator \emph{number of excitation}
\begin{eqnarray}
N_3 & : & \Hm \rightarrow \Hm \nonumber\\
N_3 & = & I_r \otimes \sum_{j=1}^{s} \frac{1+\tau_3(j)}{2},
\end{eqnarray}
$I_r$ being the identity on the register subspace.
\begin{proposition} \label{th:number}
The number of spins up, or \emph{excitations} on the program line is a constant of motion;  namely
\begin{equation} \label{eq:n3}
[H,N_3]=0.
\end{equation}
\end{proposition}
\begin{proof}
Since $N_3$ acts only on the cursor, we can  omit in this proof explicit reference to the register degrees of freedom. From the commutation rules defined above it follows:
\begin{eqnarray}
	[H,N_3]& = & [\sum_{x=1}^{s-1} \tau_+(x+1) \tau_-(x) + \tau_+(x) \tau_-(x+1),\sum_{y=1}^s \frac{1+\tau_3(x)}{2}] = \nonumber \\
	& = &   \frac{1}{2} \sum_{x=1}^{s-1} \sum_{y=1}^s \left( [ \tau_+(x+1) \tau_-(x), \tau_3(x)] + [ \tau_+(x) \tau_-(x+1), \tau_3(x)] \right ) = \nonumber \\
	& = &  \frac{1}{2}  \sum_{x=1}^{s-1} ([\tau_+(x+1) \tau_-(x), \tau_3(x)] + [\tau_+(x+1) \tau_-(x), \tau_3(x+1)] + \nonumber \\
	& + & [\tau_+(x) \tau_-(x+1), \tau_3(x)] + [\tau_+(x) \tau_-(x+1), \tau_3(x+1)]) =\nonumber \\
	& = & \frac{1}{2} \sum_{x=1}^{s-1} (\tau_+(x+1)[ \tau_-(x), \tau_3(x)] +\tau_-(x) [\tau_+(x+1) , \tau_3(x+1)] + \nonumber \\
	& + & \tau_-(x+1)[\tau_+(x) , \tau_3(x)] + \tau_+(x) [ \tau_-(x+1), \tau_3(x+1)]) =\nonumber \\
	& = & - \tau_+(x+1) \tau_-(x) + \tau_-(x)\tau_+(x+1) + \nonumber \\
	& + & \tau_-(x+1)\tau_+(x)-\tau_+(x)\tau_-(x+1) =0
  \end{eqnarray}
\end{proof}
A first consequence of proposition \ref{th:number} is that if the initial state of the cursor belongs to the $N_3=k$ subspace of \Hc, the state vector of the cursor remains, for every time $t$, in the subspace of  \Hc spanned by the $\left ( 
\begin{array}{c}
	s \\
	k
\end{array}
\right )$ basis vector
\begin{equation}
	\tau_+(j_1) \tau_+(j_2)\ldots \tau_+(j_k)\ket{0}, 1\leq j_1 < j_2 < \ldots < j_k \leq s
\end{equation}
where with \ket{0} we indicate the ``all-down'' state, that is the state
\[
 \ket{\tau_3(1)=-1;\tau_3(2)=-1;\ldots;\tau_3(s)=-1}.
\] 
 The dimensionality of the state vector, and thus the complexity of any simulation of the system, is therefore significantly reduced. In particular, we will usually restrict our considerations to initial conditions belonging to the $N_3=1$ subspace and refer the cursor subspace to the orthonormal basis 
\begin{equation} \label{eq:cj}
	\{ \ket{C(j)}=\tau_+(j)\ket{0},\;1\leq j\leq s \}.
\end{equation}
The initial state of the computing device will be of the form
\begin{equation}
\ket{\psi_1} =\ket{ R(1)} \rangle\otimes \ket{C(1)} .
\label{E:init}
\end{equation}
It helps the intuition to think of the initial state $\ket{R(1)}$ of the register as a simultaneous eigenstate of the components of the $\sigma$ spins in selected directions, encoding the initial word (or superposition of words) on which the machine is required to act.\\
The intuition of ``a single clocking excitation traveling along the program line'' emerging from the above considerations is made precise by introducing the observable \emph{position of the excitation}, or \emph{position of the cursor}:
\begin{equation} \label{eq:Q}
Q=I_r \otimes \sum_{j=1}^{s}j \; \frac{1+\tau_3(j)}{2}.
\end{equation}
Let us have a look on the Hamiltonian \eqref{eq:hamfeynman} and on the  interaction of the program line with the register degrees of freedom it describes. At any particular time $t$, if we expand $e^{-i t H}$ out as
\begin{equation}
	e^{-i H t} = 1 -i H t - H^2 t^2 / 2+ \ldots
\end{equation}
we find the operator $H$ operating an innumerable arbitrary number of times and the total state of the system is a superposition of this possibilities.\\
To illustrate the functioning of the clocking mechanism, we will consider only states of the form
\begin{equation}\label{eq:stdstate}
\ket{\psi_k}= \ket{R(k)} \otimes \ket{C(k)} ,\; 1 < k < s,
\end{equation}
where by \ket{R(j)} we indicate the $(j-1)$-th \emph{logical successor} of the initial state \ket{R(1)} of the register, that is
\begin{equation} \label{eq:Rj}
	\ket{R(j)} = U_{j-1} U_{j-2} \ldots U_2 U_1 \ket{R(1)},
\end{equation}
and the state of the cursor is a basis vector of the $N_3=1$ subspace of \Hc.\\
We begin with a look at the action of $H$ on a state of the form \eref{eq:stdstate}, with $1 < k <s$:
\begin{eqnarray}
H \ket{\psi_k} & = & \sum_{x=1}^{s-1}U_x \otimes\tau_+(x+1)\tau_-(x)  \ket{\psi_k}+ U_{x}^{-1} \otimes \tau_+(x)\tau_-(x+1)  \ket{\psi_k} \nonumber \\
& = & U_k \ket{R(k)} \otimes\ket{C(k+1)}   +  U_k^{-1} \ket{R(k)}\otimes\ket{C(k-1)}  = \nonumber \\
& = & \ket{R(k+1)}\otimes\ket{C(k+1)}+   \ket{R(k-1)} \otimes \ket{C(k-1)} =  \nonumber \\
& =& \ket{\psi_{k+1}} + \ket{\psi_{k-1}}
\end{eqnarray}
We intentionally left out the case in which $H$ acts on the states \ket{\psi_1} and \ket{\psi_s}. In those cases, due to the boundary of the system, we have 
\begin{eqnarray}
	H \ket{\psi_1} & = & \ket{\psi_2} \\
	H \ket{\psi_s} & = & \ket{\psi_{s-1}}.
	\label{eq:boundary}
\end{eqnarray}
The coupling of the register with the cursor degrees of freedom seems to be of the following kind: if the excitation moves one step further, the logical state of the computation advances by one; if the excitation moves one step backward, the logical state regresses to the previous one. This property holds also when applying higher powers of the Hamiltonian operator to the system. For example, if we consider $H^2$  and expand it we get
\begin{eqnarray}
	H^2 & = & \sum_{x=1}^{s-1} \sum_{y=1}^{s-1}U_x U_y \otimes\tau_+(x+1)\tau_-(x)\tau_+(y+1)\tau_-(y)  + \nonumber \\
	& + &  U_x U_y^{-1}\otimes\tau_+(x+1)\tau_-(x)\tau_+(y)\tau_-(y+1) \nonumber \\
	& + &  U_x^{-1} U_y\otimes\tau_+(x)\tau_-(x+1)\tau_+(y+1)\tau_-(y) \nonumber \\
	& + & U_x^{-1} U_y^{-1}\otimes \tau_+(x)\tau_-(x+1)\tau_+(y)\tau_-(y+1). 
	\label{eq:h2expansion}
\end{eqnarray}
It is straightforward to see that, due to the commutation rules \eref{eq:ccr} and to the conservation law \eref{eq:n3}, only some of the terms with $x=y$ and $x=y \pm 1$ survive. Thus, it is possible to simplify \eref{eq:h2expansion} getting
\begin{eqnarray}
	H^2 & = &\sum_{x=2}^{s-1} U_x U_{x-1}\otimes \tau_+(x+1)\tau_-(x)\tau_+(x)\tau_-(x-1)  + 	\nonumber \\
	& = & I_r  \otimes\tau_+(x)\tau_-(x+1)\tau_+(x+1)\tau_-(x)+ \nonumber \\
	& = & I_r  \otimes\tau_+(x+1)\tau_-(x)\tau_+(x)\tau_(x+1) + \nonumber \\
	& = &  U_{x-1}^{-1} U_{x-1}^{-1}\otimes\tau_+(x-1)\tau_-(x)\tau_+(x)\tau_-(x+1).
\end{eqnarray}
For example, if  $H$ acts twice on the state $ \ket{C(k)} \otimes \ket{R(k)} $, with $2 < k < s-2$ we get
\begin{eqnarray}
	H^2  \ket{R(k)} \otimes \ket{C(k)} & = & U_{k+2} U_{k+1} \ket{R(k)} \otimes \ket{C(k+2)}+ \nonumber \\
	 & + &  2   \ket{R(k)} \otimes \ket{C(k)}+\nonumber \\
	 & + &  U_{k-1}^{-1} U_{k}^{-1} \ket{R(k)} \otimes \ket{C(k-2)} = \nonumber \\
	 & = & \ket{\psi_{k+2}} + 2 \ket{\psi_k}+  \ket{\psi_{k-2}}.
\end{eqnarray}
Once more: if the position of the excitation is shifted by $j$-positions, $j \in \{0,2\}$ the logical state evolves or regresses accordingly by $j$-steps. This property extends to every power of $H$; in fact it has been shown by Peres \cite{peres85} that, once defined the projection operator on the $k$-th logical state \ket{R(k)}
\begin{equation}
	P_k \ket{R(l)}= \delta_{k,l} \ket{R(k)}.
\end{equation}
that satisfies
\begin{equation}
	P_k = U_k P_{k-1}U_k^{-1}
\end{equation}
and the operator
\begin{equation} \label{eq:p}
	P=\sum_{k=1}^{s}  P_k \otimes \ketbra{C(k)}{C(k)} 
\end{equation}
the following holds
\begin{theorem}
\begin{equation}
[P,H]=0
\end{equation}
\end{theorem}
Before showing the proof, we observe that, given \eref{eq:stdstate} and proposition \ref{th:number}, the Hamiltonian \eref{eq:hamfeynman} is equivalent to
\begin{eqnarray} \label{eq:hfeynman2}
H & = & \sum_{x=1}^{s-1}  U_x \otimes \ketbra{x+1}{x}+  U_x^{-1} \otimes \ketbra{x}{x+1}  =  \nonumber \\
	& = & \sum_{x=1}^{s-1}  U_x \ketbra{x+1}{x}  + U_x^{-1} \ketbra{x}{x+1} 
\end{eqnarray}
where \ket{x} is the eigenstate of the position operator $Q$, defined as in  \eref{eq:Q}, belonging to the eigenvalue $x$, and the tensor product symbol has been dropped to shorten the expressions. Equivalently \eref{eq:p} can be rewritten as
\begin{equation} \label{eq:p1}
	P=\sum_{k=1}^{s} P_k \ketbra{k}{k} .
\end{equation}
\begin{proof}
We compute explicitly $P\;H$ and $H\;P$.
\begin{eqnarray}
	P\; H & = & \sum_{k=1}^s P_k  \ketbra{k}{k} (\sum_{x=1}^{s-1} U_x \ketbra{x+1}{x}  + U_x^{-1}\ketbra{x}{x+1} ) = \nonumber \\
	& = &  \sum_{k=1}^s \sum_{x=1}^{s-1} P_k U_x\ket{k}\braket{k}{x+1}\bra{x}  + P_k U_x^{-1} \ket{k}\braket{k}{x}\bra{x+1}  \nonumber \\
	& = & \sum_{x=1}^{s-1} P_{x+1}U_x\ketbra{x+1}{x} + P_x U_x^{-1}\ketbra{x}{x+1}
\end{eqnarray}
\begin{eqnarray}
	H	\;P  & = &  (\sum_{x=1}^{s-1}U_x \ketbra{x+1}{x}  + U_x^{-1}\ketbra{x}{x+1} ) \sum_{k=1}^s P_k \ketbra{k}{k}   = \nonumber \\
	& = & \sum_{x=1}^{s-1} \sum_{k=1}^s  U_x P_k \ket{x+1}\braket{x}{k}\bra{k}   + U_x^{-1} \ket{x}\braket{x+1}{k}\bra{k}    \nonumber \\
	& = & \sum_{x=1}^{s-1} U_x P_{x}\ketbra{x+1}{x} +U_x^{-1} P_{x-1} \ketbra{x}{x+1} 
\end{eqnarray}
Thus
\begin{eqnarray}
PH-HP & = & (P_{x+1}U_x - U_x P_ x)\ketbra{x+1}{x}+ (P_x U_x^{-1} - U_x^{-1} P_{x+1})\ketbra{x}{x+1} = \nonumber \\
 & = & (U_x P_{x}U_x^{-1} U_x - U_x P_ x)\ketbra{x+1}{x}  +  \nonumber \\ 
 & + & (P_x U_x^{-1} - U_x^{-1} U_x P_{x}U_x^{-1})\ketbra{x}{x+1} = 0
\end{eqnarray}
\end{proof}
The space spanned by the initial state and its logical successors is, therefore, a constant of motion. The set $\left \{ \ket{\psi_k}, 1 \leq k \leq s \right \}$  forms a complete \emph{orthogonal} basis for the $s$-dimensional subspace of \Hm effectively visited during the computation (the orthogonality of different basis vectors following immediately form $\braket{C(j)}{C(k)}=\delta_{j,k}$). It is worth mentioning here that the set of logical successors of the initial state, or \emph{Peres' basis} \index{Peres' basis}, can be algorithmically constructed; in fact, if we split the Hamiltonian  \eref{eq:hfeynman2} into
\begin{eqnarray}
H_{forward} & = & \sum_{x=1}^{s-1}  U_x \ketbra{x+1}{x} \\
H_{backward} & = & \sum_{x=1}^{s-1}  U_x^{-1} \ketbra{x}{x+1};
\end{eqnarray}
the set of logical successors of the initial state \ket{\psi_1} corresponds then to the set
\begin{equation} \label{eq:costruzioneperes}
	\{ \ket{\psi_k} = H_{forward}^{k-1} \ket{\psi_1},\; 1 \leq k \leq s \}
\end{equation}
The Hamiltonian  \eref{eq:hfeynman2} can be rewritten using the Peres basis as
\begin{equation}
	H = \sum_{k=1}^{s-1} \ketbra{\psi_{k+1}}{\psi_k} + \ketbra{\psi_{k}}{\psi_{k+1}}
\end{equation}
or as a $s \times s$ bi-diagonal matrix
\begin{equation} \label{eq:hmatrix}
	H = \left (
\begin{array}[pos]{c c c c c c}
0 & 1 & 0 & \ldots & 0 & 0\\
1 & 0 & 1 & \ldots & 0 & 0\\
0 & 1 & 0 & \ldots & 0 & 0\\
\ldots & \ldots & \ldots & \ldots & \ldots & \ldots\\
0 & 0 & 0 & \ldots &  0& 1\\
0 & 0 & 0 & \ldots & 1 & 0
\end{array} \right ).
\end{equation}
which is, up to constant diagonal terms, the finite difference approximation of the Laplace operator. The evolution of the system is thus  of the form
\begin{equation} \label{eq:psit}
| \psi(t) \rangle= \sum_{k=1}^s c(t,k;s) \;  \ket{\psi_k}  = \sum_{k=1}^s c(t,k;s) \;  \ket{R(k)} \otimes \ket{C(k)},
\end{equation}
$c(t,k;s)$ being a numerical functions of the time $t$, of the label of the logical successor $k$ and parametric with respect to the length of the program line. For the sake of simplicity we postpone the discussion of the $c(t,k;s)$ coefficients to chapter \ref{chap:quantumtiming}.\\
From \eref{eq:psit} it becomes clear how the clocking mechanism works: if a measurement is performed and the cursor is found at position $N$, then the register collapses into in the $N-1$ logical successor of the initial state \ket{R(1)}. In particular, if the cursor is found in the last site, the $s$-th in our notation, the logical state of the register corresponds to the desired output state, that is $A \ket{R(1)}$. In Feynman's words (adapted to our notations), (\ref{eq:psit}) says that, starting from the initial condition (\ref{E:init}), ``\underline{If} at some later time the final site $s$ is found to be in the  $| \tau_3(s)=+1 \rangle$  state (and therefore all the others in $| \tau_3(j)=-1 \rangle$ ), then the register state has been multiplied by $U_{s-1}\cdot \dots \cdot U_2 \cdot U_1$  as desired''.
\section{Continuous time quantum walks}
In this section we give an introductory overview on continuous time quantum walks (as opposed to discrete time, or coined, quantum walks which we will not deal with here). We refer to \cite{farhi97,aharonov00} for an exhaustive treatment of quantum walks and of their algorithmic applications.\\[5pt]
Markov chains or random walks on graphs have proved to be a fundamental tool, with broad applications in various fields of mathematics, computer science and the natural sciences, such as mathematical modeling of physical systems, simulated annealing, and the Markov Chain Monte Carlo
method. In the physical sciences they provide a fundamental model for the emergence of global properties from local interactions. In the algorithmic context, they provide a general paradigm for sampling and exploring an exponentially large set of combinatorial structures (such as matchings in a graph), by using a sequence of simple, local transitions.
It is thus natural to ask whether quantum walks might be useful for quantum computation. In \cite{childs02}, for example, it is shown that there are graphs for which the  time for a \emph{quantum walker} to propagate between a particular pair of nodes is exponentially shorter than the analogous  propagation time needed by a classical walker.\\[5pt]
A continuous time classical random walk on a graph is a Markov process. A graph is an ordered couple $\langle V, E \rangle$, where $V$ is the set vertices, say $\{1, 2,\ldots , s\}$, and $E$ a set of edges between vertices.\\ 
A step in a classical random walk on a graph only occurs between two vertices connected by an edge. Let $\gamma$ denote the jumping rate. Starting at any vertex, the probability
of jumping to any connected vertex in a time $\epsilon$ is $\gamma \epsilon$ (in the limit $\epsilon \to 0$). This random walk can be described by the $s \times s$ infinitesimal generator matrix $M$ defined by
\begin{equation}
	M_{ab} = 
	\begin{cases}
	-\gamma,\mbox{ if} \ a \neq b,\ \mbox{$a$ and $b$ connected by an edge} \\
	0, \mbox{ if} \ a \neq b,\ \mbox{$a$ and $b$ not connected}\\
	k\gamma, \mbox{ if} \ a = b,\ \mbox{$k$ is the valence of vertex $a$}.
	\end{cases}
\end{equation}
If $p_a(t)$ denotes the probability of being at vertex $a$ at time $t$, then it evolves under the master equation
\begin{equation} \label{eq:mastereq}
	\frac{d p_a(t)}{dt} = -\sum_{b} M_{ab} p_b(t).
\end{equation}
Following \cite{farhi98}, a natural quantum analogue to the classical random walk described above is given by the quantum Hamiltonian with matrix elements
\begin{equation}
	\bra{a}H\ket{b}= M_{ab}
\end{equation}
\ket{a} and \ket{b} belonging to an assigned basis $\ket{1},\ket{2},\ldots,\ket{s}$ of a $v$-dimensional Hilbert space. The Schr\"odinger equation for $\ket{\psi(t)}$ can be written as
\begin{equation} \label{eq:schroqw}
	i \frac{d}{dt} \braket{a}{\psi(t)} = \sum_b \bra{a}H\ket{b}\braket{b}{\psi(t)}.
\end{equation}
We observe that, in some sense, \emph{any} evolution in a finite-dimensional Hilbert space can be thought of as an oriented graph with Hermitian weights.\\
We furthermore point out that whereas \eref{eq:mastereq} conserves the probability
\begin{equation}
	\sum_a p_a(t) =1,
\end{equation}
the Schr\"odinger equation \eref{eq:schroqw} preserves probability as the sum of the amplitudes squared
\begin{equation}
	\sum_a |\braket{a}{\psi(t)}|^2.
\end{equation}
The simplest graph we can take into account is the one-dimensional lattice $\mathbb{Z}$, resulting in a nearest neighbor Hamiltonian defined by
\begin{equation} \label{eq:hambase}
	H\ket{j} = -\frac{1}{\Delta^2} (\ket{j+1}-2\ket{j}+\ket{j-1}),
\end{equation}
that is the discrete approximation of the Laplace operator $d^2/dx^2$, with $\Delta=\sqrt{2/\lambda}$, $\lambda=2 \gamma$.\\
In the basis $\ket{1},\ket{2},\ldots,\ket{v}$, the matrix representation of the Hamiltonian specified in \eref{eq:hambase} is 
\begin{equation} \label{eq:hbasematrix}
	H = -\frac{1}{\Delta^2} \left (
\begin{array}[pos]{c c c c c c}
-2 & 1 & 0 & \ldots & 0 & 0\\
1 & -2 & 1 & \ldots & 0 & 0\\
0 & 1 & -2 & \ldots & 0 & 0\\
\ldots & \ldots & \ldots & \ldots & \ldots & \ldots\\
0 & 0 & 0 & \ldots &  -2& 1\\
0 & 0 & 0 & \ldots & 1 & -2
\end{array} \right ),
\end{equation}
which is, up to an additive constant, equivalent to the Hamiltonian \eref{eq:hmatrix} of the basic model of the Feynman machine. We point out that in \eref{eq:hmatrix} the basis states were the states $\ket{\psi_k}$, with $1 \leq k \leq s$, defined as in \eref{eq:costruzioneperes}, namely the Peres basis.\\
We delay to chapter 3 the discussion of the dynamics of an excitation traveling through the linear chain or, more generally, on graphs of the form of  \fref{fig:grovermu3}.

%
\section*{Summary}
We have presented the clocking mechanism of Feynman's quantum computer: the numerical coupling constants between nearest neighbor sites of a spin chain with $XY$ interactions are substituted by unitary operators acting on the register subspace. The so created entanglement between the clocking and the register degrees of freedom implements the same kind of timing used in classical computers to manage the ordered application of computational primitives to the input/output register.\\
We have made evident by means of the Peres basis, that the evolution of the clocking subsystem is independent of the operations performed on the register, as long as we use unitary operators on the latter.\\
The continuous time quantum walks paradigm has been deeply investigated in recent years, with the hope of getting some new technique to define new quantum algorithm working faster than any classical one. The kind of interaction used in Feynman quantum computer extends the continuous time quantum walk: the walker \emph{does something} while traversing the graph.\\[5pt]
In the following chapter we will extend the linear chain to more general graphs. This will allow the implementation of  basic flow control mechanisms, such as the IF...THEN...ELSE and the iteration of quantum subroutines, and a simplification of the interactions between elements of the system.
\chapter{Simplifying the implementation} \label{chap:grover}
As a case study, we present an implementation of Grover's algorithm in the framework of Feynman's cursor model of a quantum computer. Using the cursor degrees of freedom  as a quantum clocking mechanism  allows Grover's algorithm to be performed using a single, time independent Hamiltonian. We examine issues of locality and resource usage in implementing such a Hamiltonian. In the familiar language of Heisenberg spin-spin coupling on a linear chain of spins, we introduce occasional controlled jumps that allow for motion on a planar graph: in this sense our model implements the idea of ``timing'' a quantum algorithm using a continuous-time quantum walk. In this context we examine some consequences of the entanglement between the states of the input/output register and the states of the quantum clock.
\section{Grover's algorithm on a Feynman machine}
The starting point of our discussion is the analysis of the physical aspects of Grover's algorithm given in  \cite{grover01} and \cite{farhi98}.\\ 
Suppose one is given an ``oracle'' able to compute, in a quantum reversible way, the indicator function of a binary word $\mathbf{a} \in \{-1,1\}^\mu$  of an assigned length  $\mu$. We will assume, for the sake of definiteness, that this computation is performed by applying a unitary transformation $A$ to the input/output register of length $\nu = \mu + 1$. Suppose that $A$ results from the action, for a fixed amount $\bar{t}$ of time, of a Hamiltonian  $K(\mathbf{a})$, that is:
\begin{equation}
	A = \exp(-i \bar{t} K(\mathbf{a})).
\end{equation}
It is then possible to arrange things in such a way that the state
\begin{equation}
	\ket{\mathbf{a}} = \ket{\sigma_3(1)=a_1, \sigma_3(2) = a_2,\ldots,\sigma_3(\mu)=a_\mu}
\end{equation}
that corresponds to having the word  $\mathbf{a}$ written on the register, is the ground state of $K(\mathbf{a})$.\\
The search for the ground state of $K(\mathbf{a})$  is performed, in Reference \cite{farhi98}, following the simple idea of perturbing the Hamiltonian  $K(\mathbf{a})$,
\begin{equation} \label{eq:perturbedh}
	K(\mathbf{a}) \rightarrow K(\mathbf{a})+\beta
\end{equation}
with a perturbation $\beta$  chosen in such a way that a suitable initial condition oscillates about the state \ket{\mathbf{a}}  with a period proportional to $2^{\mu/2}$ , becoming, at a time $O(2^{\mu/2})$, parallel to the target state.\\
By applying Trotter's product formula to \eref{eq:perturbedh}, it is shown, in Reference \cite{grover01}, that no significant loss in the probability  $P(\mathbf{a})$ of finding  $\mathbf{a}$, at suitable values of time  $t$, results from alternating intervals of time in which only the ``oracle'' Hamiltonian $K(\mathbf{a})$  is active, thus in fact applying the ``oracle'' transformation $A$, with intervals in which only $\beta$ is active, thus in fact applying the ``estimator'' transformation
\begin{equation}
	B = \exp(-i \bar{t} \beta)
\end{equation}
The \emph{oscillatory} nature of the quantum search algorithm is confirmed, in this discrete time setting, by the analysis of Reference \cite{boyer98}.\\
We have seen that the administration, in the correct order, of the \emph{oracle} and \emph{estimation} transformation $A$ and $B$ to the register can be realized by means of the clocking mechanism. In our implementation we define the operator $A: \Hr \rightarrow \Hr$ through the action on simultaneous eigenstates of $\sigma_3(1),\sigma_3(2),\ldots,\sigma_3(\mu),\sigma_3(\nu)$:
\begin{eqnarray} \label{eq:oracle}
 & A  & \ket{\sigma_3(1)=z_1,\sigma_3(2)=z_2,\ldots,\sigma_3(\mu)=z_\mu,\sigma_3(\nu)=z_\nu} =\nonumber \\
& = & \begin{cases}
\ket{\sigma_3(1)=z_1,\sigma_3(2)=z_2,\ldots,\sigma_3(\mu)=z_\mu,\sigma_3(\nu)= - z_\nu},\ if\ \mathbf{z}= \mathbf{a} \\
\ket{\sigma_3(1)=z_1,\sigma_3(2)=z_2,\ldots,\sigma_3(\mu)=z_\mu,\sigma_3(\nu)= z_\nu},\ if\ \mathbf{z} \neq \mathbf{a}
\end{cases}
\end{eqnarray}
where $\mathbf{z}=(\ves{z}{\mu}) \in \{-1,1\}^\mu$, $z_\nu \in \{-1,1\}$.\\
The \emph{oracle} $A$ performs a quantum reversible computation of the indicator function of $\mathbf{a}$ by flipping the  component $3$ of the output qubit  $\underline{\sigma}(\nu)$ iff the word $\mathbf{a}$ is written on the register in terms of the $\underline{\sigma}(1), \underline{\sigma}(2),\ldots,\underline{\sigma}(\mu)$ components of the input qubits.\\
Define, in a similar way, a linear operator $B:\Hr \rightarrow \Hr$ through the following action on the simultaneous eigenstates of  $\sigma_1(1),\sigma_1(2),\ldots,\sigma_1(\mu),\sigma_3(\nu)$: 
\begin{eqnarray} \label{eq:estimator}
 & B  & \ket{\sigma_1(1)=x_1,\sigma_1(2)=x_2,\ldots,\sigma_1(\mu)=x_\mu,\sigma_3(\nu)=z_\nu} =\nonumber \\
& = & \begin{cases}
\ket{\sigma_1(1)=x_1,\sigma_1(2)=x_2,\ldots,\sigma_1(\mu)=x_\mu,\sigma_3(\nu)= -z_\nu},\ if\ \mathbf{x}= \mathbf{1}_\mu \\
\ket{\sigma_1(1)=x_1,\sigma_1(2)=x_2,\ldots,\sigma_1(\mu)=x_\mu,\sigma_3(\nu)= z_\nu},\ if\ \mathbf{x}\neq \mathbf{1}_\mu
\end{cases}
\end{eqnarray}
where $\mathbf{x}=(\ves{x}{\mu}) \in \{-1,1\}^\mu$, $z_\nu \in \{-1,1\}$ and $\mathbf{1}_\mu = (\underbrace{1,1,\ldots,1)}_{\mu \ times}$.
\begin{figure}[!h]
	\centering	\includegraphics[width=10cm]{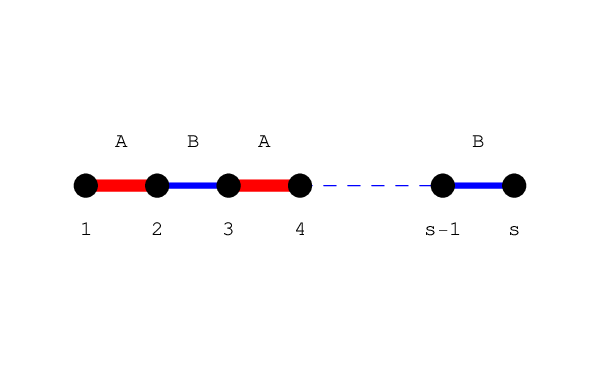}
	\caption{The Hamiltonian \eref{eq:hgrover} describes an $XY$ interaction between nearest neighbor cursor spins.  $A$ is the ``coupling constant'' between spins corresponding to odd links; $B$ is the ``coupling constant'' between spins corresponding to even links. Both $A$ and $B$ are, in fact, functions of the register spins. 
}
	\label{fig:CursoreAB}
\end{figure}
Following the prescription of \sref{sec:themodel} we define the Hamiltonian
\begin{equation} \label{eq:hgrover}
	H = \sum_{x=1}^{s-1} U_j \ \ketbra{j+1}{j} + h.c.
\end{equation}
where
\begin{equation}
	U_j=
	\begin{cases}
	A \ if \ j \ is \ odd \\
	B \ if \ j \ is \ even.
	\end{cases}
\end{equation}
We set the initial state of the machine to
\begin{eqnarray} \label{eq:condinigrover1}
	\ket{\psi_1} & = & \ket{R_G(1)} \ket{C(1)} = \nonumber \\
	& = & \ket{\sigma_1(1)=1,\sigma_1(2)=1,\ldots,\sigma_1(\mu)=1,\sigma_3(\nu)= -1} \ket{Q = 1},
\end{eqnarray}
where $C(1)$ is defied by \eref{eq:cj} and $\ket{Q=1}$ is the eigenstate of the position operator $Q$ defined in \eref{eq:Q} belonging to the eigenvalue $1$.\\
Following \eref{eq:costruzioneperes} we define the basis states $\ket{\psi_j},\;1 \leq j \leq s$ of the subspace of \Hm the system evolves in. We recall that the basis states are of the form
\begin{equation}
	\ket{\psi_k} = U_{k-1}\cdot \ldots \cdot U_1 \ket{R_G(1)}\ket{Q=k};
\end{equation}
if $k$ is an \emph{odd} number, $k=2n+1$, it is
\begin{equation}
	U_{k-1}\ldots U_1 = (B\cdot A)^n.
\end{equation}
An explicit expression for $(B\cdot A)^n \; \ket{\sigma_1(1)=x_1,\sigma_1(2)=x_2,\ldots,\sigma_1(\mu)=x_\mu}$ can be found by the iterative procedure of Reference \cite{boyer98}:
\begin{eqnarray} \label{eq:statogroverdispari}
&(B\cdot A)^n   \ket{\sigma_1(1)=x_1,\sigma_1(2)=x_2,\ldots,\sigma_1(\mu)=x_\mu} &= \nonumber \\
& = 
	\left (  
	\alpha_n(\mu)\ket{\mathbf{a}}_3 + \beta_n(\mu) \ \sum_{\mathbf{z} \neq \mathbf{a}} \ket{\mathbf{z}_3}
	\right)&
\end{eqnarray}
where
\begin{eqnarray} \label{eq:alphabeta}
\alpha_n(\mu) & = & (-1)^n \sin((2n+1)\chi(\mu)) \\ 
\chi_\mu & = & \arcsin(2^{-\mu/2}) \\ 
\beta_n(\mu) & = & \frac{(-1)^n}{\sqrt{2^\mu -1}} \cos((2n+1)\chi(\mu))
\end{eqnarray}
In \eref{eq:statogroverdispari} we have omitted explicit reference to \ket{\sigma_1(\nu)=-1}, as the conservation law
\begin{equation}
	[H,\sigma_1(\nu)]=0.
\end{equation}
allows us to do, and we have set, for every $\mathbf{z} =(\ve{z}{\mu}) \in \{-1,1\}^\mu$
\begin{equation}
	\ket{\mathbf{z}}_3 = \ket{\sigma_3(1)=z_1,\sigma_3(2)=z_2,\ldots,\sigma_3(\mu)=z_\mu}.
\end{equation}
The case of an even value of $k = 2 n + 2$, can be similarly analyzed, by observing that
\begin{eqnarray} \label{eq:statogroverpari}
&A\cdot (B\cdot A)^n   \ket{\sigma_1(1)=x_1,\sigma_1(2)=x_2,\ldots,\sigma_1(\mu)=x_\mu} &= \nonumber \\
& = 
	\left (  
	-\alpha_n(\mu)\ket{\mathbf{a}}_3 + \beta_n(\mu) \ \sum_{\mathbf{z} \neq \mathbf{a}} \ket{\mathbf{z}_3}
	\right)&
\end{eqnarray}
\begin{observation}
The state of the register evolves in a proper two dimensional subspace of \Hr. The evolution of Grover's register state can be simulated using a single qubit rotating in a two dimensional Hilbert space $\mathcal{H}$. We will exploit this feature later.
\end{observation}
Summarizing, and considering, for the sake of definiteness, the case of an odd value of  $s = 2 g  + 1  $ the solution of the Schr\"odinger equation
\begin{equation}
	i \frac{d}{dt} \ket{\psi(t)} = H \ket{\psi(t)},
\end{equation}
under the initial condition \eref{eq:condinigrover1} can be written as:
\begin{eqnarray}
\ket{\psi(t)} & = & \sum_{n=0}^g
 c(t,2n+1;s)\left (  
	\alpha_n(\mu)\ket{\mathbf{a}}_3 + \beta_n(\mu) \ \sum_{\mathbf{z} \neq \mathbf{a}} \ket{\mathbf{z}_3}
	\right) \otimes  \nonumber \\
	& \otimes & \ket{\sigma_1(\nu)=-1} \otimes \ket{Q=2n+1} + \nonumber \\
	& + & 
	\sum_{n=0}^{g-1}
 c(t,2n+2;s) \left (  
	-\alpha_n(\mu)\ket{\mathbf{a}}_3 + \beta_n(\mu) \ \sum_{\mathbf{z} \neq \mathbf{a}} \ket{\mathbf{z}_3}
	\right) \otimes  \nonumber \\
	& \otimes & \ket{\sigma_1(\nu)=-1}\otimes \ket{Q=2n+2} 
	\end{eqnarray}
The conditional probability of reading upon measurement the word $\mathbf{a}$ on the register \emph{given} that the cursor has been found at site $k$ is then given by $|\alpha_k(\mu)|^2$.
\section{Iteration of quantum subroutines}
Grover's algorithm provides a quadratic speedup with respect to any classical algorithm for the search of a keyword in an unstructured database of $2^\mu$ words \cite{grover01}: it is sufficient to apply the $B\cdot A$ operator $O(2^{\mu/2})$ times to the register to have the word $\mathbf{a}$ we are looking for written on it. Thus, if we sticked to the linear chain model, we would need an exponential number of functional blocks $A$ and $B$ and program line sites to perform the algorithm.\\
In this section we show that this cost in terms of space can be made linear in $\mu$ by using quantum subroutines.\\
For every non negative integer $K$, we wish to show a quantum clocking mechanism able to apply  $2^K$ times the transformation $BA$  to the register qubits $\ve{\underline{\sigma}}{\nu}$, by repeatedly using the same ``piece of hardware'' that applies $BA$ just once. We will show that this clocking mechanism will involve
\begin{equation}
	s(K)=4K+3
\end{equation}
cursor qubits $\ve{\underline{\tau}}{s(K)}$.\\
In order to keep track of the progress of the  $2^K$ executions of the assigned \emph{subroutine} $BA$, there must be a subsystem (the \emph{subroutine counter}) having  $2^K$ different states: it will be constructed in terms of  $K$ qubits \ve{\underline{\rho}}{K}, $\underline{\rho}(j)=(\rho_1(j),\rho_2(j),\rho_3(j))$, \mbox{$\;1 \leq j \leq K$}.\\
We will denote with $\mathcal{H}_{counter}$ the $2^K$ dimensional state space of the counter degrees of freedom.\\
The definition of the Hamiltonian operator on $\Hr \otimes \mathcal{H}_{counter}\otimes \Hc$ will be given by an iterative scheme.
\begin{figure}[h]
	\centering	\includegraphics[width=10cm]{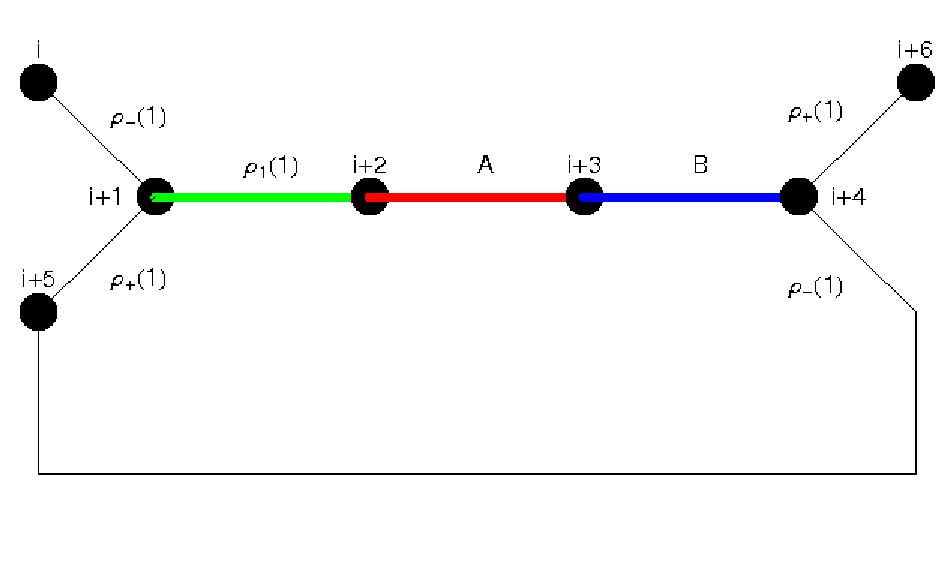}
		\caption{Do $BA$ twice.}
	\label{fig:SubroutinesAB}
\end{figure}
\begin{figure}[h]
	\centering	\includegraphics[width=10cm]{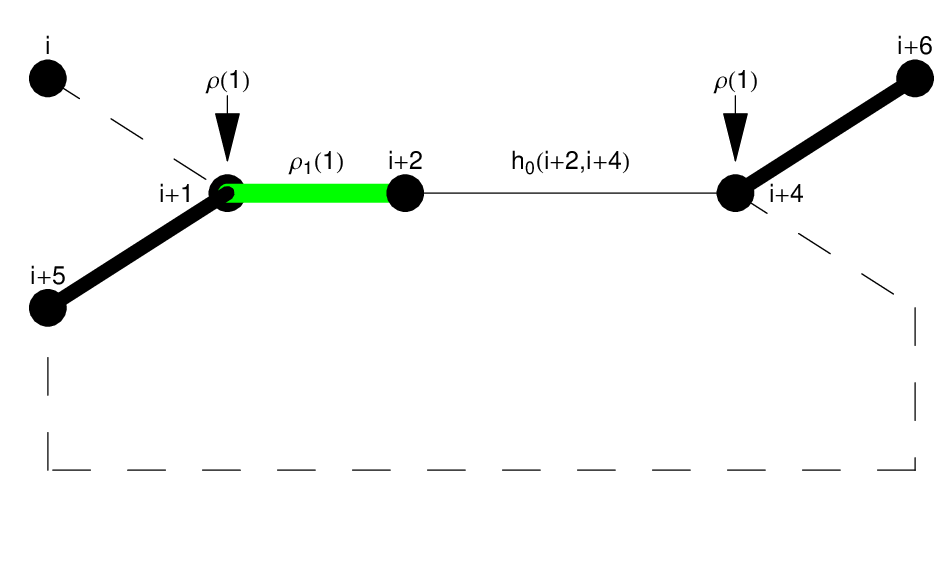}
		\caption{Do $BA$ twice. Here the operators acting on the counter degrees of freedom are encoded in the following graphical convention: to a raising operator corresponds a thick solid black line; to a lowering operator a dashed line; to the negation $\rho_1$ a thick solid green line. The functional block $h_0(i+2,i+4)$ is indicated as a coupling constant on the corresponding edge.}
	\label{fig:Subroutines}
\end{figure}
Set, for $i=1,2,\ldots, s(k)-2$
\begin{eqnarray}
	h_0(i,i+2) & = & \left ( A_{forward}(i)+B_{forward}(i+1) \right )\otimes I_{counter} \nonumber \\
	& = & (A \ketbra{i+1}{i} + B \ketbra{i+2}{i+1}) \otimes I_{counter}
\end{eqnarray}
where $I_{counter}$ is the identity operator in $H_{counter}$. The operator $h_0(i,i+2)$ applies the transformation $BA$ to the register while the cursor jumps from site $i$ to site \mbox{$i+2$} (see \fref{fig:SubroutinesAB}).\\
For $i=1,2,\ldots, s(K)-6$
\begin{eqnarray} \label{eq:hsubroutine}
h_1(i,i+6) & = & \rho_+(1) \ketbra{i+1}{i} +\nonumber \\
 & + &\rho_x(1)  \ketbra{i+2}{i+1}+  \nonumber \\
 & + & h_0(i+2,i+4) +\nonumber \\
 & + & \rho_-(1) \ketbra{i+6}{i+4} + \nonumber \\ 
 & + &  \rho_+(1) \ketbra{i+5}{i+4} +\nonumber \\ 
 & + & \rho_-(1)\ketbra{i+1}{i+5}
\end{eqnarray}
The term $  \rho_-(1) \ketbra{i+6}{i+4} +  \rho_+(1) \ketbra{i+5}{i+4}$  in \eref{eq:hsubroutine} is an example of the implementation of a conditional jump through the \emph{SWITCH} primitive. The first addendum acts non-vanishingly only in the subspace belonging to the eigenvalue $+1$ of the controlling qubit  $\rho_3(1)$  and sends the excitation of the cursor from  $i+4$ to $i+6$; the second addendum, in turn, acts non-vanishingly only in the subspace belonging to the eigenvalue $-1$ of  $\rho_3(1)$  and sends the excitation of the cursor from $i+4$  to $i+5$. Notice that in this implementation of the \emph{IF...THEN, ELSE...} construct, the controlling bit  $\rho_3(1)$ gets inverted.\\
The iteration step from  $h_{j-1}$ to $h_j$ is given by: 
\begin{eqnarray} \label{eq:hsubroutine2}
h_j(i,i+ 4j+2) & = & \rho_+(j) \ketbra{i+1}{i}+\nonumber \\
& + & \rho_x(j) \ketbra{i+2}{i+1}+  \nonumber \\
 & + & h_j(i+2,i+4 j)\nonumber \\ 
 & + & \rho_-(j) \ketbra{i+4j+2}{i+4 j} + \nonumber \\ 
 & + & \rho_+(j) \ketbra{i+4j+1}{i+4 j} +  \nonumber \\ 
 & + & \rho_-(j) \ketbra{i+1}{i+ 4j +1}
\end{eqnarray}
and is represented in \fref{fig:iteratore2mu}.
\begin{figure}[h]
	\centering
		\includegraphics[width=10cm]{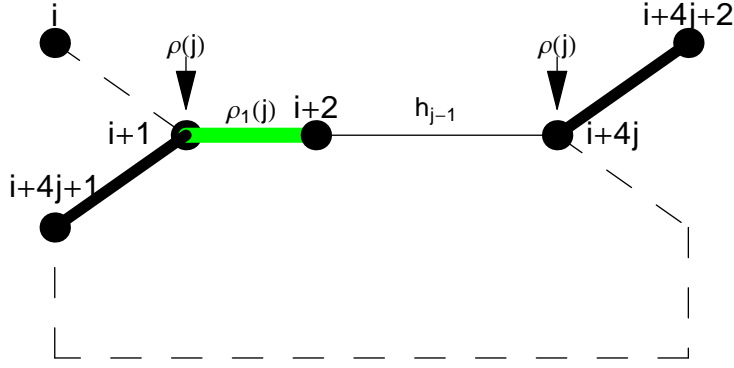}
	\caption{Do $BA$ $2^j$ times while the cursor moves from $i$ to $i+s(j)-1$}
	\label{fig:iteratore2mu}
\end{figure}
For a fixed value of the positive integer $K$ we define the forward part of the Hamiltonian as
\begin{equation}
	H_{forward}(K) = h_K(1,s(K)) = h_K(1,4K+3)
\end{equation}
and the Hamiltonian as
\begin{equation} \label{eq:hgroveriter}
	H(K) = 	H_{forward}(K)+	H_{backward}(K)= 	H_{forward}(K)+	H_{forward}(K)^\dagger.
\end{equation}
Let us consider the Schr\"odinger equation
\begin{equation} \label{eq:schroiter}
	i \frac{d}{dt} \ket{\psi(t)} = -\frac{\lambda}{2} H(K) \ket{\psi(t)}
\end{equation}
and the initial condition
\begin{eqnarray} \label{eq:inititer}
\ket{\psi(0)} & = & \ket{\sigma_1(1)=1,\sigma_1(2)=1,\ldots,\sigma_1(\mu)=1,\sigma_1(\nu)=-1} \otimes \nonumber \\
& \otimes & \ket{\rho_3(1)=-1,\rho_3(2)=-1,\ldots,\rho_3(K)=-1} \otimes \nonumber \\
& \otimes & \ket{Q=1}.
\end{eqnarray}
Equation \eref{eq:schroiter} under initial condition \eref{eq:inititer} is extremely easy to solve because of the conservation laws
\begin{eqnarray}
\left[ H(K),\sigma_1(\nu) \right] & = &0  \label{eq:consnu}\\
\left [H(K),N_3 \right ] & = & 0 \\
\left [H(K),P(K) \right ] & = & 0
\end{eqnarray}
where the operator $P(K)$ is the projector operator on the subspace of $\Hr \otimes \mathcal{H}_{counter} \otimes \Hc$ spanned by the $2^{K+3}-5$ orthonormal vectors defined, similarly to \eref{eq:costruzioneperes}, by
\begin{equation} \label{eq:soliter}
		\{ \ket{\psi_k} = H_{forward}^{k-1} \ket{\psi_1},\; 1 \leq k \leq p(K) \}
\end{equation}
with $\ket{\psi_1}=\ket{\psi(0)}$.
Because of the above considerations, the solution of \eref{eq:schroiter}, \eref{eq:inititer} will be of the form
\begin{equation} \label{eq:psit2}
	\ket{\psi(t)} = \sum_{j=1}^{p(K)} c(t,j;p(K)) \ket{\psi_j}
\end{equation}
A full understanding of the solution \eref{eq:soliter} requires the analysis of the states  \ket{\psi_j}, for $j=1,2,\ldots,p(K)$.\\
Because of \eref{eq:consnu} all of them are eigenstates of $\sigma_1(\nu)$ belonging to the eigenvalue $-1$; from now on, we omit the explicit reference to this fact, using the shorthand notation
\begin{equation}
	\ket{\sigma_1(1)=x_1,\ldots,\sigma_1(\mu)=x_\mu,\sigma_1(\nu)=-1} \equiv \ket{\sigma_1(1)=x_1,\ldots,\sigma_1(\mu)=x_\mu} = \ket{\mathbf{x_\mu}}
\end{equation}
for $\mathbf{x}_\mu \in \{-1,1\}^\mu$.\\
Each of the vectors  \ket{\psi_j} will be, furthermore, a simultaneous eigenvector of each of the operators $\mathbf{\rho}_3 = ( \ve{\rho_3}{K})$. Calling  $\mathbf{r_j} \in \{-1,1\}^K$ the collection of the eigenvalues to which  \ket{\psi_j} belongs,  we will write, for istance,
\begin{equation}
	\ket{\psi_1}=\ket{\mathbf{1}_1} \otimes \ket{\mathbf{\rho}_3=-\mathbf{1}} \otimes \ket{Q=1}.
\end{equation}
and
\begin{equation} \label{eq:succiter2}
	\ket{\psi_j} = A^{\epsilon_j} (BA)^{n_j} \ket{\mathbf{1}_1} \otimes \ket{\mathbf{\rho}_3=\mathbf{r}_j} \otimes \ket{Q=q_j}.
\end{equation}
The explicit iterative algorithm by which $\epsilon_j,\ n_j,\ q_j,\ \mathbf{r}_j$ can be computed is strictly parallel to the iteration procedure of \fref{fig:iteratore2mu}. For the following discussion it is sufficient to define the exponents $\epsilon_j$ and $n_j$.
\begin{equation}
	\epsilon_j = 
	\begin{cases}
	1 \;\;\;  if\ j \in \{\ves{j}{2^K}\} \\
	0 \;\;\; otherwise
	\end{cases}
\end{equation}
where, for $i=1,2,\ldots,2^K$
\begin{eqnarray} \label{eq:conteggio}
j_i & = & 2 K + 2 +5 (i-1) + 3 \sum_{x=1}^{i-1} e_2(x) = \nonumber \\
& = &  2 K + 2 +5 (i-1) + 3 \sum_{h=1}^{K-1} \left\lfloor (i-1) / 2^{K-h} \right \rfloor.
\end{eqnarray}
In \eref{eq:conteggio} we have indicated by $e_2(x)$  the exponent of the prime factor $2$ in the factorization of the positive integer $x$, and by  $\left\lfloor y \right \rfloor$  the integer part of the positive real number $y$.\\
Let us focus our attention on the states $\ket{\psi_{j_1},\ket{\psi}_{j_2},\ldots, \ket{\psi_{j_{2^k}}}}$,
\begin{eqnarray} \label{eq:sub1}
	\ket{\psi_{j_1}} & =&  A \ket{\mathbf{1}_1}  \ket{\mathbf{\rho}_3=\mathbf{1}_1}  \ket{Q=2 K + 2} = \nonumber \\
	& = &  \ket{\mathbf{1}_1} \ket{N_{\mathbf{\rho}_3}=1} \ket{Q=j_1}.
\end{eqnarray}
In \eref{eq:sub1} we have given a numerical meaning to the content of a subroutine counter by defining the operator
\begin{equation}
	N_{\mathbf{\rho}_3}= 1 + \sum_{y=1}^K \frac{1+\rho_3(y)}{2} 2^{y-1}.
\end{equation}
In all the predecessors $\ket{\psi_1},\ket{\psi_2},\ldots,\ket{\psi_{j_1-1}},\ket{\psi_{j_1}}$ the register is in its initial state $\ket{\mathbf{1}}_1$. The immediate successor of \ket{\psi_{j_1}} is
\begin{equation}
	\ket{\psi_{j_1+1}} = BA \ket{\mathbf{1}}_1  \ket{N_{\mathbf{\rho}_3}=1} \ket{Q=2 K + 3}.
\end{equation}
In all of the states  $\ket{\psi_{j_1+1}},\ldots,\ket{\psi_{j_2-1}}$  the register remains in the state $BA \ket{\mathbf{1}}_1$ ; the content of the register changes only at step  $j_2$, where it is
\begin{equation}
	\ket{\psi_{j_2}} = ABA \ket{\mathbf{1}}_1  \ket{N_{\mathbf{\rho}_3}=2} \ket{Q=2 K + 7}.
\end{equation}
At each of the steps  $j_i$ the state of the register gets acted upon by an additional  $A$ and at step $j_i+1$  by an additional  $B$. In steps from  $j_i+2$ to $j_{i+1}-1$  the state of the register remains unaltered.\\
The content of the register becomes $(BA)^{2^K}$ for the first time at step $j_{2^K}+1=p(K)-K$  and such remains until  the last step $p(K)$.\\
The exponent $n_j$  in \eref{eq:succiter2} is therefore equal to the number of ``non-trivial'' steps  $j_i$  that precede step $j$:
\begin{equation}
	n_j  = \left | \{1 \leq i \leq 2^K : j_i  < j  \}\right|.
\end{equation} It is, therefore
\begin{eqnarray}
n_j & = & 0,  \mbox{ for } j <2K +2 \nonumber \\
n_{2K+3} & = & 1, \nonumber \\
n_j &  = & 2^K, \mbox{ for } j \geq p(K)-K.
\end{eqnarray}
For $2K+3 \leq j \leq p(K)-K$, $n_j$ grows in an approximately linear way because of the inequality
\begin{eqnarray}
	j_i & \geq & 2K+2+5(i-1)+3\left(1-\frac{1}{2^{L(i)}} \right)(i-1)-3L(i) \\
	j_i & \leq  & 2K+2+5(i-1)+3\left(1-\frac{1}{2^{L(i)}} \right)(i-1)\\
	L(i)& = & \left\lfloor \log_2(i-1) \right\rfloor
\end{eqnarray}
which easily follows from \eref{eq:conteggio} and from the fact that $x-1 < \left\lfloor x \right\rfloor \leq x$.\\
This justifies the approximation
\begin{equation}
n_j \approx 
\begin{cases}
0,  & \mbox{for } 1 \leq j \leq 2K+2 \\
1+\frac{2^K-1}{p(K)-3K-3}(j-(2K+3)), & \mbox{for }   2K+3 \leq j \leq p(K)-K \\
2^K, & \mbox{for }  p(K)-K \leq j \leq p(K).
\end{cases}
\end{equation}
The iteration of quantum subroutines mechanism provides thus a way to reduce the space complexity of an algorithm; this, in turn, reduces the complexity of the physical implementation of the Grover algorithm on a Feynman computer. In the next section we address the problem of further simplify the implementation of quantum functional blocks.
\section{Equivalent `local' Hamiltonian } \label{sec:localhamiltonian}
We  restrict the class of unitary operators $U_j$ acting on the register to rotation operators acting on a single spin. In this way  each addendum in the Hamiltonian involves at most 3 bodies: two spins of the program counter and one of the register. The fulfillment of this requirement makes the architecture of the quantum computer \emph{modular} \cite{peres85}, thus simplifying the physical implementation of the desired interactions. We will show  that this requirement does not affect the computational power of the Feynman machine. Furthermore, we will provide an analysis of the computational cost, in terms of space (additional qubits) involved in substituting such non-local terms with equivalent terms in which only interactions between two cursor spins and at most one register spin appear.\\[5pt]
The explicit expression, in terms of the register spins, of the \emph{oracle} operator $A$ defined in \eref{eq:oracle} is:
\begin{equation}
	A = 1 + (\sigma_1(\nu)-1) \prod_{i=1}^\mu \frac{1+a_i \sigma_3(i)}{2}.
\end{equation}
The analogous expression for the \emph{estimator} operator $B$ defined in \eref{eq:estimator} is
\begin{equation}
	B = 1 + (\sigma_1(\nu)-1) \prod_{i=1}^\mu \frac{1+a_i \sigma_1(i)}{2}.
\end{equation}
where $1=I_r$ is the identity operator in $\Hr$.\\
In the Hamiltonian $H$ defined in \eref{eq:hgrover} and in the Hamiltonian $H(K)$ \eref{eq:hgroveriter}, there are, therefore, non local terms such as $A \ketbra{j+1}{j}$ and $B \ketbra{j+1}{j}$ involving many-body interactions among two cursor spins and \emph{all} the register spins. Our goal, in this section, is to show that both the transformation $A$ and $B$ can be implemented by means of local operations involving at most three bodies: two degrees of freedom of the clock and one of the register.\\[5pt]
For the sake of definiteness we concentrate our attention, to start with, on the clocked implementation of the $C-NOT$ logical operator
\begin{equation}
	CNOT(j,j+1)=C-NOT\otimes \ketbra{j+1}{j}+ h.c.
\end{equation}
where
\begin{equation}
	C-NOT = 1+(\sigma_1(\nu)-1)\prod_{i=1}^\mu \frac{1+\sigma_3(i)}{2}.
\end{equation}
The C-NOT is a logical \emph{reversible} binary operator whose truth table is given in table \ref{tab:truthCNOT}\footnote{The C-NOT can also be seen as a \emph{reversible} version of the binary sum.}.
\begin{table*}[h]
	\centering
		\begin{tabular}{||c | c || c | c||} \hline
			$\sigma_3(1)$ & $\sigma_3(\nu)$ &$\sigma_3(1)'$ & $\sigma_3(\nu)'$ \\ \hline \hline
			-1 & -1 &-1 & -1  \\ \hline
			-1 & +1 &-1 & +1  \\ \hline
			+1 & -1 &+1 & +1  \\ \hline
			+1 & +1 &+1 & -1  \\ \hline\hline
		\end{tabular}
	\caption{The truth table of the C-NOT; $\sigma_3(1)$ and $\sigma_3(\nu)$ are the input values whereas $\sigma_3(1)'$ and $\sigma_3(\nu)'$ are the output values.}
	\label{tab:truthCNOT}
\end{table*}\\
The action of $CNOT(j,j+1)$ on the computational basis is ``Flip the  component $3$ of the $\nu$-th qubit iff  $\sigma_3(1)$  input qubit points in the  $+1$ direction, starting with the cursor in position $j$".
\begin{figure}[h!]
	\centering	\includegraphics[width=10cm]{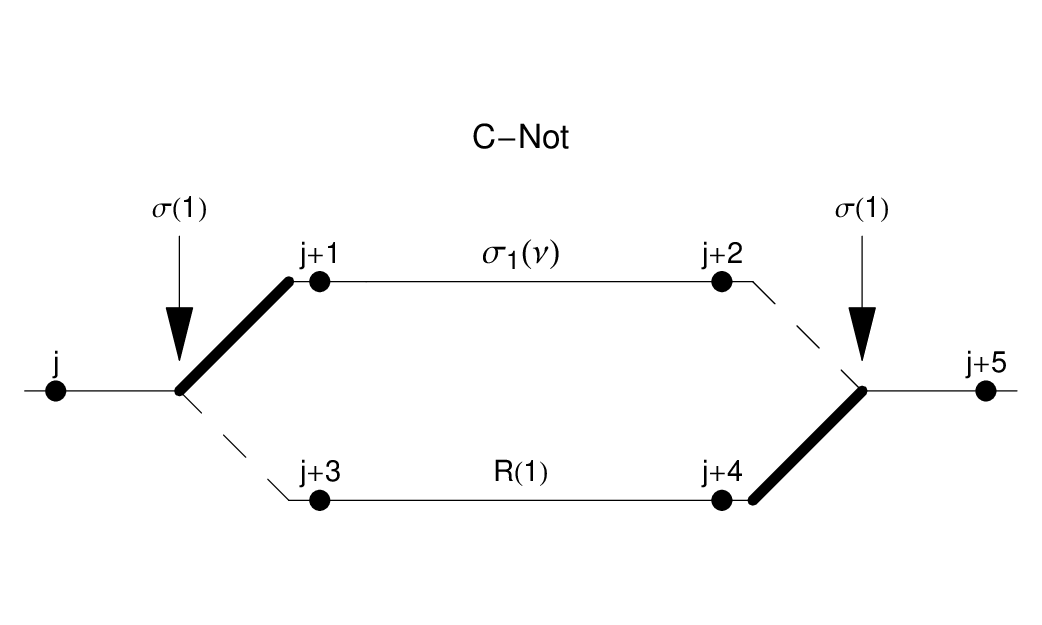}
	\caption{$C^1NOT(j,j+5)$; we have used the graphical conventions of \fref{fig:Subroutines}.}
	\label{fig:Switch}
\end{figure}
\begin{figure}[!h]
	\centering	\includegraphics[width=10cm]{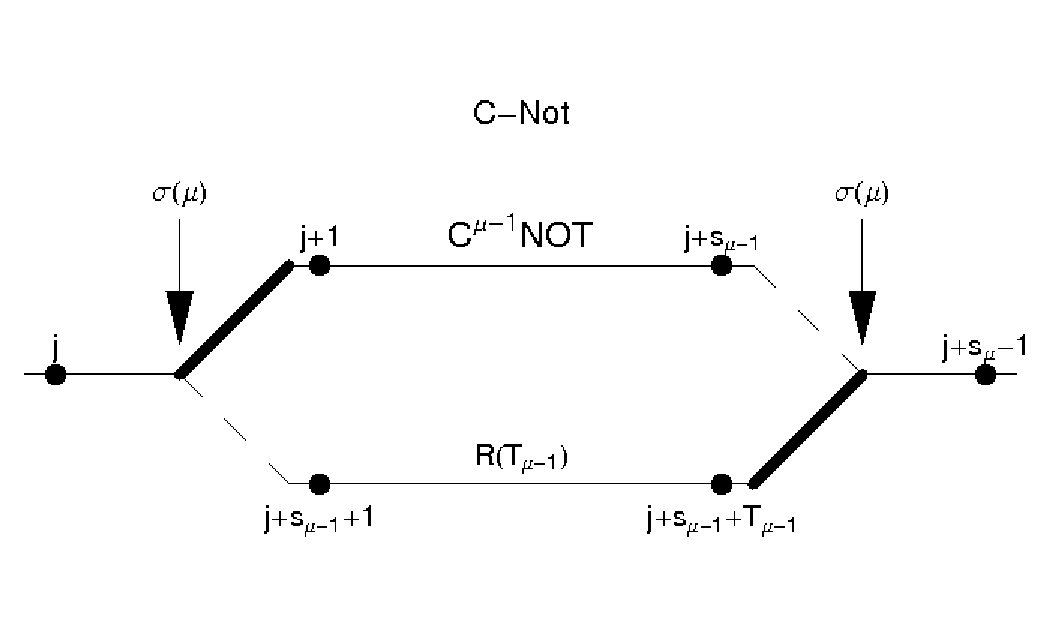}
	\caption{$C^{\mu-1}NOT \rightarrow C^\mu NOT$.}
	\label{fig:Switch1}
\end{figure}
The case  $\mu=1$ of one controlling qubit has been studied in \cite{feyn86}. It involves the introduction of $s_1=6$ cursor qubits $\underline{\tau}(j),\underline{\tau}(j+1),\ldots,\underline{\tau}(j+5)$ and, supposing that the controlling qubit is $\sigma_3(1)$  and the controlled one is $\sigma_3(\nu)$, of the \emph{local} Hamiltonian
\begin{eqnarray} \label{eq:hcnot}
H_{CNOT} & = & \sigma_-(1) \ketbra{j+1}{j} +  \nonumber \\ 
& + & \sigma_1(\nu) \ketbra{j+2}{j+1} + \nonumber \\ 
& + & \sigma_+(1) \ketbra{j+5}{j+2} + \nonumber \\
& + &  \sigma_+(1) \ketbra{j+3}{j} + \nonumber \\ 
& + & \underbrace{\ketbra{j+4}{j+3}}_{\mbox{delay line}} + \nonumber \\ 
& + &\sigma_-(1) \ketbra{j+5}{j+4}+ H.c.
\end{eqnarray}
A graphical representation of \eref{eq:hcnot} is given in \fref{fig:Switch}. \\
The term  R$(j+3,j+4) =  1 \ketbra{j+4}{j+3} $ in \eref{eq:hcnot}, represented as R$(1)$ in \fref{fig:Switch}, plays the role of a \emph{delay line} \index{delay line} of length $1$. It makes the length $T_1=4$  of the computation independent of the input word in the sense that an initial state of the form  $\mbox{\ket{\sigma_3(1)=+1}}\otimes \mbox{\ket{\sigma_3(\nu)=z_\nu}} \mbox{\ket{Q=j}}$ has the same \emph{number of logical successors}
\begin{eqnarray}
\ket{\sigma_3(1)=-1}\otimes \ket{\sigma_3(\nu)=z_\nu} \ket{Q=j+1} \nonumber \\
\ket{\sigma_3(1)=-1}\otimes \ket{\sigma_3(\nu)=-z_\nu} \ket{Q=j+2}  \nonumber \\
\ket{\sigma_3(1)=+1}\otimes \ket{\sigma_3(\nu)=-z_\nu} \ket{Q=j+5} 
\end{eqnarray}
as an initial state of the form $\ket{\sigma_3(1)=-1}\otimes \ket{\sigma_3(\nu)=z_\nu} \ket{Q=j}$, which has the successors
\begin{eqnarray}
\ket{\sigma_3(1)=+1}\otimes \ket{\sigma_3(\nu)=z_\nu} \ket{Q=j+3} \nonumber \\
\ket{\sigma_3(1)=+1}\otimes \ket{\sigma_3(\nu)=z_\nu} \ket{Q=j+4}  \nonumber \\
\ket{\sigma_3(1)=-1}\otimes \ket{\sigma_3(\nu)=z_\nu} \ket{Q=j+5}. 
\end{eqnarray}
\Fref{fig:Switch1} shows the iteration step leading from $C^{\mu-1}NOT$  to $C^{\mu}NOT$  through the introduction of the additional controlling qubit $\sigma_3(\mu)$. The length of each computation increases from the previous value  $T_{\mu-1}$ to
\begin{equation}
	T_\mu = T_{\mu-1} + 2 = 2 (\mu+1).
\end{equation}
The number of cursor qubits increases, because also of the delay line $R(T_{\mu-1})$, from the previous value  $s_{\mu - 1}$ to 
\begin{equation}
	s_\mu = 2 + s_{\mu-1} + T_{\mu-1} = (\mu+1)(\mu+2).
\end{equation}
The iteration step $C^{\mu-1}NOT \rightarrow C^{\mu}NOT$ is explicitly given by
\begin{eqnarray}
C^{\mu}NOT(j,j+s_\mu-1) & = & \sigma_-(\mu) \ketbra{j+1}{j} + C^{\mu-1}NOT(j+1,j+s_{\mu-1}) + \nonumber \\
& + & \sigma_+(\mu) \ketbra{j+s_\mu-1}{j+s_{\mu-1}} + \nonumber \\
& + & \sigma_+(\mu) \ketbra{j+s_{\mu-1}+1}{j} + \nonumber \\
& + & \sum_{k=1}^{T_{\mu-1}-1} \ketbra{j+s_{\mu-1}+k+1}{j+s_{\mu-1}+k}+ \nonumber \\
& + & \ketbra{j+s_\mu-1}{j+s_\mu-1+T_{\mu-1}}+ H.c.
\end{eqnarray}
\begin{observation}
For $\mu=2$ we implement the $CCNOT$ gate, or \emph{Toffoli gate}, which is a complete logical basis for the reversible boolean functions. The Feynman machine is thus able to compute \emph{at least} every function computable by a deterministic Turing machine (or, in other words, the class of functions reversibly computable by a Feynman machine is at least as wide as the class of partial recursive functions) by means of only three body interactions . In fact, every term in the Hamiltonian for the $CCNOT$ will involve two spins of the program line and one spin of the register.
The truth table of the CCNOT is reported in table \ref{tab:truthCCNOT}.
\begin{table*}[h]
	\centering
		\begin{tabular}{||c | c | c || c | c| c||} \hline
			$\sigma_3(1)$ & $\sigma_3(2)$ & $\sigma_3(\nu)$ & $\sigma_3(1)'$ & $\sigma_3(2)'$ & $\sigma_3(\nu)'$  \\ \hline \hline
			-1 & -1 & -1 & -1 & -1 & -1  \\ \hline
			-1 & -1 & +1 & -1 & -1 & +1  \\ \hline
			-1 & +1 & -1 & -1 & +1 & -1  \\ \hline
			-1 & +1 & +1 & -1 & +1 & +1  \\ \hline
			+1 & -1 & -1 & +1 & -1 & -1  \\ \hline
			+1 & -1 & +1 & +1 & -1 & +1  \\ \hline
			+1 & +1 & -1 & +1 & +1 & +1  \\ \hline
			+1 & +1 & +1 & +1 & +1 & -1  \\ \hline \hline
		\end{tabular}
	\caption{The truth table of the C-C-NOT.}
	\label{tab:truthCCNOT}
\end{table*}
\end{observation}
\begin{figure}[h]
	\centering
		\includegraphics{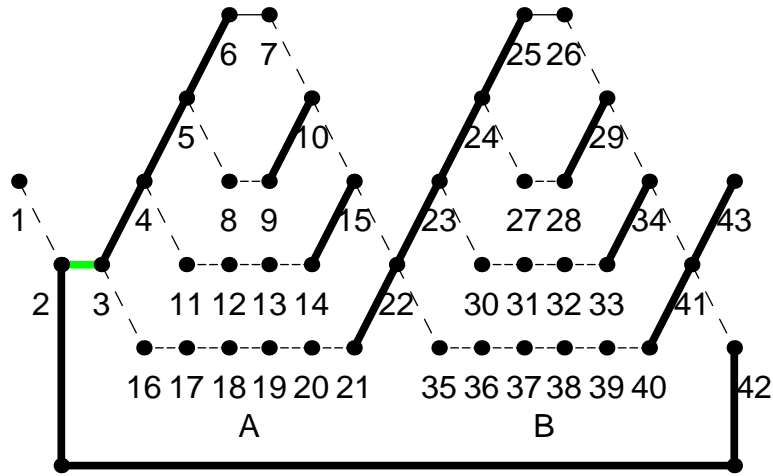}
	\caption{Do $BA$ twice on a register of $\mu=3$ qubits as described in terms of three body interactions. Delay lines are represented by the horizontal dashed lines.}
	\label{fig:grovermu3}
\end{figure}
The ``double diamond'' circuit of \fref{fig:grovermu3} is an example of a planar graph implementing, with the right setting of the input/output register, Grover's algorithm for $\mu=3$: the excitation, which in our setup will initially be located at site $1$, moves on the planar graph depending on the state of the input/output register. During the ``walk'', the excitation interacts with the register degrees of freedom. The implementation of an algorithm on a Feynman quantum computer thus leads to the definition of \emph{interacting quantum walks}.
\section{The SWITCH and the projective control}
Although the CCNOT is universal with respect to the class of reversible boolean functions, it is useful to generalize the kind of trajectory control used for the $CCNOT$ to the ``pure'' IF...THEN...ELSE construct. Consider the Hamiltonian
\begin{eqnarray} \label{eq:hswitch}
H_{SWITCH} & = & \sigma_- \ketbra{j+1}{j} + A \ketbra{j+2}{j+1} + \sigma_+ \ketbra{j+5}{j+2} + \nonumber \\
& + &  \sigma_+ \ketbra{j+3}{j} + B \ketbra{j+4}{j+3} + \sigma_+ \ketbra{j+5}{j+4}+h.c.
\end{eqnarray}
\begin{figure}[h!]
	\centering
		\includegraphics[width=10cm]{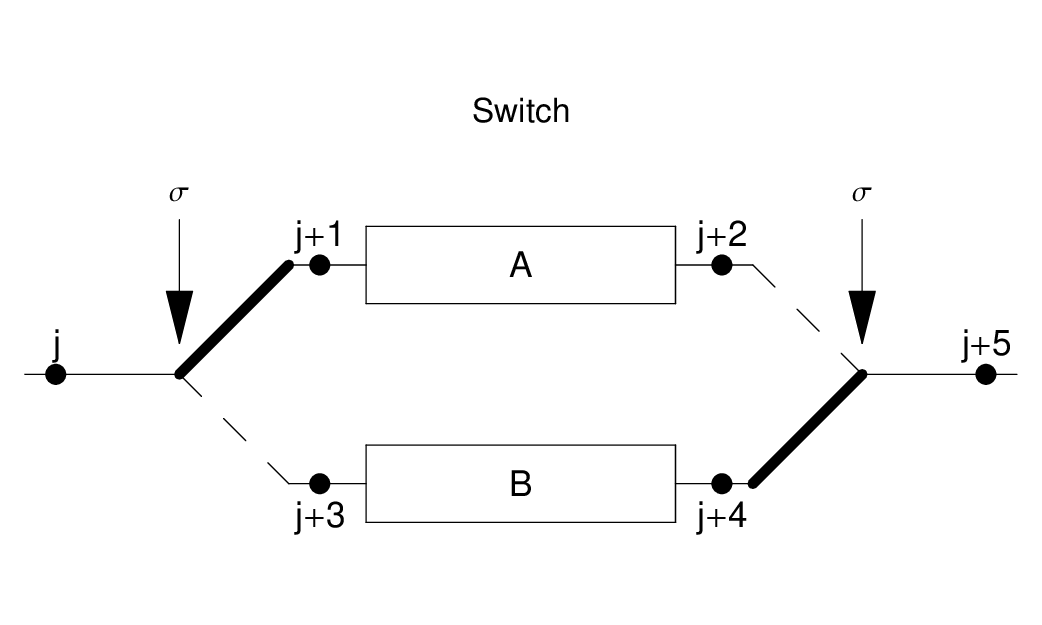}
	\caption{A SWITCH implementing the $IF(\sigma=+1)\ THEN\  A\  ELSE\  B$ instruction of Hamiltonian \eref{eq:hswitch}. $A$ and $B$ are primitive acting on the register. It is required that $A$ and $B$ do not modify the state of the controlling qubit $\sigma$.}
\end{figure}
It is worth mentioning that the same control on the trajectory of the walker can be implemented by means of  a purely projective mechanism; consider the following Hamiltonian
\begin{eqnarray}
H_{SWITCH} & = & \left ( \frac{1+\sigma_3}{2}\right ) \ \ketbra{j+1}{j}+  A \ketbra{j+2}{j+1} + \nonumber \\
& + & \left( \frac{1+ \sigma_3}{2} \right )  \ketbra{j+5}{j+2} + \nonumber \\
& + &  \left (\frac{1- \sigma_3}{2} \right )  \ketbra{j+3}{j} +
 B \ketbra{j+4}{j+3} + \nonumber \\
& + & \left (\frac{1-\sigma_3}{2}\right )  \ketbra{j+5}{j+4} + h.c.
\end{eqnarray}
Although the control of the trajectory of the cursor traveling along the program line is the same as the one of \eref{eq:hswitch}, the state of the controlling qubit $\sigma$ is left unchanged all through the SWITCH. This projective control mechanism has been introduced in \cite{defa04} and makes it easier to track different computational paths; for example, to understand what happens on the upper branch of a switch it suffices to consider only the projection on the $\sigma_3=+1$ subspace of the state space of the machine.
%
%
%
%
\section*{Summary}
Throughout this section we have used a top-down approach to the implementation of Grover's algorithm on a Feynman machine: we started from the computational primitives $A$ and $B$ seen as black boxes and ended up with the same primitives decomposed in three body interactions. Moreover, we have shown that the iteration of quantum subroutines can be implemented by means of a \emph{linear} amount of space resources. The possibility of using the same functional block many times amply justifies this linear space, an therefore time, overhead.\\
The topology of a Feynman machine can then be extended from a linear graph to a planar graph; at the same time, the interaction between the register and the cursor can be exploited to make the clocking excitation visit different computational paths; the possibility of having the controlling qubit in a superposition of states allows for the different computational paths to be explored simultaneously. In this framework,  the meaning of \emph{quantum parallelism} becomes clear: different physical computational trajectories are visited simultaneously.\\
The algorithmic construction of the Peres basis is directly extensible to this general form of graph. This allows for the state of the system to be described by \eref{eq:psit2}.
\chapter{Quantum timing and synchronization problems} \label{chap:quantumtiming}
The interacting quantum walk presented so far implements a clocking mechanism. In fact  ``\underline{If} at some later time the final site $s$ is found to be in the  $\ket{C(s)}$  state then all the computational primitives  $U_{s-1}\cdot \dots \cdot U_2 \cdot U_1$  have been administered to the register in the correct order'' \cite{feyn86}. We have  shown in chapter \ref{chap:grover} that the same kind of interaction can be exploited to control the motion of the clocking excitation on a planar graph. However, the quantum nature of the cursor imposes limitations on our ability to know, without preforming a measurement, whether the computation has finished. In this chapter we study the motion of the cursor and make quantitative the following qualitative assertions:
\begin{enumerate}
\item  at no instant of time the probability $| c(t,s;s) |^2$ is larger than $const \cdot s^{-\frac{2}{3}}$ \cite{apolloni02} ; \label{R:i}
\item the cursor keeps bouncing back and forth between positions 1 and $s$, thus in effect making the above upper bound attainable only at selected instants of time.\label{R:ii}
\end{enumerate}
In other words, by scaling down the clocking mechanism of the  computing device to the quantum regime two quantum phenomena become relevant: the spreading of the probability distribution of the excitation (or \emph{pointer}) along the program lines, and the scattering of the probability amplitude at the two endpoints of the physical space allowed for its motion.\\
We begin this chapter by discussing the \emph{timing problems} \ref{R:i} and \ref{R:ii}. We then propose a measurement scheme which, as proved in \cite{apolloni02}, makes the upper bound on the \emph{probability cost} \index{probability cost} of the implementation of an algorithm on a Feynman machine less severe.\\
We conclude the chapter by showing that the results obtained for the linear chain can be extended to a more general class of planar graphs, such as the one of \sref{chap:grover}, by means of synchronizing delay lines. 
\begin{figure}[h]\label{fig:figPxt1}
	\centering	
	\subfigure[]{\includegraphics[width=145pt]{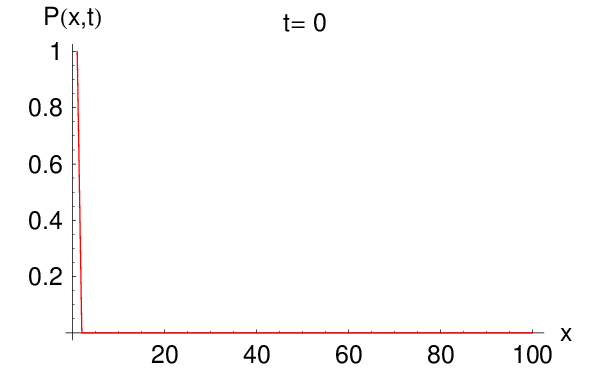}\label{fig:Pxt1a}} 
	\subfigure[]{\includegraphics[width=145pt]{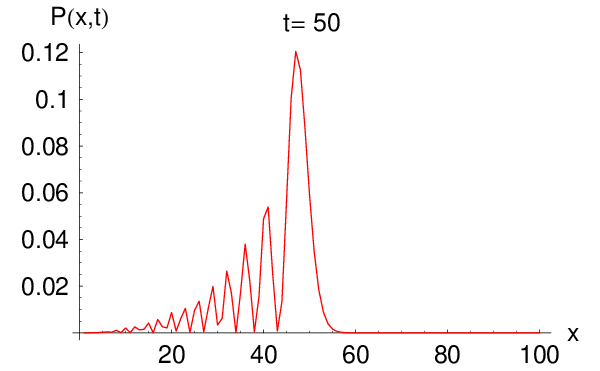}\label{fig:Pxt1b}}
	\subfigure[]{\includegraphics[width=145pt]{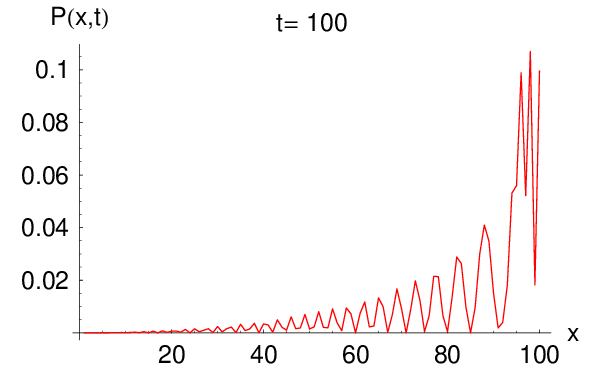}\label{fig:Pxt1c}}
	\subfigure[]{\includegraphics[width=145pt]{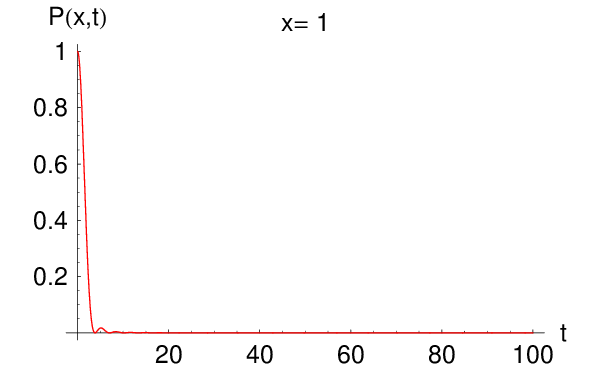}\label{fig:Pxt1d}} 
	\subfigure[]{\includegraphics[width=145pt]{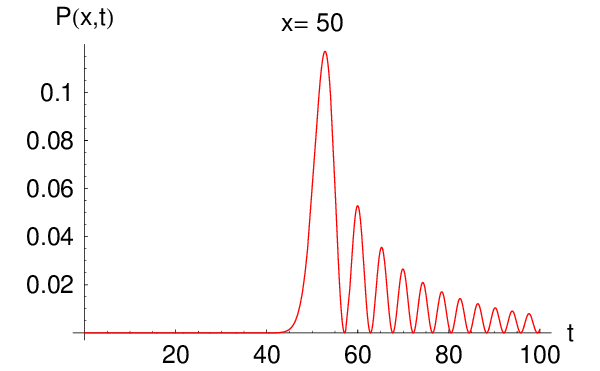}\label{fig:Pxt1e}}
	\subfigure[]{\includegraphics[width=145pt]{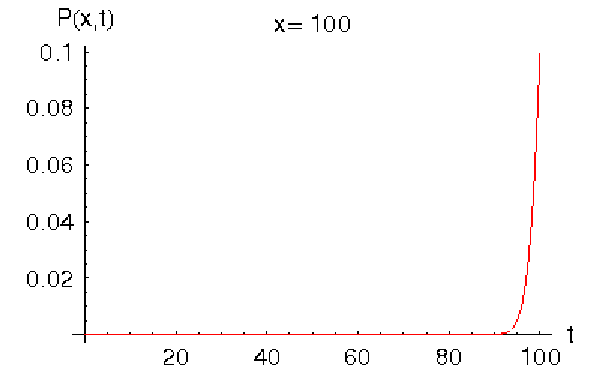}\label{fig:Pxt1f}}
	\caption{We take a linear chain of length $s=100$. Figures (a) (b) and (c) represent the probability distribution for the position observable $Q$ at different times; the fast spreading of the wave packet is confirmed by the; \ref{fig:Pxt1a} \ref{fig:Pxt1b} \ref{fig:Pxt1c} represent the probability of finding the clocking excitation at sites $1,50$ and $100$ respectively. The reader should pay attention to the vertical axes scale.}
\end{figure}
\section{Motion of the cursor}
The eigenvalue problem for the Hamiltonian
\begin{equation} \label{eq:hinteragente}
	H = -\frac{\lambda}{2}\sum_{x=1}^{s-1} U_x \ \ketbra{x+1}{x}+h.c.,
\end{equation}
corresponding to a \emph{sequential program line} of length $s$, is solved by  the following Ansatz for the eigenstates:
\begin{equation}
	\ket{e}= \sum_{x=1}^{s} v(x) U_{x-1}U_x \ldots U_2 U_1 \ket{r} \otimes \ket{C(x)}, 
\end{equation}
where \ket{r} is a unit  vector in \Hr and \ket{C(x)} is defined as in \eref{eq:cj},  together with the boundary conditions $v(0)=v(s+1)=0$. Inserting this Ansatz, the eigenvalue problem
\begin{equation}
	H \ket{e} = e \ket{e}
\end{equation}
 becomes:
\begin{equation}
	e \ v(x) = -\frac{\lambda}{2}(v(x+1)+v(x-1)). 
\end{equation}
This leads in an obvious way \cite{gramss95} to the eigenvalues
\begin{equation}
	e_k = -\lambda \ \cos\left( \frac{k \pi}{s+1}\right),\ k=1,2,\ldots,s.
\end{equation}
The multiplicity of each eigenvalue is equal to $d=dim(\Hr)$. An orthonormal basis in the eigenspace belonging to the eigenvalue $e_k$ is given by:
\begin{equation} \label{eq:ekrj}
	\ket{e_k; r_j}= \sum_{x=1}^{s} v_k(x)   U_{x-1}\cdot \ldots \cdot U_1 \ket{r_j}\otimes \ket{C(x)} ,
\end{equation}
where $\ket{r_1},\ldots,\ket{r_d}$ is an orthonormal basis in \Hr, and
\begin{equation} \label{eq:autofunzioni}
	v_k(x)= \sqrt{\frac{2}{s+1}} \sin\left( \frac{k \pi}{s+1} x \right).
\end{equation}
We observe that, for $1 \leq x_0 \leq s$, it holds
\begin{equation}
	\sum_{j=1}^d \braket{e_k;r_j}{C(x_0)} =  \sqrt{\frac{2}{s+1}}\sin\left( \frac{k \pi x_0}{s+1}\right).
\end{equation}
For $x_0=1$ we obtain for the coefficients $c(t,k;s)$
\begin{equation} \label{eq:c}
	c(t,x;s)= \frac{2}{s+1} \sum_{k=1}^s \exp\left[ i \lambda t \cos\left( \frac{k \pi}{s+1}\right)\right] \sin\left( \frac{k \pi}{s+1}\right) \sin\left( \frac{k \pi x}{s+1} \right).
\end{equation}
We point out that, given an Hamiltonian of the form \eref{eq:hinteragente} and an initial state $\ket{\psi_1}=\ket{R(1)}\otimes \ket{C(1)}$, once defined the corresponding Peres' basis $\{\psi_k\}_{k=1}^{s}$, the dynamics of the system is always of the form
\begin{equation}
	\ket{\psi(t)} = \sum_{x=1}^s c(t,x;s) \ket{\psi_x},
\end{equation}
thus independent of the transformation applied to the register.\\
This property of the computational device allows for the timing problem to be discussed independently of the transformations carried on the register, as long as we act with unitary operators on the latter.\\
With the analysis of Reference \cite{apolloni02} it has been shown that:
\begin{theorem} \label{th:apo}
Given a linear chain of length $s$ evolving under the Hamiltonian $H$ of \eref{eq:hmatrix}, the probability $P(t,s;s)$ of finding, upon measurement, the clocking excitation at the final site $s$ is bounded by
\begin{equation}
	P(t,s;s) = |c(t,s;s)|^2 = O(s^{-\frac{2}{3}}).
\end{equation}
We refer to Reference \cite{apolloni02} for a proof of the statement.
\end{theorem}
Theorem \ref{th:apo} makes it clear that, as soon as the clocking mechanism of a quantum computer is scaled down to the quantum regime, a new problem, to which we refer as the \emph{completion of computation problem} appears: we are never sure to find the computation completed, \emph{however carefully} we choose the instant of time at which to measure the clock. Not only: the clocking excitation goes back and forth along the program line, thus doing and undoing the computation. On the other side, the entanglement between the register and the cursor makes it possible to concentrate all the quantum uncertainty about the state of the computation on the position of the clocking particle \cite{margolus90}.\\[5pt]
Given the initial condition for the program line $\ket{Q=x_0}$, with $x_0 \in \{1,\ldots,s\}$, by taking the limit of a ``long'' computation ($s \rightarrow \infty$) we obtain
\begin{eqnarray}
	\lim_{s \rightarrow \infty} c(t,x;s)&  = & \nonumber \\
	& = &  \lim_{s \rightarrow \infty} \frac{2}{\pi} \frac{\pi}{s+1} \sum_{k=1}^s \exp\left[ i \lambda t \cos\left( \frac{k \pi}{s+1}\right)\right] \sin\left( \frac{k \pi x_0}{s+1}\right) \sin\left( \frac{k \pi x}{s+1} \right) = \nonumber \\
	& = & \frac{2}{\pi} \int_{0}^{\pi} \exp(i \lambda \cos y) \sin(y x_0) \sin(xy) dy = \nonumber \\
	& = &  i^{x-x_0} J_{x-x_0}(t) - i^{x+x_0} J_{x+x_0}(t);
\end{eqnarray}
$J_y(t)$ being Bessel functions of the first type \cite{watson62} and $\lambda=1$. In particular, for $x_0=1$ we get
\begin{equation} \label{eq:bj}
		 i^{x-1}(J_{x-1}(t)+J_{x+1}(t)) = i^{x-1} \frac{2x}{t} J_x(t).
\end{equation}
It is worth noticing that the solution in the $s \rightarrow +\infty$ limit  gives a very good approximation of the behavior of the system even in the case of a finite chain as long as we consider periods of time sufficiently short not to include reflections of the wave packet associated to the cursor on the boundaries (see \fref{fig:approxbessel}).
\begin{figure}[h]
	\centering
		\includegraphics[width=10cm]{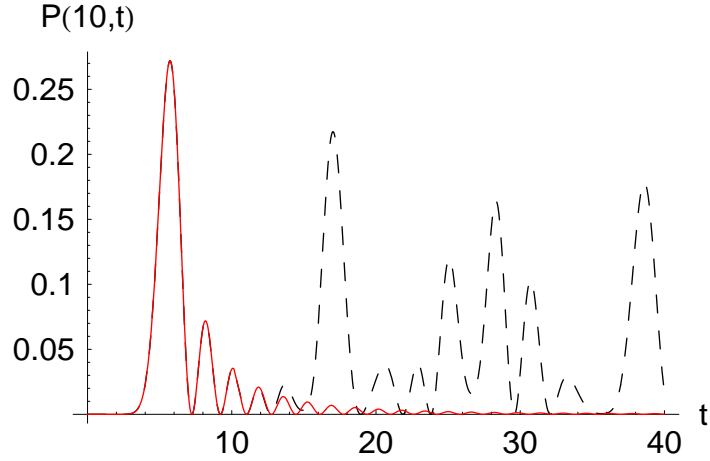}
	\caption{solid red line: $|c(t,10;20)|^2$; dashed blue $\left (\frac{20}{t} J_{10}(t)\right)^2$. The approximation via Bessel function is precise until the wave packet gets reflected on the boudary of the approximated finite chain.}
	\label{fig:approxbessel}
\end{figure}
\section{Relaxing the upper bound} \label{sec:relaxing}
The main purpose of this Section is to give examples of the behavior recalled in the assertions \ref{R:i} and \ref{R:ii} made at the beginning of this chapter.\\
This we do with the help  of the following Hamiltonian:
\begin{eqnarray}
\label{E:htelchain}
H& = &  \Bigl( \sum_{j=1}^{s-1} U_j \tau_+(j+1) \tau_-(j)+ \nonumber \\
 &   & + \rho_- \tau_+(s+1)\tau_-(s)+ \nonumber \\
 &   & +\sum_{j=s+1}^{s+\delta-1}\tau_+(j+1)\tau_-(j)+ h.c. \Bigr ).
\end{eqnarray}
$U_j,\ 1 \leq j \leq s-1$ being unitary operators acting on the register. With respect to the Hamiltonian (\ref{eq:hamfeynman}), we have introduced an additional {\it control} q-bit $\underline{\rho}=(\rho_1,\rho_2,\rho_3)$ in the term $\tau_+(s+1)\rho_- \tau_-(s)$; this is an example of a \emph{conditional jump} in the quantum walk performed by the cursor: it acts non trivially only in the eigenspace belonging to the eigenvalue +1 of $\rho_3$, \emph{enabling}  the transition  $| Q=s \rangle \rightarrow | Q = s+1 \rangle $. If this transition is enabled, then the cursor can visit the additional {\it telomeric} sites $s+1,\dots, s+\delta$, else it gets reflected back.\\
Figures \ref{F:cursor} and \ref{F:telo} give examples of the behaviour of the probability
\begin{equation}
P_{(s \leq Q)}(t)=P_{(s\leq Q \leq s+\delta)}(t)=\sum_{j=s}^{s+\delta} | \gamma(t,j) |^2
\label{E:Probchain}
\end{equation}
of finding the register in the state $A = U_{s-1}\cdot \dots \cdot U_2 \cdot U_1 \ket{R(1)}$, under two different initial conditions, which determine two different forms of the amplitudes~$\gamma$.\\
Figure \ref{F:cursor} corresponds to the initial condition \mbox{$\ket{\psi_1}$} = \mbox{$|\rho_3=-1\rangle \otimes$}\mbox{$| Q=1 \rangle $}: the motion of the cursor remains confined to the sites $1,\dots,s$, as it is $\gamma(t,k)=c(t,k;s)$ if $1 \leq k \leq s$, 0 otherwise. The probability $P_{(s \leq Q)}(t)$ of finding the computation completed satisfies in this case the inequality:\cite{apolloni02}
\begin{equation}
P_{(s \leq Q)}(t) \leq \frac{8.}{s^{\frac{2}{3}}}
\label{E:probnotelo}
\end{equation}
\begin{figure}[!h]
\hspace{2cm}
\begin{picture}(180,130)(0,0)
\centering
\includegraphics[width=8cm]{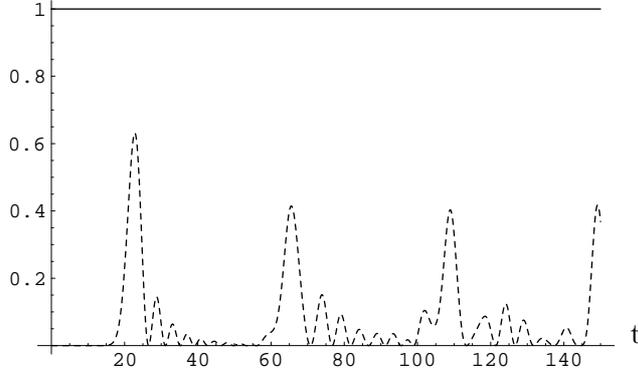}
\end{picture}
\begin{picture}(10,130)
\put(50,11){t}
\end{picture}
\caption{$\ket{\psi_1} = |\rho_3=-1 \rangle \otimes | Q=1 \rangle$; $s=20$.}
\label{F:cursor}
\end{figure}
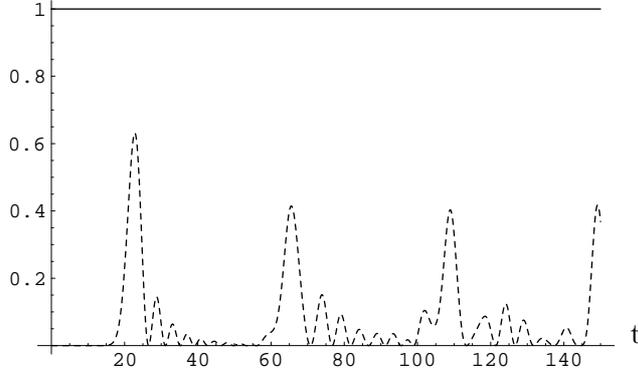
Figure \ref{F:telo} corresponds to the initial condition: $\ket{\psi_1}=|\rho_3=+1\rangle \otimes$ \mbox{$| Q=1 \rangle$}, leading to $\gamma(t,k)=c(t,k;s+\delta)$ for $1 \leq k \leq s+\delta$. For $t$ just below $s+2\delta$ the probability $P_{(s \leq Q)}(t)$ of finding the computation completed is close to the much less severe upper bound:\cite{apolloni02}
\begin{equation}
P_{(s \leq Q \leq s+\delta)}(t) \leq 1-\frac{2}{\pi} \Biggl (  \arcsin \Bigl (  \frac{1}{1+2 \delta/s}\Bigr )- \Bigr (   \frac{1}{1+2\delta/s} \Bigr ) \sqrt{1- \Bigl (\frac{1}{1+2 \delta/s} \Bigr ) ^2 }\Biggr).
\label{E:upperbound}
\end{equation}
\begin{figure}[!h]
\hspace{2cm}
\begin{picture}(180,130)(0,0)
\centering
\includegraphics[width=8cm]{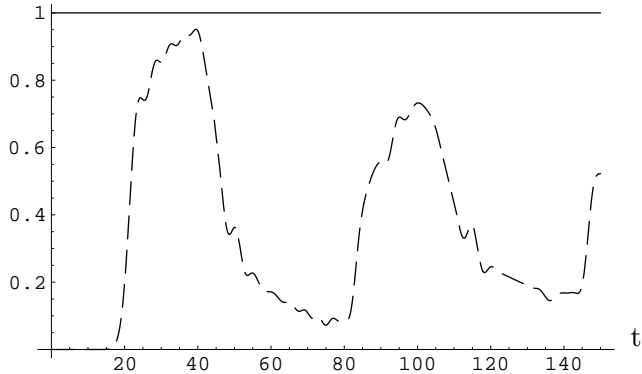}
\end{picture}
\begin{picture}(10,130)
\put(50,11){t}
\end{picture}
\caption{$\ket{\psi_1}= |\rho_3=+1\rangle \otimes |Q=1 \rangle$; $s=20$; $\delta =10$.}
\label{F:telo}
\end{figure}
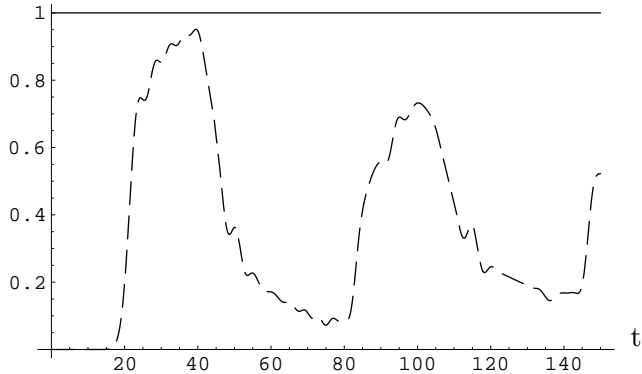
\section{The Quantum END Instruction} \label{S:pipulse}
\begin{figure}[!hb]
\hspace{1.7cm}
\begin{picture}(180,130)(0,0)
\centering
\includegraphics[width=9cm]{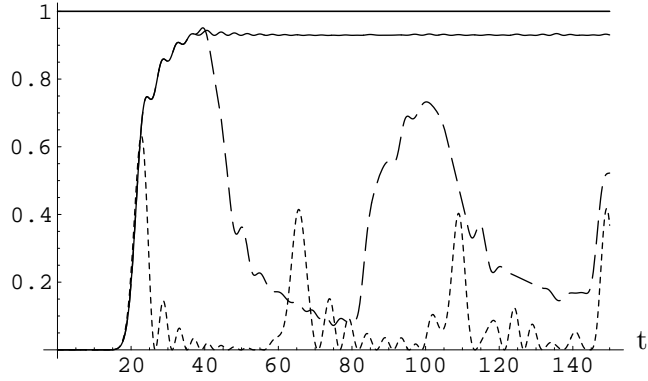}
\end{picture}
\begin{picture}(10,130)
\put(60,11){t}
\end{picture}
\caption{ The solid line represents the probability of finding the cursor in the telomeric chain using a $\pi$-pulse applied at time $t_0=s+2\delta$. The dashed lines correspond to Figs.\ref{F:cursor} and \ref{F:telo}.}
\label{F:pipulse}
\end{figure}
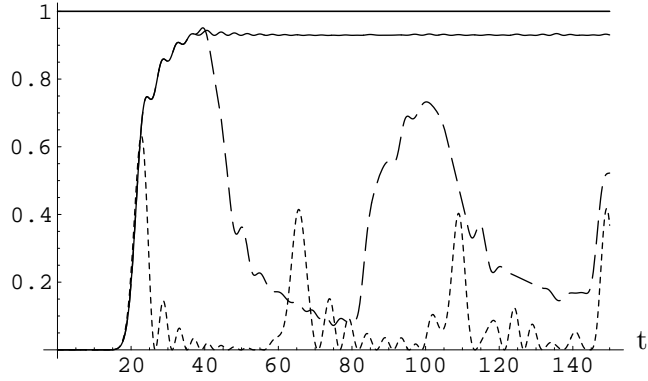
The abrupt collapse of $P_{(s \leq Q)}(t)$ at time $t\approx s+2\delta$, evident from Fig. \ref{F:telo}, corresponds to the following fact: traveling with average speed close to $1$, at time $t \approx s+ 2\delta$ the cursor ``returns down the active part of the program line'', thus, in effect, undoing the calculation.\\
Bringing the computation to an END, and storing the result is not completely trivial in the case examined here of a reversible quantum clocking mechanism: ``Surely a computer has eventually to be \underline{in interaction with the external world}, both for putting data in and for taking it out\cite{feyn86}''.\\
A simple model of such interaction is suggested by inspection of the Hamiltonian (\ref{E:htelchain}): starting from the initial condition $| \rho_3=+1 \rangle$, the transition \mbox{$| Q=s \rangle$} \mbox{$\rightarrow$} \mbox{$|Q=s+1 \rangle$} is enabled by the control term $\tau_+(s+1)\rho_- \tau_-(s)$ which, simultaneously, determines the transition $|\rho_3=+1 \rangle \rightarrow |\rho_3=-1 \rangle$.\\
The transition $|Q=s+1 \rangle \rightarrow |Q=s \rangle$, enabled by the hermitian conjugate term $\tau_+(s)\rho_+ \tau_-(s+1)$, will be therefore inhibited if, by external means, we enforce the transition $|\rho_3=-1 \rangle \rightarrow |\rho_3=+1 \rangle$ at a time, close to $t_0=s+2\delta$, when most of the probability mass is in the region $s,\dots, s+\delta$.\\
Figures \ref{F:pipulse} (where Figs. \ref{F:cursor} and \ref{F:telo} are also reproduced for comparison purpose) presents the effect of the addition to (\ref{E:htelchain}) of the time dependent perturbation
\begin{equation}
h(t)=B(t)\cdot \rho_1
\label{E:pipulse}
\end{equation}
where the ``magnetic field'' $B(t)$ is non vanishing only in a unit time interval around $t_0$, in which it takes the value $\pi$.\\
With a probability depending only on the ratio $\delta/s$ (see (\ref{E:upperbound})) between the lengths of the active part and the telomeric part of the program line, the $\pi$-pulse (\ref{E:pipulse}) definitively prevents the cursor from undoing the computation.\\
The idea of a $\pi$-{\it pulse trap} just presented works only if the control q-bit is initialised in the $|\rho_3=+1\rangle$ state. It is immediate to convince oneself that the following {\it double trap} Hamiltonian does not suffer from the above limitation:
\begin{eqnarray}
\label{E:htelfork}
H& = &h(t)+ \nonumber \\
 &   &-\Bigl( \sum_{j=1}^{s-1} U_j\  \tau_+(j+1) \tau_-(j)+\sum_{j=s+1}^{s+\delta-1}\tau_+(j+1)\tau_-(j)+\sum_{j=s+\delta+1}^{s+2\delta-1}\tau_+(j+1)\tau_-(j)+
  \nonumber \\
 &   &+ \rho_- \  \tau_+(s+1)\tau_-(s)+ \rho_+ \ \tau_+(s+\delta+1)  \tau_-(s) +h.c. \Bigr ).
\end{eqnarray}
With any initial condition for the control q-bit, under the action of the above Hamiltonian, the $|\rho_3=+1\rangle$ component of the state gets definitively trapped in the first telomeric region $\{s+1,\dots,s+\delta \}$, the $|\rho_3=-1 \rangle$ component in the second one $\{s+\delta+1,\dots,s+2\delta \}$.\\[4pt]
As a final remark of this section, we observe that, acting in effect as a Stern-Gerlach apparatus providing \emph{space} separation between two different \emph{spin} states, the term
\begin{equation}
switch=(\rho_- \tau_+(s+1)  \tau_-(s)+  \rho_+ \tau_+(s+\delta+1)  \tau_-(s))+ h.c.
\label{E:switch}
\end{equation}
can be used also to model the preparation (``putting the data in'') of a {\it register} qubit in a given spin state.
\section{Synchronizing the computational paths} \label{sec:synchro}
When defining the $C^\mu NOT$ circuit (\sref{sec:localhamiltonian}) we made massive use of \emph{delay lines}, that is edges of the clocking device ``during which''  the state of the register does not change. As we said in there, this allows for every computational path to be of the same length \footnote{The same principle is applied in electric circuits when two or more signals are required to arrive at the same time through paths of different length.}; in turn, this makes it possible  the interference between different computational paths corresponding to different conditions of the controlling qubits to happen.\\[5pt]
As an example of the role of synchronization in our interacting quantum walks, let us consider the following Hamiltonians. 
\begin{eqnarray} \label{eq:hamh1}
	H_1 & = & SWITCH(1,6,\sigma_3(\nu),1) = \nonumber \\
	  & = & \sigma_-(1)\ketbra{2}{1} + \nonumber \\
	  & + & \sigma_3(\nu) \ketbra{3}{2} + \nonumber \\
	  & + & \sigma_+(1)\ketbra{6}{3} + \nonumber \\
	  & + & \sigma_+(1)\ketbra{4}{1}+ \nonumber \\
	  & + & \ketbra{5}{4}+ \nonumber \\
	  & + & \sigma_-(1)\ketbra{6}{5} +H.c.,
\end{eqnarray}
corresponding to \fref{fig:synchro}, and
\begin{eqnarray} \label{eq:hamh2}
	H_2 & = & \sigma_-(1)\ketbra{2}{1} + \nonumber \\
	  & + & \sigma_3(\nu) \ketbra{3}{2} + \nonumber \\
	  & + & \sigma_+(1)\ketbra{5}{3} + \nonumber \\
	  & + & \sigma_+(1)\ketbra{4}{1}+ \nonumber \\
	  & + & \sigma_-(1)\ketbra{5}{4}+H.c.,
\end{eqnarray}
corresponding to \fref{fig:nonsynchro}.\\
The difference between the two Hamiltonians consists in the lack of the delay line $4 \rightarrow 5$ in \eref{eq:hamh2}.\\
\begin{figure}[t]
	\centering	\includegraphics[width=10cm]{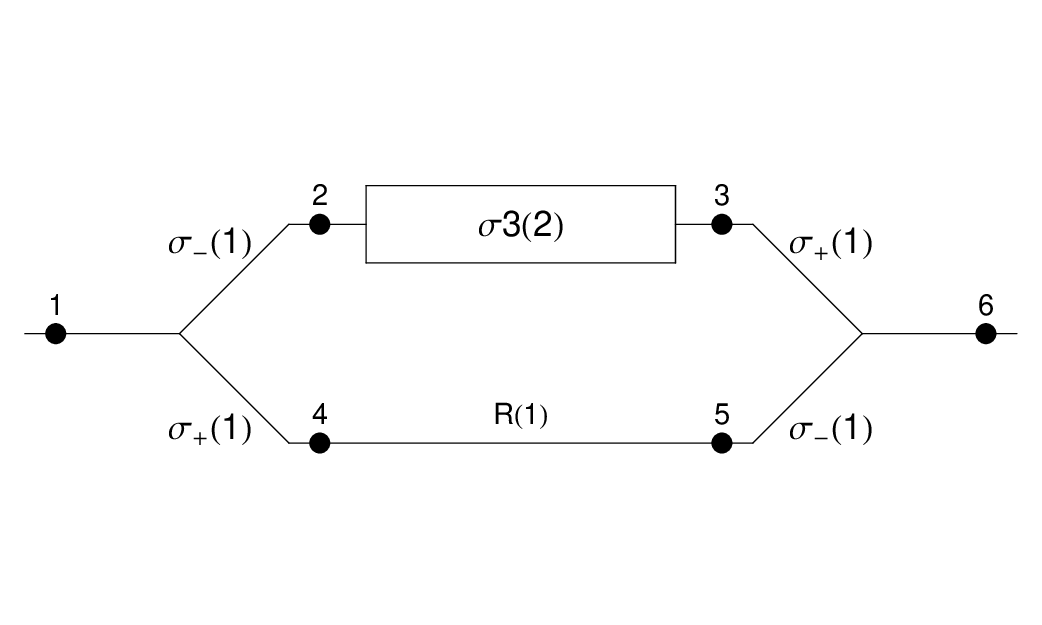}
	\caption{Synchronized circuit}
	\label{fig:synchro}
\end{figure}
\begin{figure}[t]
	\centering	\includegraphics[width=10cm]{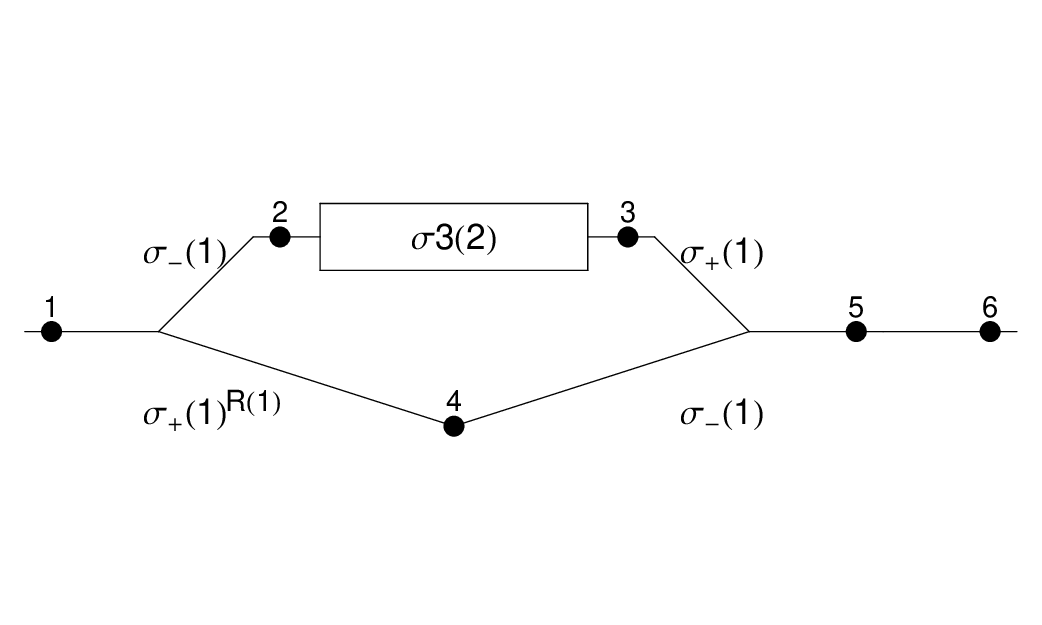}
	\caption{Non synchronized circuit}
	\label{fig:nonsynchro}
\end{figure}
We set the initial state of the register to 
\begin{eqnarray}
	\ket{R(1)} = \ket{\sigma_1(1)=+1,\sigma_3(2)=-1}
\end{eqnarray}
and the initial condition for the cursor to \ket{C(1)}.\\
The initial state of the machine is
\begin{equation}
	\ket{\phi_1} = \ket{R(1)}\ket{C(1)}.
\end{equation}
We solve the Cauchy problems
\begin{equation}
	\begin{cases}
	\frac{d}{dt} \ket{\phi(t)} = H_1 \ket{\phi(t)} \\
	\ket{\phi(0)} = \ket{\phi_1} 
	\end{cases}
\end{equation}
\begin{equation}
	\begin{cases}
	\frac{d}{dt} \ket{\phi'(t)} = H_2 \ket{\phi'(t)} \\
	\ket{\phi'(0)} = \ket{\phi_1}
	\end{cases}
\end{equation}
\begin{figure}[t]
	\centering	\includegraphics[width=10cm]{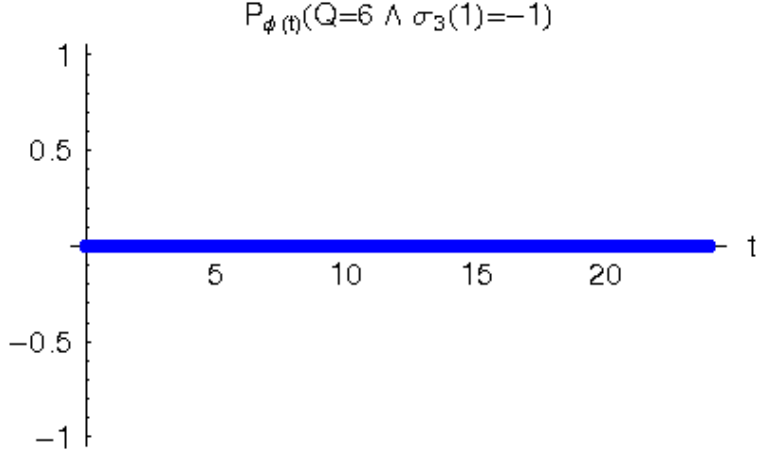}
	\caption{Synchronous case. The probability of finding, upon measurement, the cursor at position $6$ and the register in the state \ket{\sigma_1(1)=+1, \sigma_3(2)=-1}.}
	\label{fig:synchswitch}
\end{figure}
The choice of the initial state of the machine is such that both branches of the switch are visited with the same probability amplitude. In the circuit of \fref{fig:synchro}, along the upper branch the operator $\sigma_3(2)$ gives a $i \pi$ phase factor when acting on the register; the lower branch does nothing. Since the upper and lower computational paths have the same length, the clocking signal are synchronized and interfere destructively at site $6$ giving the result of \fref{fig:synchswitch}.\\
As \fref{fig:nsynchswitch} shows, the situation is completely different in the case of the circuit of \fref{fig:nonsynchro}. In fact, the computational path along the upper branch is one step longer than the computational path along the lower branch. Due to the lack of synchronization the interference pattern of the clocking signals at site $6$ is completely different from the one of the synchronized circuit.\\[5pt]
\begin{figure}[t]
	\centering	\includegraphics[width=10cm]{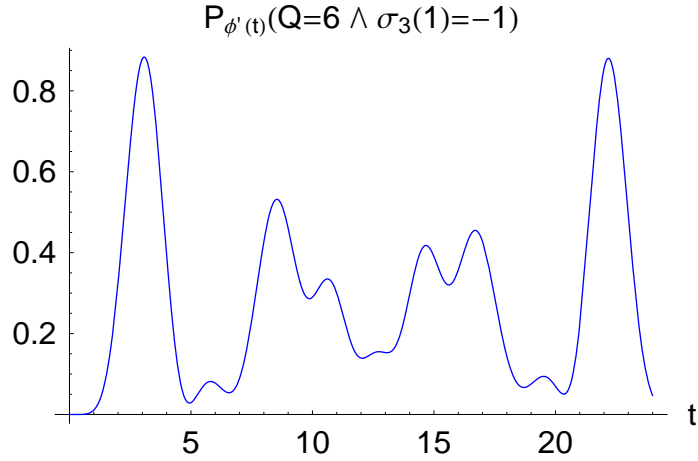}
	\caption{Asynchronous case. The probability of measuring $Q=6$ and the register in the state \mbox{\ket{\sigma_1(1)=1,\sigma_3(2)=-1}}.}
	\label{fig:nsynchswitch}
\end{figure}
The interference between different computational paths plays a key role in quantum algorithmics which, in fact, is essentially based on interference and on entanglement \cite{nielsen}. 
As anticipated in \sref{sec:localhamiltonian}, this architectural choice makes it possible to define the basis of logical successors of a given initial state of  the machine in the same way adopted in the case of the linear program line.  As a matter of fact, the Peres basis for the interacting quantum walk of Feynman is a generalization of the projection mechanism used in  \cite{childs02} to study the motion of a quantum walker on a graph $G_n$ such as the one of \fref{fig:graph}.
Similar results hold in the case, studied in \cite{defa03}, in which, because of {\it conditional jumps} in the program line (such as the ones needed in the iteration of {\it quantum subroutines} or the implementation of the $C^\mu NOT$), the cursor performs, in effect, a continuous time quantum walk \cite{childs02,farhi98} on a planar graph.
\begin{figure}[h]
	\centering		\includegraphics[width=10cm]{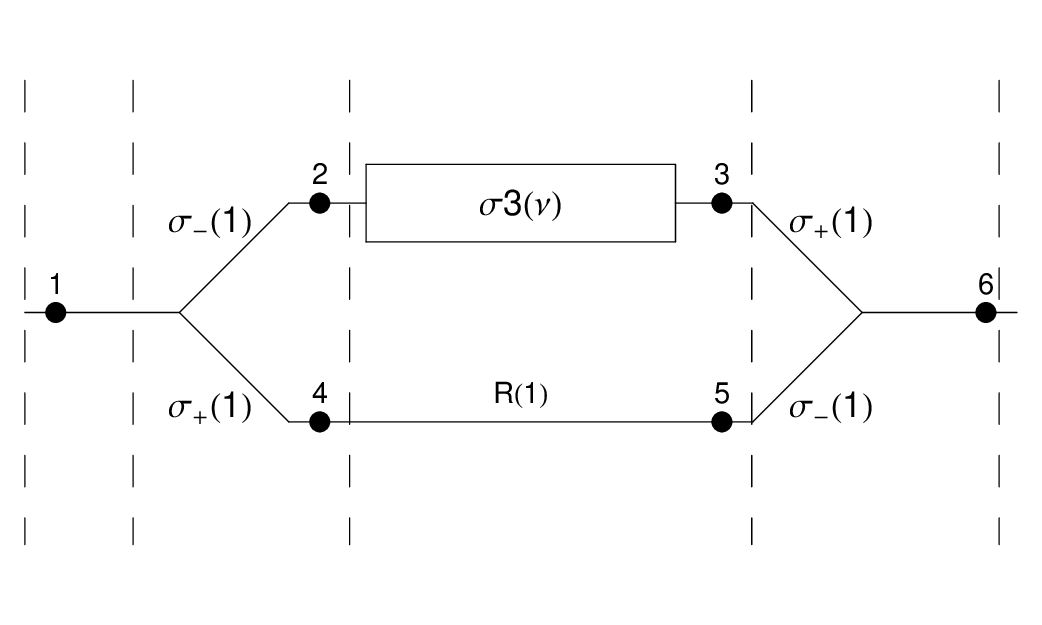}
	\caption{The same SWITCH of \fref{fig:synchswitch}. The vertical dashed lines divide the circuit in slices, or levels.}
	\label{fig:slicedswitch}
\end{figure}
\begin{figure}[h]
	\centering	\includegraphics[width=6cm]{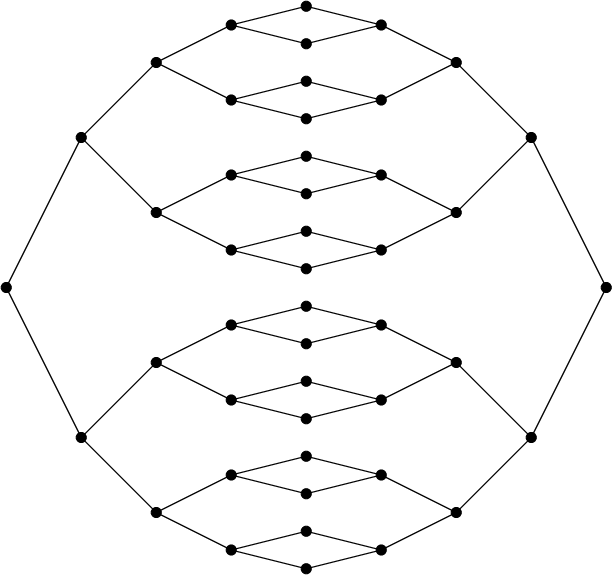}
		\caption{The graph $G_4$.}
	\label{fig:graph}
\end{figure}
Indeed, if we traced out the register degrees of freedom we could take the basis
\begin{equation}
 \ket{col(j)} = \frac{1}{N_j} \sum_{a \in \ slice \ j} \ket{Q=a}
\end{equation}
where $N_j$ is the number of program line sites in the $j$-th slice; the vectors \ket{col(j)} are hence uniform superposition of the basis states of \Hc belonging to the $j$-th column. With this notation, the motion on the SWITCH circuit of \fref{fig:slicedswitch} of the clocking excitation is exactly the same as the one of the traveling excitation of \cite{childs02}.
\section*{Summary}
The qualitative statement ``the longer is the computation the smaller is the probability of getting the cursor in the END position'' has been made quantitative in \cite{apolloni02}: $P_{opt}(s)= O(s^{-\frac{3}{2}})$. By means of a telomeric chain of length $\delta$, that is extra space resources, the latter bound can be raised to $O(\sqrt{\delta/s})$. We proposed a measurement scheme which has the same probability of success; it exploits the possibility of controlling the trajectory of the clocking excitation by means of the state of a register qubit: we wait for most of the probability mass of the cursor to be in the telomeric region and than prevent the excitation to get out of the telomeric chain by disabling the transition $s \to s+1$.\\
In Feynman's scheme the functioning of a quantum algorithm gets a ``topological'' characterization: different computational paths are simultaneously visited, each with a certain amplitude depending on the initial state of the register; at a certain point (in space) the signals coming from alternative routes interfere with each other. To get the right interference pattern it is necessary that the clocking signals are synchronized with each other.\\
The synchronization of different computational paths, on the other hand, makes it possible to realize the projection mechanism which, in turn, allows for the motion of the cursor on planar graphs to be described as if the clocking agent were traveling on a linear chain.

\chapter{Speed, entropy and energy} \label{chap:speedandentropy}
In this chapter we study the random variable \emph{time of flight speed} of the clocking agent on a linear chain. A quantitative assessment on the spreading rate of the wave packet associated to the clocking excitation makes it possible to define a quasi-optimal preparation scheme and  post measurement strategy \cite{defa06a}. We show, moreover, that the intrinsic uncertainty on the position of the cursor induces decoherence on the state of the register and that this decoherence is of Lindblad type \cite{lindblad76}.\\[5pt]
Throughout this chapter we will consider Hamiltonians of the form
\begin{equation}\label{eq:hamiltoniana2}
	H = -\frac{\lambda}{2} \sum_{x=1}^{s-1} U_x \ketbra{C(x+1)}{C(x)} + U_x^{-1} \ketbra{C(x)}{C(x+1)},
\end{equation}
where only \ves{U}{N-1} are assigned by the algorithm we are interested in whereas $U_N,\ldots, U_{s-1}$ are to be assigned as a part of the description of the clocking mechanism. For instance, we have seen in \sref{sec:relaxing} the case in which $N-s=\delta$, and  $U_N,\ldots, U_{s-1}=I_r$, the identity in \Hr,  and verified the role of the cursor sites $N,\ldots,s$  as a storage mechanism of the desired output state \ket{R(N)}.\\
An alternative point of view was taken in some of the numerical examples of \cite{defa04}, motivated by Grover's algorithm: one may suppose all of the $U_x$ to coincide, in such a context, with Grover's $estimation \cdot oracle$ step $G$, and study the effect of applying $G$ more than the optimal number $N-1$ of times.\\
Most of our numerical examples will refer in fact to the following particular instance (Toy model):
\begin{equation} \label{eq:hamiltoniana4}
		H_T = -\frac{\lambda}{2} \sum_{x=1}^{s-1} e^{-i\frac{\alpha}{2} \sigma_2} \ketbra{C(x+1)}{C(x)} + e^{i\frac{\alpha}{2} \sigma_2}\ketbra{C(x)}{C(x+1)}.
\end{equation}
which, with the right choice of the parameter $\alpha$, models the search of the ``needle in the haystack'' of Grover algorithm \cite{grover96,grover01}.
\section{Speed of computation} \label{sec:speed}
In Chapter \ref{chap:fqc} we have seen that a Feynman machine initialized (at time $t=0$) in the state
\begin{equation} \label{eq:condini1}
 \ket{R(1)} \otimes \ket{C(1)} = \ket{M_1} 
\end{equation}
evolves, under the Hamiltonian \eref{eq:hamiltoniana2}, into
\begin{equation} \label{eq:evoluzione1}
	\ket{M_1(t)} = e^{-iHt} \ket{M_1}= \sum_{x=1}^s c(t,x;s)\ket{R(x)} \otimes \ket{C(x)}
\end{equation}
where the register states \ket{R(x)} are defined as in \eref{eq:Rj}.\\
The observable $Q/t$ acquires thus the meaning of number of primitives per unit time applied to the initial condition \ket{R(1)} in the time interval $(0,t)$. In order to study the behavior over long intervals of time ($t \rightarrow +\infty$) of this observable in the case of a long computation ($s \rightarrow +\infty$) it is expedient to study its characteristic function
\begin{equation} \label{eq:fcaratteristica}
	\phi_{s,t}(z) = \bra{M_1(t)} \exp \left(iz \frac{Q}{t}\right) \ket{M_1(t)} = \sum_{x=1}^s |c(t,x;s)|^2 \exp \left(iz \frac{x}{t} \right),
\end{equation}
namely the Fourier transform of its probability distribution.
\begin{theorem}\label{th:funzionecaratteristica} \footnote{The same result (see also \cite{apolloni02}) has been obtained by N. Konno \cite{konno06} in a different context: there the author investigates the propagation of an excitation through a homogeneous tree, whereas here we deal with a linear chain. The obvious explanation of the coincidence of the results in two such different contexts comes from the projection mechanism explained in \sref{sec:synchro}.}
\begin{equation} \label{eq:limitecaratteristica}
\lim_{t \to +\infty} \lim_{s \to \infty} \phi_{s,t}(z) = \int_0^1 \frac{ 4 v^2}{\pi \sqrt{1-v^2}}.
\end{equation} 
\end{theorem}
\begin{proof}
The large $s$ behavior is easily studied by inserting the explicit integral representation of the $s \rightarrow +\infty$ limit of  \eref{eq:c} into \eref{eq:fcaratteristica}
\begin{eqnarray}
\lim_{s \to +\infty} c(t,x;s) &  = & 
\frac{2}{\pi} \int_{0}^\pi e^{i t  \cos(p) } \sin(p) \sin(x p) dp.
\end{eqnarray}
We thus obtain
\begin{equation}
	\lim_{t \to \infty} \sum_{x=1}^\infty e^{i z \frac{x}{t}} \left | \frac{2}{\pi} \int_{0}^\pi e^{i t  \cos(p) } \sin(p) \sin(x p) \ dp \right |^2.
\end{equation}
The $t \rightarrow +\infty$ limit is  studied by substituting the sum over $v=x/t$, step $1/t$,  appearing in \eref{eq:fcaratteristica} with an integral 
\begin{eqnarray}  \label{eq:int01}
 \lim_{t \to +\infty} \int_0^1 e^{i z v}  \left | \frac{2 \sqrt{t} }{\pi} \int_{0}^\pi e^{i t  \cos(p)} \sin(p) \sin(v t p) \ dp \right |^2 dv.
\end{eqnarray}
We evaluate the leading contributions to \eref{eq:int01}  by a standard stationary phase argument.\\
First of all we observe that 
\begin{eqnarray}
	 e^{i t  \cos(p)} \sin(p) \sin(v t p) & = & e^{i t  \cos(p)} \left ( \frac{e^{i v t p}+ e^{-i v t p} }{2i} \right) \sin(p) = \nonumber \\
	 & = & \frac{ e^{i t ( \cos(p)+vp)}+ e^{i t ( \cos(p) - v p )}}{2i}\sin(p). 
\end{eqnarray}
Since the equation
\[
\frac{d}{dp} ( \cos(p) + vp) = 0
\]
admits two distinguished solutions
\begin{eqnarray}
	p_1 & = & arcsin(v)  \\
	p_2 & = & \pi - arcsin(v)
\end{eqnarray}
whereas the equation
\[
\frac{d}{dp} ( \cos(p) - vp )= 0
\]
has no solution for $ 0 \leq v \leq 1$, the contributions to the limit
\begin{equation}
	\lim_{t \to +\infty} \left ( \int_0^\pi \frac{ e^{i t ( \cos(p)+vp)}{2i}}\sin(p) dp + \int_0^\pi \frac{ e^{i t ( \cos(p) - v p )}}{2i}\sin(p) \right) dp
\end{equation}
come only from the first term; since these contributions come from the region around $p_1$ and $p_2$,  the following equality holds
\begin{eqnarray} \label{eq:lim02}
\lim_{t \to +\infty} \int_0^\pi  \frac{ e^{i t ( \cos(p)+vp)}}{2i} dp  
	=\lim_{t \to +\infty}\sum_{j \in \{1,2\}} e^{i t ( \cos(p_j)+ v p_j)} \int_{-\infty}^{+\infty}  e^{-\frac{1}{2}  \cos(p_j)(p-p_j)^2 } dp.
\end{eqnarray}
We recall that
\begin{equation}
	\lim_{\sigma \to 0} \left ( \frac{1}{\sigma \sqrt{2 \pi}}\  e^{-\frac{(p-p_j)^2}{2\sigma^2}} \right ) = \delta_{p_j}(p)
\end{equation}
where $\delta_{p_j}$ is the Dirac-$\delta$ function concentrated around $p_j$.\\
The right hand side of equation \eref{eq:lim02} is then equivalent to
\begin{equation}
	\lim_{t \to +\infty} \sum_{j \in \{1,2\}} e^{i t ( \cos(p_j)+ v p_j)} \sqrt{\frac{2 \pi}{i t  \cos(p_j)}}.
\end{equation}
By inserting the last expression in \eref{eq:int01} we get
\begin{eqnarray}
\int_0^1 e^{izv} \frac{4 (v^2}{\pi \sqrt{1-v^2}}+ \lim_{t \to +\infty} \int_0^1 e^{izv} \frac{\cos((\cos(p_1)-\cos(p_2))+v(p_1-p_2)))}{\sqrt{\cos(p1) \cos(p_2)}} dv.
\end{eqnarray}
Since the equation
\begin{equation}
	\frac{d}{dv}(\cos(p_1)-\cos(p_2))+v(p_1-p_2) = 0
\end{equation}
has no solution for $0 \leq v \leq 1$, the second term vanishes and we get the result.
\end{proof}
As convergence in the sense of characteric functions implies convergence in the sense of cumulative distribution functions (\emph{convergence in law}),  we conclude that a ``long''  computation starting from the initial condition \eref{eq:condini1} proceeds ``in the long run'' at a rate of $V(M_1)$ steps per unit time (the unit of time having been set so that $\lambda=2$ ),  $V(M_1)$  being the random variable defined by having as its characteristic function the right hand side of  \eref{eq:limitecaratteristica}; equivalently stated it has probability density 
\begin{equation} \label{eq:densitav}
	f_{V(M_1)}(v)= I_{(0,1)}(v)  \frac{4 v^2}{\pi \sqrt{1-v^2}}
\end{equation}
Here and in what follows we denote by $I_{(a,b)}$the indicator function of an interval $(a,b)$:
\begin{equation}
	I_{(a,b)}(x)= \left \{ 
	\begin{array}{c r}
	1 & \mbox{if $x \in (a,b)$}\\
	0 & \mbox{otherwise.} 
	\end{array}\right.
\end{equation}
The mean value
\begin{equation} \label{eq:EM1}
	E(V(M_1))= \int_0^1 v\; f_{V(M_1)}(v) dv = \frac{8}{3 \pi}
\end{equation}
and the variance
\begin{equation} \label{eq:varM1}
	var(V(M_1))=E\left((V(M_1))^2\right)-E\left(V(M_1)\right)^2=\frac{3}{4}- \left(\frac{8}{3 \pi} \right)^2
\end{equation}
are then easy to compute from \eref{eq:densitav}.\\[5pt]
It is worth extending our analysis to more general initial conditions; for any positive integer $x_0$, a state such as 
\begin{equation} \label{eq:condini2}
	\ket{M_{x_0}}=\ket{R(x_0)} \otimes \ket{C(x_0)}
\end{equation}
having at a certain instant the cursor in $x_0$, evolves under the action of the Hamiltonian \eref{eq:hamiltoniana2} as
\begin{equation}
	\ket{M_{x_0}(t)} = e^{-i H t} \ket{M_(x_0)} = \sum_{x=1}^s c_{x_0}(t,x,s) \ket{R(x)} \otimes \ket{C(x)}.
\end{equation}
With the same techniques used in the proof of Theorem \ref{th:funzionecaratteristica} we can prove that
\begin{theorem}
The random variable \emph{time of flight} speed $V(M_{x_0})$ has a characteristic function 
\begin{equation} \label{eq:funzcarx0}
\lim_{t \to _\infty} \lim_{s \to \infty} \bra{M_{x_0}(t)}e^{i z \frac{Q}{t}} \ket{M_{x_0}(t)} = \frac{4}{\pi} \frac{\sin(x_0 \ arcsin(v)^2)^2}{\sqrt{1-v^2}}.
\end{equation}
\end{theorem}
The cumulative distribution function of  $V(M_{x_0})$ is consequently
\begin{eqnarray}
F_{V(M_{x_0})}(v) & \equiv & Prob(V(M_{x_0}) \leq v) =  \\
& = & I_{(0,1)}(v)\left( \frac{2 \arcsin(v)}{\pi} - \frac{\sin(2 x_0 \arcsin(v))}{\pi x_0}\right)+I_{(1,+\infty)}(x) \nonumber
\end{eqnarray}
corresponding to an expectation value
\begin{equation} \label{eq:EMx0}
	E(V(M_{x_0}))=\frac{8}{4\pi-\pi/x_0^2}.
\end{equation}
Comparison between \eref{eq:EM1} and \eref{eq:EMx0} shows the effect of a measurement of $Q$. If, at a given $t$, $Q$ is measured and the result $x_0$ is found, then the state \eref{eq:evoluzione1}, into which the initial condition \eref{eq:condini1} has evolved, collapses into the state \eref{eq:condini2}. From this moment on the computation proceeds at the mean rate \eref{eq:EMx0}: for large values of $t$, reading the \emph{clock} is likely to reduce the speed of further computation by a factor $3/4$ (without, because of \eref{eq:ekrj}, altering its correctness).  
%
%
%

\section{Entropy}
Motivated by the experience gained under the particular initial conditions \eref{eq:condini1} and \eref{eq:condini2} we define, for any (unentangled) initial condition of the form (for fixed $\epsilon \geq 1$)
\begin{equation} \label{eq:condini3}
	\ket{R; \psi_0}= \ket{R} \otimes\sum_{x=1}^{\epsilon} \psi_0(x) \ket{C(x)},
\end{equation}
the ``time-of-flight speed'' \cite{feyn65} of computation in the state $\psi_0$ as the random variable $V(\psi_0)$ having characteristic function
\begin{equation} \label{eq:fcaratteristica2}
	\phi_{V(\psi_0)}(z) = \lim_{t \to +\infty} \lim_{s \to +\infty} \bra{R; \psi_0} e^{itH}\exp \left(iz \frac{Q}{t}\right) e^{-itH}\ket{R; \psi_0}.
\end{equation}
The above limit is easily shown to exist by the techniques outlined in the previous section; it corresponds to the probability density
\begin{equation} \label{eq:densitapsi0}
	f_{V(\psi_0)}(v)=I_{(0,1)}(v) \frac{|\Psi(\arcsin(v))|^2+|\Psi(\pi-\arcsin(v))|^2}{\sqrt{1-v^2}}
\end{equation}
where
\begin{equation}
	\Psi(p)=\sqrt{\frac{2}{\pi}} \sum_{x=1}^{\epsilon} \sin(px) \psi_0(x)
\end{equation}
is the sine transform of the initial state of the cursor.\\
The observable $Q$ retains in this context the meaning of relational time \cite{gambini04} in the sense that,  \underline{given} that at any \emph{parameter time} $t$ the cursor is found at $x$, it is \underline{then} certain that the register is found in the state $U_{x-1}\cdot \ldots \cdot U_2 \cdot U_1 \ket{R}$.\\
In reading the output at any time $t$, namely in  the measurement of any, however carefully chosen,  observable of the register, there is an intrinsic uncertainty corresponding to the uncertainty about how far the computation has proceeded. The fact that $Q/t$ has a non trivial limit in law means that the leading term of the variance of $Q$ is proportional to $t^2$ and therefore that the uncertainty increases with $t$. This section is devoted to the examination of an example in which the notion of ``the most careful choice'' of the observable to read on the register can be made precise and shown to be pertinent to the algorithm considered.\\
We consider for the moment the initial condition \ket{M_1} given in \eref{eq:condini1} and its  time evolution \ket{M_1(t)} described in \eref{eq:evoluzione1}. More general initial conditions of the form  \eref{eq:condini3} will be examined in the next section.\\
Call
\begin{equation}
	\rho_m(t)=\ketbra{M_1(t)}{M_1(t)}
\end{equation}
the density matrix of the machine at time $t$.\\
By taking the partial trace $Tr_{\Hc}(\rho_m(t))$ with respect to the cursor degrees of freedom, we get the density matrix $\rho_r(t)$ of the register:
\begin{equation} \label{eq:matricereg}
	\rho_r(t)=\sum_{x=1}^s |c(t,x;s)|^2 \ketbra{R(x)}{R(x)}.
\end{equation}
Call $\lambda_j(t)$ the positive eigenvalues of $\rho_r(t)$  and \ket{b_j(t)} the corresponding eigenstates. A simple computation, amounting to the Schmidt decomposition \cite{peres93} of the state \eref{eq:evoluzione1}, shows, then, that the density matrix of the cursor is given by
\begin{equation} \label{eq:matricecur}
	\rho_c(t)= \sum_j \lambda_j(t) \ketbra{d_j(t)}{d_j(t)}
\end{equation}
where
\begin{equation}
	\ket{d_j(t)}= \frac{1}{\sqrt{\lambda_j(t)}} \sum_{x=1}^s c(t,x;s) \braket{b_j(t)}{R(x)}\> \ket{C(x)}.
\end{equation}
Because of \eref{eq:matricecur} and of the orthonormality of the states \ket{d_j(t)}, the von Neumann entropy of the register and also of the cursor is then given by
\begin{equation}
	S(\rho_c(t))=- \sum_j \lambda_j(t) \ln \lambda_j(t)= S(\rho_r(t)).
\end{equation}
We observe that, as \eref{eq:matricereg} shows, the von Neumann entropy of each subsystem does depend on the algorithm being performed. It is, indeed, \underline{only} under the hypothesis, nowhere made above, that the states  \ket{R(x)} are orthonormal that \eref{eq:matricereg} is the spectral decomposition of $\rho_r(t)$  (the von Neumann entropy becoming in this case equal to the Shannon entropy of the distribution of $Q$).\\
We focus our attention, in what follows, on our Toy model \eref{eq:hamiltoniana4}, in which the register is a single spin $1/2$ system. We indicate by $\underline{e_1},\underline{e_2},\underline{e_3}$ the versors of the three coordinate axes to which the components $\underline{\sigma}=(\sigma_1,\sigma_2,\sigma_3)$ of such a spin are referred.\\
In the basis \ket{\sigma_3=\pm 1}, the density operator $\rho_r(t)$ will be represented by the matrix
\begin{equation} \label{eq:rhor}
\rho_r(t)=\frac{1}{2} \left (
\begin{array}[pos]{c c}
1+s_3(t) & s_1(t)-i\>s_2(t)\\
s_1(t) + i\>s_2(t)  & 1-s_3(t)
\end{array}
\right )
\end{equation}
where
\begin{equation}
s_j(t)=Tr \left ( \rho_r(t) \cdot \sigma_j \right ),\;j=1,2,3.
\end{equation}
Equivalently stated, the \emph{Bloch representative} of the state $\rho_r(t)$ is given by the three-dimensional real vector
\begin{equation}
\underline{s}(t) = \sum_{x=1}^s  \left | c(t,x;s)  \right | ^2 \bra{R(x)} \underline{\sigma} \ket{R(x)}.
\end{equation}
We shall assume, in what follows, that the initial state of the cursor is \ket{C(1)} and that the initial state of the register is of the form
\begin{equation} \label{eq:condinigrover}
	\ket{R(1)}=\cos \left ( \frac{\theta}{2} \right ) \ket{\sigma_3=+1} + \sin \left(  \frac{\theta}{2} \right ) \ket{\sigma_3=-1}
\end{equation}
namely the eigenstate belonging to the eigenvalue $+1$ of $\underline{n}(1) \cdot \underline{\sigma}$, with
\begin{equation}
\underline{n}(1)= \underline{e}_1 \sin\theta + \underline{e}_3 \cos \theta. 
\end{equation}
We wish to remark that the above example captures the geometric aspects not only of such simple computational tasks as $NOT$ or $\sqrt{NOT}$ (viewed as rotations of an angle $\pi$ or $\pi/2$ respectively, decomposed into smaller steps of amplitude $\alpha$) but also of Grover's quantum search \cite{grover96}. If, indeed, the positive integer $\mu$ is the length of the marked binary word to be retrieved, and we set
\begin{equation}
\chi(\mu)=\arcsin(2^{-\frac{\mu}{2}})
\end{equation}
and
\begin{equation} \label{eq:theta}
\theta= \pi-2\>\chi(\mu)
\end{equation}
then the state \eref{eq:condinigrover} correctly describes the initial state $\ket{\iota}$ of the quantum search as having a component $2^{- \mu/2}$ in the direction of the target state, here indicated by $\ket{\omega}=\ket{\sigma_3=+1}$, and a component $\sqrt{1-2^{-\mu}}$ in the direction of the flat superposition, here indicated by $\ket{\sigma_3=-1}$, of the $2^\mu-1$ basis vectors orthogonal to the target state.
In this notations, if
\begin{equation} \label{eq:alpha}
\alpha=-4\>\chi(\mu),
\end{equation}
then the unitary transformation $\exp(-i \> \alpha\> \sigma_2/2)$ corresponds to the product $B\cdot A$ of the \emph{oracle} step
\begin{equation} \label{eq:oracle1}
A=I_r-2 \> \ketbra{\omega}{\omega}
\end{equation}
and the \emph{estimation} step
\begin{equation}\label{eq:estimation}
B=2\>\ketbra{\iota}{\iota}-I_r.
\end{equation}
Following  the beautifully pedagogical approach of Jozsa \cite{jozsa99}, we observe that the operator $A$ corresponds to a reflection in the hyperplane orthogonal to $\ket{\omega}$ and $B$ corresponds to a reflection in the hyperplane orthogonal to $\ket{\iota}$, with a minus sign.
\begin{lemma} \label{lemma:reflections}
If \ket{\chi} is any state in \Hr, then $I_{\ket{\psi}} = I_r-2\ketbra{\psi}{\psi}$ preserves the subspace $\mathcal{S}$ of \Hr spanned by \ket{\chi} and \ket{\psi}.
\end{lemma}
\begin{proof}
Geometrically, $\mathcal{S}$ and the mirror hyperplane are orthogonal to each other (in the sense that the orthogonal complement of either subspace is contained in the other subspace) so the reflection preserves $\mathcal{S}$. Alternatively in terms of algebra, \eref{eq:estimation} shows that $B$ takes \ket{\iota} to $-\ket{\iota}$ and, for any \ket{\omega}, it adds a multiple of \ket{\omega} to\ket{\iota}. Hence any linear combination
is mapped to a linear combination of the same two states.
\end{proof}
Since Grover's algorithm consists  of the iterated application of two reflections (a reflection in the hyperplane orthogonal to the target state \ket{\omega} and  a reflection in the hyperplane orthogonal to the initial state \ket{\iota}), it preserves the space $\mathcal{S}$ spanned by \ket{\omega} and \ket{\iota}.\\
Now we may introduce a basis $\{\ket{e_1},\ket{e_2}\}$
into $\mathcal{S}$ such that \ket{\omega} and \ket{\iota}, up to an overall phase, have real coordinates. Indeed
choose $\ket{e_1}=\ket{\iota}$ so \ket{\iota} has coordinates $(1, 0)$. Then $e^{i \xi} \ket{\omega} = a \ket{e_1}+b \ket{e_2}$ where \ket{e_2}, orthogonal to \ket{e_1}, has still an overall phase freedom. We can thus choose $\xi$ to make $a$ real and the phase of \ket{e_2} to make $b$ real. In this basis, then, since \ket{\omega} and \ket{\iota} have real coordinates, the operators $A$ and $B$ when acting on $\mathcal{S}$, are also described by real 2 by 2 matrices: they are just the real 2 dimensional reflections in the lines perpendicular \ket{\omega} and \ket{\iota}. Finally we have the following
\begin{theorem}
Let $M_1$ and $M_2$ be two mirror lines in the Euclidean plane $\mathbb{R}$ intersecting at a point $O$ and let $\alpha$ be the angle in the plane from $M_1$ to $M_2$. Then the operation of reflection in $M_1$ followed by reflection in $M_2$ is just rotation by angle $2 \alpha$ about the point
$O$.
\end{theorem}
It is having in mind the connection with Grover's algorithm that, for the sake of definiteness, in the examples that follow we are going to consider the one-parameter family of models, parametrized by the positive integers $\mu$, corresponding to the choice \eref{eq:theta} and \eref{eq:alpha} of the parameters $\theta$ and $\alpha$ and to the choice $s=2^\mu+1$ of the number of cursor sites, corresponding to the possibility of performing up to an exhaustive search.\\
In the example defined by the above conditions it is
\begin{equation}
\bra{R(x)}\> \underline{\sigma} \ket{R(x)} = \sin \left ( \theta+(x-1)\alpha \right)\> \underline{e}_1+\cos \left ( \theta+(x-1)\alpha \right)\> \underline{e}_3
\end{equation}
and, therefore,
\begin{equation} 
\underline{s}(t) = \sum_{x=1}^s \left |c(t,x;s) \right | ^2  \left( \sin \left ( \theta+(x-1)\alpha \right)\> \underline{e}_1+\cos \left ( \theta+(x-1)\alpha \right)\> \underline{e}_3 \right ).
\end{equation}
Figure 1 presents, inscribed in the unit circle, a parametric plot of $(s_1(t),s_3(t))$ under the above assumptions .
\begin{figure}[t] 
	\centering
		\includegraphics{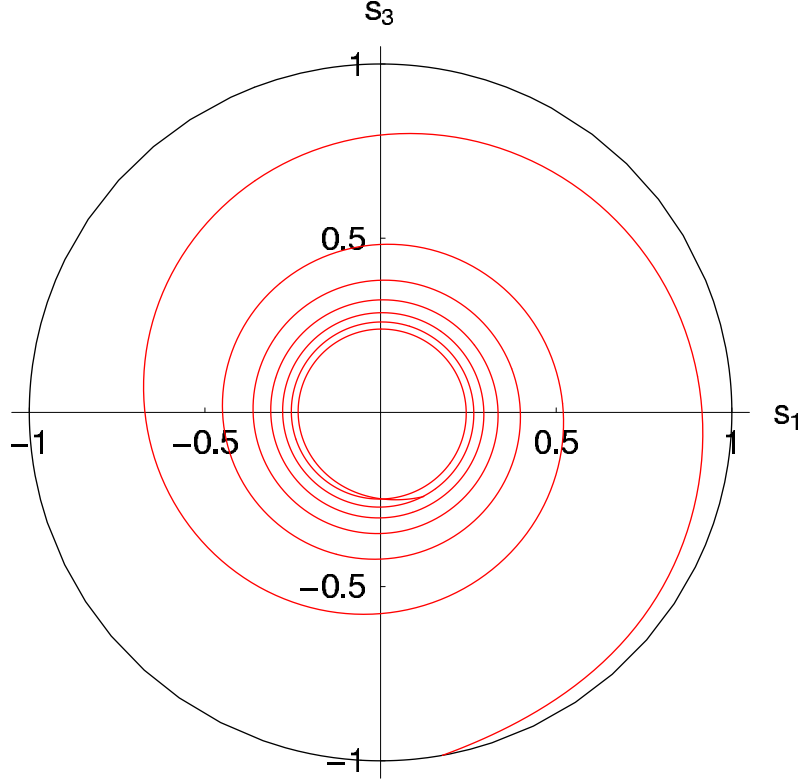}
\caption{A parametric plot of $\left (s_1(t),s_3(t)\right )$ for \mbox{$0 \leq t < s$}, $\lambda=1$. The choice $\mu=7,\;\chi= \arcsin(1/2^{\mu/2}),\; s=2^\mu +1,\; \alpha = -4 \chi,$ \mbox{$ \theta= \pi -2 \chi$} of the parameters is motivated by the connection with Grover's algorithm. Only the initial state lies on the unit circumference, the locus of pure states.}
\label{fig:figura1}
\end{figure}
It is convenient to describe the Bloch vector $\underline{s}(t)= s_1(t)\> \underline{e}_1+ s_3(t)\> \underline{e}_3$ in polar coordinates as
\begin{equation} \label{eq:toybloch}
\begin{array}{lcr}
s_1(t)=r(t)\sin{\gamma(t)},
& &
s_3(t)=r(t)\cos{\gamma(t)}.
\end{array}
\end{equation}
A very simple approximate representation of $\underline{s}(t)$ becomes then possible:
\begin{eqnarray}
r(t)e^{i \gamma(t)} & = & \sum_{x=1}^s |c(t,x;s)|^2 \exp(i(\theta+(x-1)\alpha)) = \nonumber \\
& = & \exp(i (\theta - \alpha))\sum_{x=1}^s |c(t,x;s)|^2 \exp(i\alpha x) \approx \nonumber \\
& \approx &  \exp(i (\theta - \alpha)) E\left( \exp(i \alpha \lambda t V(M_1(0))) \right)
\end{eqnarray}
The last step, legitimate for $1 << \lambda\; t <s$, requires only the explicit computation of the characteristic function corresponding to the probability density \eref{eq:densitav}, which leads to
\begin{equation}
	r(t)e^{i \gamma(t)} \approx \frac{2 \exp(i(\theta-\alpha))}{T}\left((J_1(T)- T \, J_2(T)+i (T \,H_0(T)-H_1(T))) \right)
\end{equation}
where $J_k$ and $H_k$ are, respectively, Bessel functions and Struve functions \cite{bessel}, and $T=\alpha \lambda t$.\\
The time evolution of the register subsystem is summarized by the Lindblad equation \cite{gorini76,lindblad76}
\begin{equation} \label{eq:lindblad}
	\frac{d \rho_r(t)}{dt} = -\frac{i}{2} \frac{d \gamma(t)}{dt} \left[ \sigma_2, \rho_r(t) \right]+ \frac{1}{4} \frac{d \ln r(t)}{dt} \left[\sigma_2 ,\left[ \sigma_2, \rho_r(t) \right] \right].
\end{equation}
The commutator term $\left[ \sigma_2, \rho_r(t) \right]$ describes the Hamiltonian part of the dynamics (after all we are considering a rotation about the $x_2$ axis); the double commutator $\left[\sigma_2 ,\left[ \sigma_2, \rho_r(t) \right] \right]$ describes, in much the same sense as equation 2.8 of \cite{milburn91}, the decohering effect of this rotation being administered by the cursor in discrete steps at random times.\\
The eigenvalues of $\rho_r(t)$ can be written as
\begin{equation} \label{eq:eigenvalues1}
\begin{array}{lcr}
\lambda_1(t)=\frac{1}{2} (1+r(t)), 
& &
\lambda_2(t)=\frac{1}{2} (1-r(t)).
\end{array}
\end{equation}
The von Neumann entropy $S\left( \rho_r(t) \right)$ is therefore
\begin{equation} \label{eq:vne}
S\left( \rho_r(t) \right) = - \frac{1+r(t)}{2} \ln \frac{1+r(t)}{2} - \frac{1-r(t)}{2} \ln \frac{1-r(t)}{2}.
\end{equation}
An example of its behaviour is shown in \fref{fig:figura2}.
\begin{figure}[t] 
	\centering
		\includegraphics{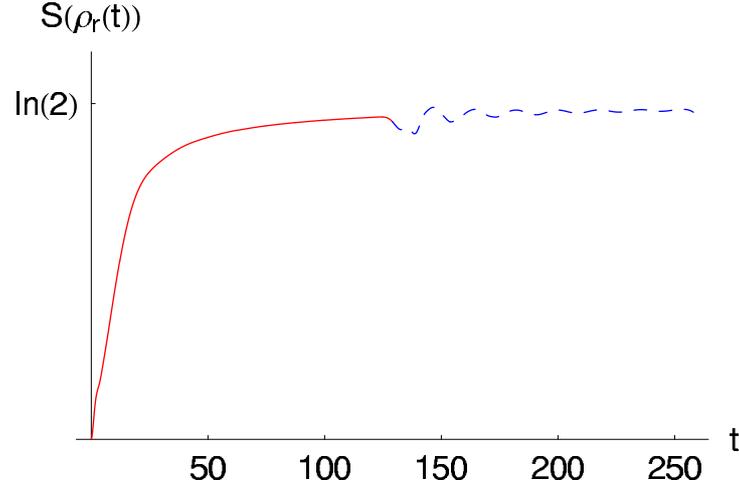}
		\caption{The von Neumann entropy of the register as a function of time, 
for the same model as in \fref{fig:figura1},  for $ 0 \leq t <s$  (solid line) and for  $ s \leq t < 2s$  (dashed line).}
\label{fig:figura2}
\end{figure}
The  eigenvectors corresponding to the eigenvalues \eref{eq:eigenvalues1} are, respectively
\begin{equation}
	\begin{array}{lcr}
		\ket{b_1(t)}= 
		\left (
				\begin{array}{c}
				\cos(\gamma(t)/2)\\
				\sin(\gamma(t)/2)
				\end{array}
		\right ),
		& &
		\ket{b_2(t)}= 
		\left (
				\begin{array}{c}
				-\sin(\gamma(t)/2)\\
				\cos(\gamma(t)/2)
				\end{array}
		\right ).
	\end{array}
\end{equation}
It is to be stressed that, at each time $t$, the projector \ketbra{b_1(t)}{b_1(t)} is, among the projectors on the state space of the register, the one having in the state  $\rho_r(t)$ the greatest probability of assuming, under measurement, the value $1$. Thus, the most careful choice (the one affected by minimum uncertainty) of the observable to read on the register at time $t$ is the projector \ketbra{b_1(t)}{b_1(t)}. In the case of Grover's algorithm one must measure the projector $\ketbra{\omega}{\omega}=\ketbra{\sigma_3=+1}{\sigma_3=+1}$ (and one easily can, because of the kickback mechanism analyzed, for instance, in \cite{defa04}) and has the freedom of choosing the time $\tau$ at which to perform the measurement. The best choice is therefore such that $\ket{b_1(\tau)}=\ket{\omega}$  (in our notational setting, $\tau$ is the time at which the helix of \fref{fig:figura1} crosses for the first time the positive $s_3$ axis). In spite of the fact of being now in the most favorable setting, one has, nevertheless, a deficit  $1-\lambda_1(\tau)$ in the probability of finding the target state.\\
As \fref{fig:figura3} shows,
\begin{figure}[!ht]
	\centering
		\includegraphics{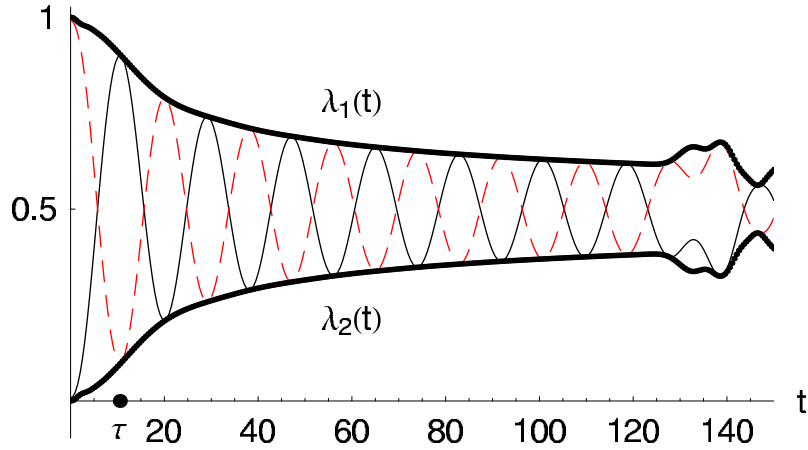}
		\caption{The same model as in \fref{fig:figura1} and \fref{fig:figura2}; $0 \leq t <1.2\,s$. The thin solid line is a graph of $Tr(\rho_r(t)\cdot (I_r+\sigma_3)/2)$, the probability of observing the target state $\ket{\omega}=\ket{\sigma_3=+1}$ in the example of Grover's algorithm. The dashed line is a graph of $Tr(\rho_r(t)\cdot (I_r-\sigma_3)/2)$, the probability of observing the ``undesired'' output $\ket{\sigma_3=-1}$. The upper and lower bounds on the probability of observing the target state are represented by the thick solid lines $\lambda_1(t)$ and $\lambda_2(t)$.}
	\label{fig:figura3}
\end{figure}
there are successive instants of time at which the probability of successful retrieval has a local maximum (a remnant of the periodic nature of Grover's algorithm when applied by an outside macroscopic agent) but the heights of these successive maxima form a sequence having a decreasing trend.\\
Further insight into our toy model is gained by examining the $t$ dependence of $E(Q(t)) = \bra{M_1(t)}Q \ket{M_1(t)}$ and of the angle of polarization $\gamma(t)$. The example of \fref{fig:figuraeq}.a suggests that the mean value of speed derived from asymptotic considerations correctly describes the average behavior of the ``clocking'' subsystem also for finite values of  $0<t<s$. As \fref{fig:figuraeq}.b shows, the ``clocked'' subsystem system $\bsigma$ is, in turn, driven, on the average, into uniform rotational motion.
\begin{figure}[!h]
	\centering
	\subfigure[]{\includegraphics[width=7cm]{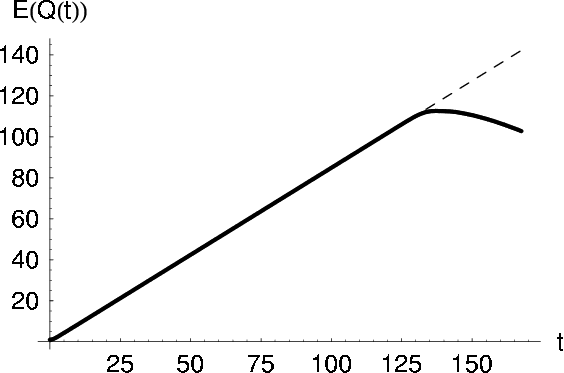}}
	\hspace{1cm}
	\subfigure[]{\includegraphics[width=7cm]{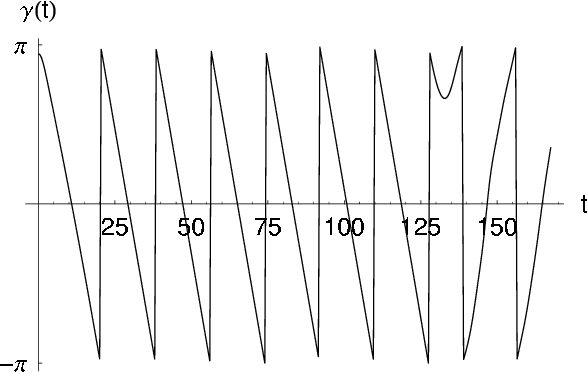}}
	\caption{Same parameters as in \fref{fig:figura1}. (a) $E(Q(t))$ (solid line) compared with the dashed straight line of slope $8/(3\pi)$. (b) The polar angle  $\gamma(t)$ of the Bloch vector \eref{eq:toybloch} as a function of $t$.}
\label{fig:figuraeq}
\end{figure}
\\Our model is so simple that we can explicitly study how the above semiclassical picture (in which the time parameter $t$ acquires operational meaning from its linear relation with mean values of configurational observables of clocking and clocked subsystem) is distorted by a measurement performed on either subsystem. The observations made at the end of the previous section about the effect of \emph{reading the clock} can indeed be complemented by the examination of the effect of \emph{reading the register}.\\
Suppose that the observable $\sigma_3$ has been measured at time $\tau$ and the result $+1$ has been found: the Bloch diagram of \fref{fig:figura20}.a shows then that the evolution of the register proceeds in much the same way as in the undisturbed situation of \fref{fig:figura1} (with the only obvious difference that the post-measurement initial condition  \ket{b_1(\tau)} lies  on the unit circumference).
\begin{figure}[!h]
\subfigure[]{\includegraphics[width=50mm]{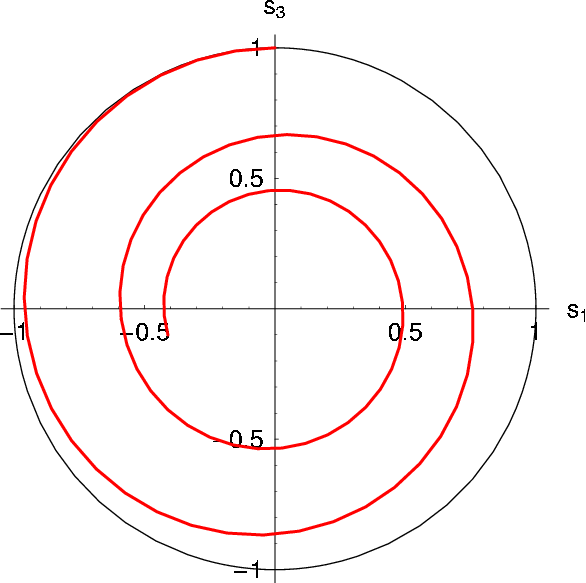}}
\hspace{2cm}
\subfigure[]{\includegraphics[width=50mm]{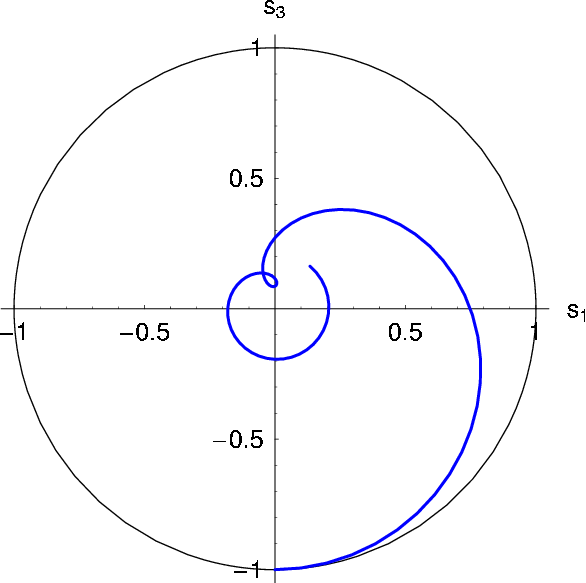}} \\
\centering
\subfigure[]{\includegraphics[width=79mm]{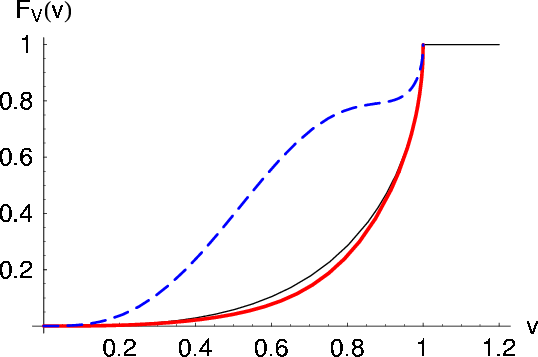}} 
\caption{Same parameters as in \fref{fig:figura1}; measurement of the observable $\sigma_3$  at time $\tau$. Frames (a) and (b) represent the evolution for $\tau < t \leq 4\tau$ of the Bloch vector when measurement returns $+1$ and $-1$ respectively. Frame (c) represents the cumulative distribution functions of the speed $V$ when the results $+1$ (solid thick line) and $-1$ (dashed line) have respectively been found; the solid thin line represents the c.d.f. of $V$ in case of no measurement.}
\label{fig:figura20}
\end{figure}
If, instead, the result $-1$ has been found (\fref{fig:figura20}.b) the post-measurement evolution of the register is completely different from the unperturbed one.\\
We conclude this section with an example of the insight that the time evolution of $S(\rho_r(t))$ can give on the algorithm \ves{U}{s-1} being performed by the machine. Suppose of using, instead of the assignment \eref{eq:hamiltoniana4} of the primitive steps, $U_1=U_2=\ldots=U_{s-1}=\exp(-i \alpha \sigma_2/2)$, the alternative assignment
\begin{equation} \label{eq:Ux}
	U_x= \left \{ 
	\begin{array}{c r}
	A & \mbox{for odd $x$}\\
	B & \mbox{for even $x$} 
	\end{array}\right.
\end{equation}
\begin{figure}[!ht]
\subfigure[]{\includegraphics[width=69mm]{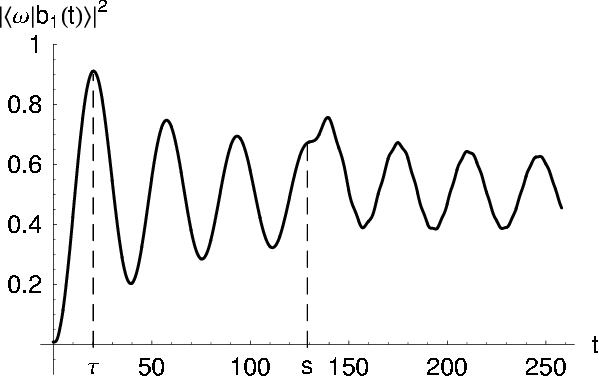}}
\subfigure[]{\includegraphics[width=69mm]{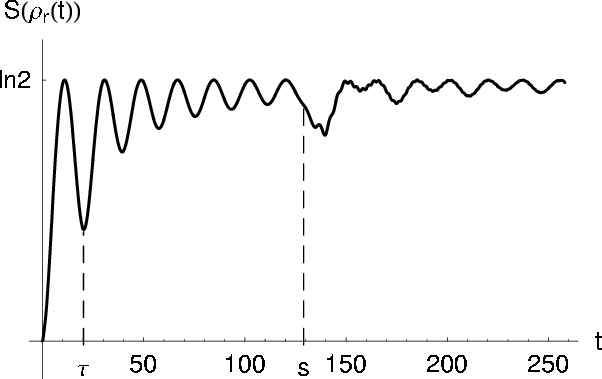}} \\
\centering
\subfigure[]{\includegraphics[width=59mm]{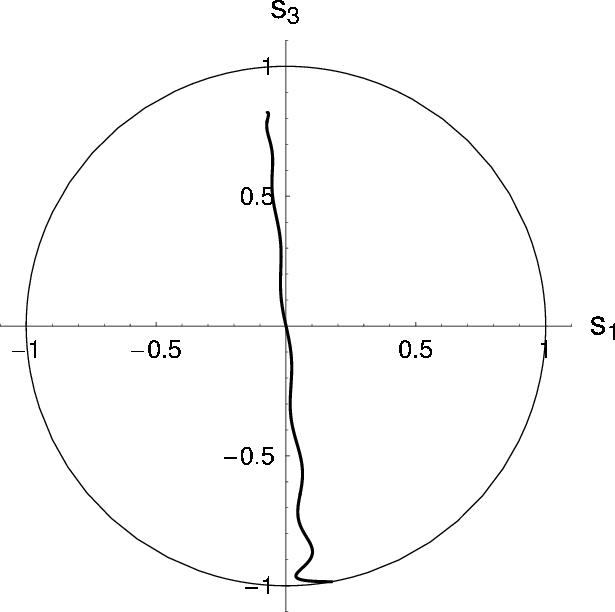}} 
\caption{Same parameters as in the previous figures; $U_x$ as in \eref{eq:Ux}. The Bloch diagram refers to the time interval $(0,\tau)$ needed to reach the first maximum in probability. We observe that the state of the register passes through the maximally mixed state before cooling down toward the target state.}
\label{fig:figura4}
\end{figure}
where $A$ and $B$ are given by \eref{eq:oracle1} and \eref{eq:estimation}. \Fref{fig:figura4} gives, for this example, a full account of the diffusive character \cite{grover01} of Grover's quantum search: the first maximum of the probability of finding the target state (\fref{fig:figura4}.a) is reached in correspondence of the first local minimum of entropy (\fref{fig:figura4}.b): that the search has gone, before this instant, through a local \underline{maximum} of entropy is shown with particular evidence by the Bloch diagram of \fref{fig:figura4}.c.
\section{Energy} \label{sec:energy}
With reference, for definiteness, to the example of figure 3, call $\tau$ the instant of time at which the probability  $Tr(\rho_r(t)\ \ketbra{\omega}{\omega})$ reaches its first  and absolute maximum. We recall that, in the above example, the target state  $\ket{\omega}$ is taken to be the ``up'' state \ket{\sigma_3=+1} of the register.\\
The whole point of the analysis of the previous section is that $\lambda_1(\tau)$ is \emph{strictly} \mbox{smaller than $1$.}
This amounts, in turn, to a deficit  $1-\lambda_1(\tau)$  in the probability of finding the target state. This deficit is not, in itself, a strong limitation in a quantum search algorithm, because we can in principle identify the right target through a majority vote among a ``gas'' of a large number $N$  of machines. The trouble is that  if we want to use the same machines once more, we need to purify the ``gas'' of registers from the fraction $\lambda_2(\tau)$  of them which have collapsed into the wrong state: standard thermodynamic reasoning \cite{landauer61}  shows then that this requires the removal from the gas, supposing a heat reservoir at temperature $T$ is available, of an amount of heat of $N k_B T S(\rho_r(\tau))$, $k_B$ being Boltzmann's constant.\\
We wish, in this section, to supplement the above considerations with an explicit description of the post-measurement state of the machine, showing, in particular, the effect \emph{onto the clock} of the act of reading the register\cite{defa06b}.\\[5pt]
Suppose that at the optimally chosen instant $\tau$, at which it is $\gamma(\tau)=0$, while the machine is in the state  \ket{M(\tau)},  a measurement of the projector $(I_r+\sigma_3)/2$ is performed.\\
If the measurement gives the result $1$, then the state \ket{M(\tau)} collapses to
\begin{equation}\label{eq:M1}
	\ket{M_1(\tau)}=\ket{\sigma_3=+1} \otimes \frac{1}{\sqrt{\lambda_1(\tau)}} \sum_{x=1}^s c(\tau,x;s)\cos((\theta+(x-1)\alpha)/2)\ket{C(x)}.
\end{equation}
If, instead, the measurement gives the result $0$, then the state \ket{M(\tau)} collapses to
\begin{equation}\label{eq:M2}
	\ket{M_2(\tau)}=\ket{\sigma_3=-1} \otimes \frac{1}{\sqrt{\lambda_2(\tau)}} \sum_{x=1}^s c(\tau,x;s)\sin((\theta+(x-1)\alpha)/2)\ket{C(x)}.
\end{equation}
\Fref{fig:VNSh}(a) and \fref{fig:VNSh}(b) show the probability distributions
\begin{eqnarray}
P_1(x,\tau)& = & \left| \left(c(\tau,x;s) \cos((\theta+(x-1) \alpha)/2)\right) \right|^2 / \lambda_1(\tau)\\
P_2(x,\tau)& = & \left| \left(c(\tau,x;s) \sin((\theta+(x-1) \alpha)/2)\right) \right|^2 / \lambda_2(\tau)
\end{eqnarray}
of the observable $Q$ (position of the cursor)  in the states  \ket{M_1(\tau)} and \ket{M_2(\tau)}, respectively.
\begin{figure}[htbp]
	\centering
		\includegraphics[width=55mm]{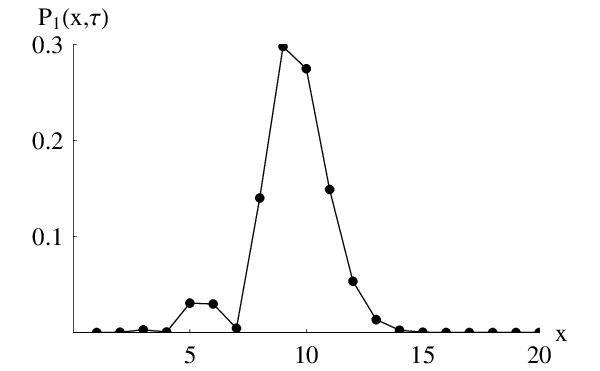}
		\hspace{1cm}
		\includegraphics[width=55mm]{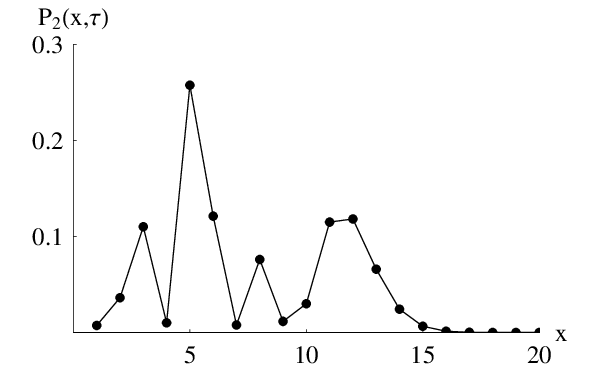}
		\begin{picture}(230,10)
			\put(20,0){(a)}
			\put(200,0){(b)}
		\end{picture}
	\caption{Figures (a) and (b) represent, for the same choice of parameters as in figure \ref{fig:figura1}, respectively the probabilities $P_1(x,\tau)$ and   $P_2(x,\tau)$ as  functions of $x$.}
	\label{fig:VNSh}
\end{figure}
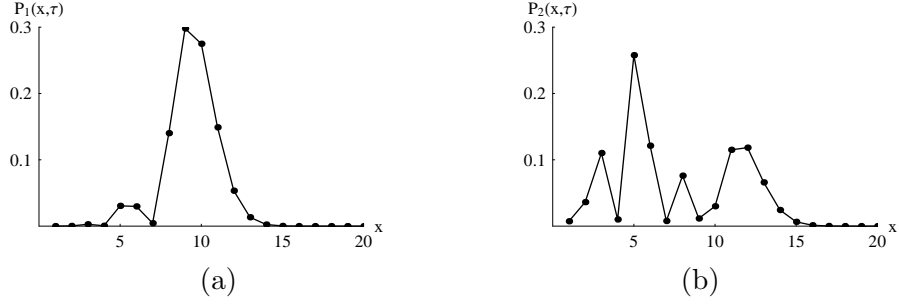
\Fref{fig:distribuzioneenergia}(a)and \fref{fig:distribuzioneenergia}(b) show the probability distributions of the observable $H$ (energy of the machine)  in the states \ket{M_1(\tau)} and \ket{M_2(\tau)}, respectively.
\begin{figure}[htbp]
	\centering
		\includegraphics[width=55mm]{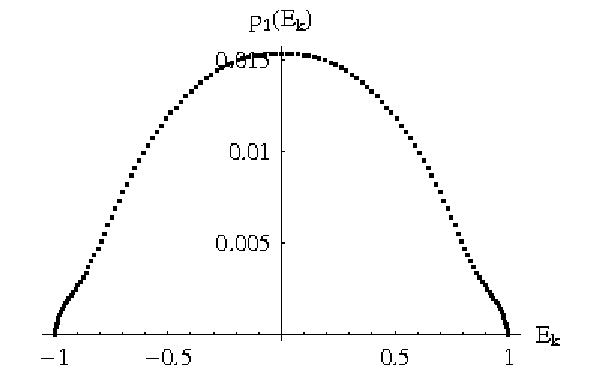}
		\hspace{1cm}
		\includegraphics[width=55mm]{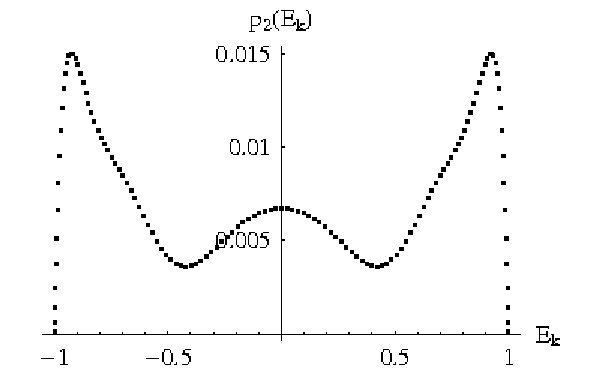}
		\begin{picture}(230,10)
			\put(20,0){(a)}
			\put(200,0){(b)}
		\end{picture}
		\caption{Figures (a) and (b) represent the probability distribution $p_1(E_k)$ and $p_2(E_k)$ of the energy $H$ in the state \ket{M_1(\tau)} and \ket{M_2(\tau)} respectively.}
	\label{fig:distribuzioneenergia}
\end{figure}
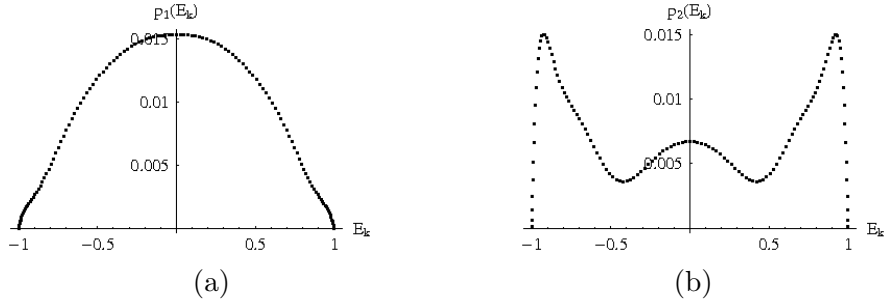
The two energy distributions of figures 5 are easily derived from the fact that the Hamiltonian $H$ defined in (4) has the eigenvalues
\begin{equation}\label{eq:eigenvaluesH}
	E_k= -\lambda\; \cos(\vartheta(k;s)),\ k=1,\ldots,s;
\end{equation}
each doubly degenerate, an orthonormal basis in the eigenspace belonging to the eigenvalue $E_k$ being given, for instance, by the two eigenvectors
\begin{equation}
	\ket{E_k;\sigma_2=\pm 1}= \ket{\sigma_2=\pm 1} \otimes \sum_{x=1}^s v_k(x) \exp(\mp i \alpha (x-1)/2)\ket{C(x)},
\end{equation}
where
\begin{equation}
	v_k(x)=\sqrt \frac{2}{s+1} \sin(x\;\vartheta(k;s)).
\end{equation}
This leads to the explicit expressions
\begin{equation}
	p_j(E_k)=\sum_{\eta=\pm 1} \left|\braket{M_j(\tau)}{E_k;\sigma_2=\eta} \right|^2, \ j=1,2.
\end{equation}
Figures \ref{fig:distribuzioneenergia} and \ref{fig:VNSh} show that a collection of identically prepared and independently evolving machines becomes in fact, under the operation of reading the register at time $\tau$, a mixture of two distinct ``molecular'' species, ``1'' (present in a concentration $\lambda_1(\tau)$ ),  and ``2'' (present in a concentration $\lambda_2(\tau)$). In each of these two molecular species, the same ``atomic'' constituents have arranged themselves in a different geometrical shape (figures \ref{fig:VNSh}), with a different orientation of the register spin (equations \eqref{eq:M1} and \eqref{eq:M2}),  because of a different energy distribution (figures \ref{fig:distribuzioneenergia}).\\
Comparison with the distribution of $H$ in the pre-measurement state \ket{M(\tau)}, given in \fref{fig:spettro}, shows that the presence of the impurities of type ``2'' is due to unusually intense exchanges of energy between the machine and the reading (measurement) apparatus.
\begin{figure}[htbp]
	\centering
		\includegraphics[width=7cm]{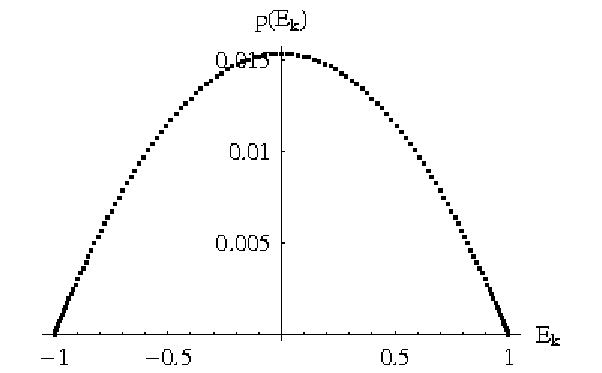}
	\caption{In the state \ket{M(\tau)} the probability distribution of $H$ is given by $p(E_k)=(v_k(x))^2$.}
	\label{fig:spettro}
\end{figure}

\section{The role of initial conditions}
An initial condition of the form
\begin{equation}
	\ket{R(1);\psi_0} = \ket{R(1)} \otimes \sum_{x=1}^{\epsilon} \psi_0(x)\ket{C(x)}
\end{equation}
with $\psi_0$ having support in a bounded region $\Lambda_{\epsilon} = \{1,2,\ldots,\epsilon\} \subseteq \{1,2,\ldots,s\}$ evolves, under \eref{eq:hamiltoniana2} as
\begin{equation}
	e^{-itH} \ket{R(1);\psi_0} = \sum_{x=1}^{\epsilon} \psi_0(t,x) \ket{R(x)}\otimes \ket{C(x)},
\end{equation}
where $\psi_0(t,x)$ solves, with the obvious boundary and initial conditions, the (discretized) free Schr\"odinger equation. The ensuing spreading of the wave packet  leads to an increasing trend (with the exception of the effects of reflection at time $t \approx s$ evidenced in figures \ref{fig:figura4}.b and \ref{fig:figura2}) of the von Neumann entropy $S(\rho_r(t))$ of the state
\begin{equation}
	\rho_r(t)=\sum_{x=1}^s \left| \psi_0(t,x)\right|^2 \ketbra{R(x)}{R(x)} 
\end{equation}
of the register. This is an undesirable feature because $S(\rho_r(t))$ gives a lower bound on the Shannon entropy of the distribution of any observable of the register, for short on the uncertainty in any reading of the output.\\
The models of the previous section where intended to show the above effect; in this section we devote some effort to the goal of decreasing it, by suitable choices of initial condition aimed at reducing the spreading of $Q$  in the state $\psi_0(t,x)$. It is sufficient, for \emph{this} purpose, to study only the cursor, evolving under the Hamiltonian
\begin{equation} \label{eq:freehamiltonian}
	H_0 = - \frac{\lambda}{2} \sum_{x=1}^{s-1} \tau_+(x+1) \tau_-(x)+\tau_+(x) \tau_-(x+1).
\end{equation}
The point is to devise an initial condition $\psi_0$  which uses  whatever additional finite amount $\Lambda_{\epsilon}= \{1,2,\ldots,\epsilon\}$ of space resources is available as a \emph{launch pad} for the cursor in an ``efficient'' way:  this means both a high value of the expectation of $V(\psi_0)$ and a small value of the variance of $V(\psi_0)$ (we want the spreading of $Q$ to increase at a low rate for a short time of computation). That both goals can be achieved is shown by examining the family of initial conditions, given by the eigenstates of a Hamiltonian of the form \eref{eq:freehamiltonian} restricted to qubits in $\Lambda_{\epsilon}$:
\begin{equation}
	\ket{c_k}= \sum_{x=1}^{\epsilon} \sqrt{ \frac{2}{\epsilon + 1} }\sin \left( \frac{k \pi x}{\epsilon+1}\right) \ket{C(x)},\; k=1,2,\ldots,\epsilon.
\end{equation}
The probability density of the speed $V_k \equiv V(c_k)$ corresponding to each of the above states is easily computed from \eref{eq:densitapsi0}:
\begin{eqnarray} \label{eq:densitavk}
	f_{V_k}(v)& = & I_{(0,1)}(v) \cdot  \\
	& \cdot & \frac{4 \left(3-2 v^2 + \cos \left( \frac{2 k \pi}{\epsilon + 1}\right) \right) \left( \sin \left( \frac{k \pi}{\epsilon + 1} \right) \right)^2 \left( \sin((\epsilon+1) \arcsin(v))\right)^2}{\pi \sqrt{1-v^2}(\epsilon + 1) \left( 2 v^2 + \cos \left( \frac{2 k \pi}{\epsilon + 1 }\right)-1\right)^2.  \nonumber}
\end{eqnarray}
\begin{figure}[htbp]
	\centering
		\includegraphics[width=300pt]{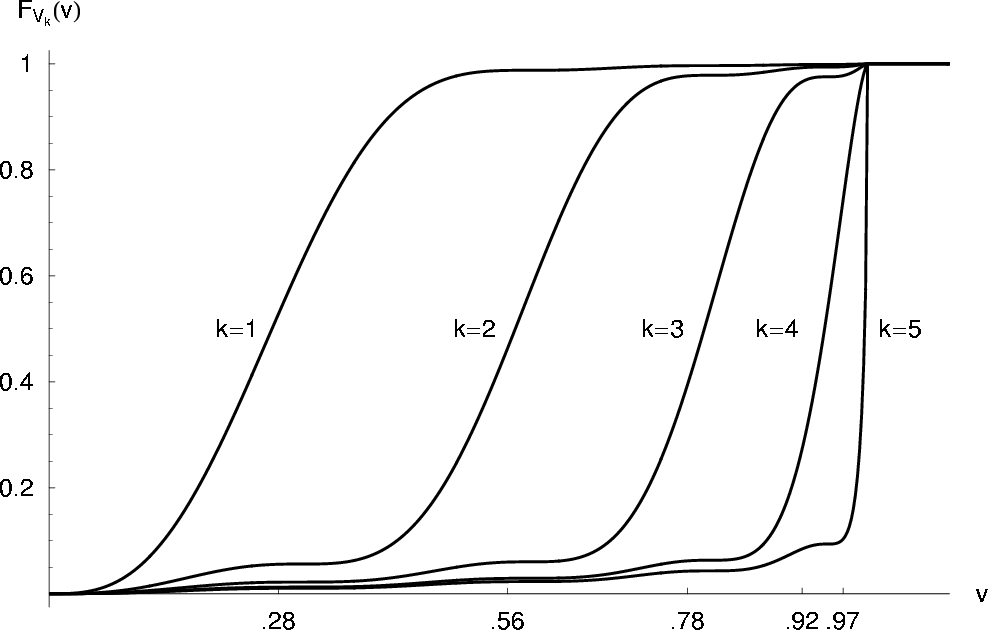}
	\caption{$n=5,\ \epsilon=2n-1=9$. The cumulative distribution functions  \mbox{$F_{V_k}(v)=Prob(V_k \leq v)$} corresponding to the densities \eref{eq:densitavk}, for $k$ going from $1$ to $n$. The ticks on the $v$ axes are $E(V_1)<E(V_2) < \ldots <E(V_5)$.}
	\label{fig:figura13}
\end{figure}
The behavior of $V_k$ is examplified by \fref{fig:figura13}. We are taking there, as we will always do in this section for the sake of notational convenience, $\epsilon$ to be odd
\begin{equation}
	\epsilon=2 n - 1.
\end{equation}
The examples of \fref{fig:figura13} clearly show the dispersive  nature of the medium \eref{eq:freehamiltonian}; they also show that increase of the mean value is accompanied by decrease of the variance (as shown by the increase in the steepness of the graph as $k$ goes from $1$ to $n$). An obvious choice for the initial state of the cursor emerges from the above example:
\begin{eqnarray} \label{eq:condiniflat}
	\ket{c_n} & = & \sum_{x=1}^{2n-1} \sqrt{\frac{1}{n}} \sin\left(\frac{\pi}{2} x \right) \ket{C(x)} =  \nonumber \\
	&=& \frac{\ket{C(1)}-\ket{C(3)}+\ket{C(5)}+ \ldots - (-1)^n \ket{C(2n-1)}}{\sqrt{n}} ,
\end{eqnarray}
conforming to the idea of packing the maximum number of wavelengths in the \emph{launch pad} $\Lambda_{\epsilon} = \{1,2,\ldots,\epsilon\}$, and having a, presumably easy to prepare, stationary state of the free  $XY$ chain localized in  $\Lambda_{\epsilon}$.\\
The random variable $V_n \equiv V(c_n)$  has probability density
\begin{equation}
	f_{V_n} = I_{(0,1)}(v) \frac{\left( \sin(2 n \arcsin(v))\right)^2}{ \pi n (1-v^2)^{3/2}}
\end{equation}
and, therefore, expectation value
\begin{eqnarray} \label{eq:attesavn}
	E(V_n) & = & \frac{4}{\pi}\sum_{h=1}^n \left(\frac{1}{ 4h-3}-\frac{1}{4h-1} \right) =    \nonumber \\
	& = & 1- \frac{4}{\pi}\sum_{h=n+1}^{+\infty} \left(\frac{1}{ 4h-3}-\frac{1}{4h-1} \right) \approx \nonumber \\
	& \approx & 1-\frac{1}{2 \pi n}
\end{eqnarray}
The second moment of $V_n$ is explicitly given by
\begin{equation}
	E(V_n^2) = 1-\frac{1}{4n}.
\end{equation}
The above considerations lead to the following asymptotic behavior, for large $n$, of the variance of $V_n$:
\begin{equation} \label{eq:varianzavn}
	var(V_n) = \frac{4 -\pi}{4 \pi n}.
\end{equation}
Equation \eref{eq:attesavn} is a quantitative assessment of the cost in terms of space resources of achieving the first requisite of \emph{efficiency}, namely high mean speed;  similarly, \eref{eq:varianzavn} gives the cost of decreasing the variance of $V_n$.\\
Incidentally, as the observable $Q$ has, in the state \ket{c_n}, expectation value
\begin{equation}
	E(Q_n) = \bra{c_n} Q \ket{c_n} = n
\end{equation}
and variance
\begin{equation}
	var(Q_n)=\frac{n^2-1}{3},
\end{equation}
equation \eref{eq:varianzavn} can be read as saying that, in the initial state \ket{c_n}, the position-velocity uncertainty product is given by
\begin{equation}
	var(Q_n)var(V_n) \approx \frac{n (4-\pi)}{12 \pi}
\end{equation}
Figures \ref{fig:figura5} and \ref{fig:figura7} show the relevance of the above \emph{asymptotic} considerations for the case of \emph{finite} $\epsilon$  and finite $s$ for $t<s$.
\begin{figure}[t]
	\centering
		\includegraphics{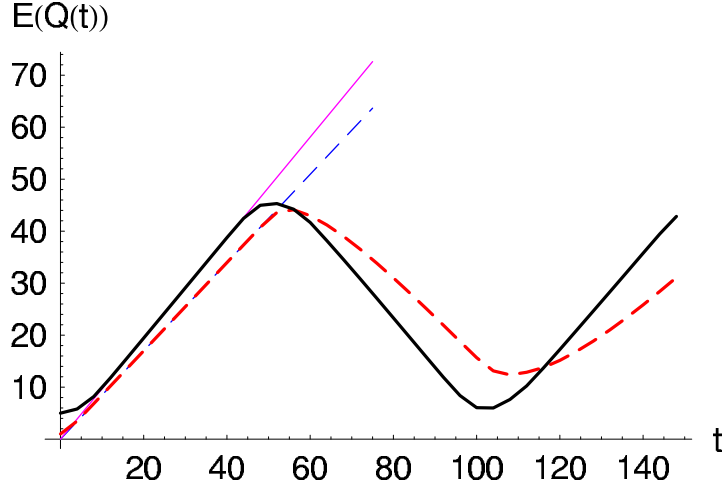}
	\caption{Solid lines: the expectation value of $Q$ in  a state \ket{M_n(t)} evolving from an initial condition having the cursor in  \ket{c_n}; the slope of the initial linear part of the graph is correctly predicted by \eref{eq:attesavn}. For comparison purposes the dashed lines show $\bra{M_1(t)}Q \ket{M_1(t)}$ as a function of time and the corresponding linear fit with slope given by \eref{eq:EM1}.}
	\label{fig:figura5}
\end{figure}
\begin{figure}[htbp]
	\centering
	\subfigure[]{\includegraphics[width=7cm]{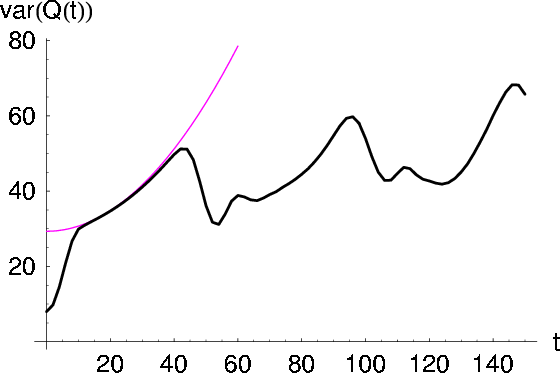}}
	\subfigure[]{\includegraphics[width=7cm]{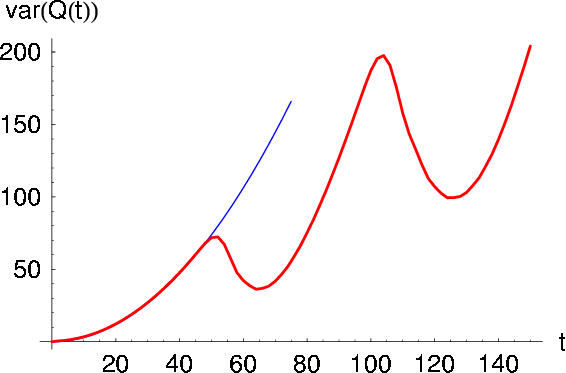}}
	\caption{$s=50,\ n=5$. (a) The variance of  $Q$ in a state \ket{M_n(t)} as a function of $t$, compared with its best fit of the form $const.+ t^2(4-\pi)/(4 \pi n)$, in the time interval $(\epsilon, s-\epsilon)$ in which boundary effects can be neglected. (b) The variance of $Q$ in the state \ket{M_1(t)}, compared with its approximation $t^2(3/4 -(8/(3 \pi))^2)$, suggested by \eref{eq:varM1}.}
\label{fig:figura7}
\end{figure}
The effect of the initial condition is most evident  if we compare the evolution of the state of the register from the initial state $\ket{M_1} = \ket{R(1)} \otimes \ket{C(1)}$ with the evolution starting from
\begin{equation}
	\ket{M_n} = \ket{R(1)} \otimes \ket{c_n}.
\end{equation}
\begin{figure}[ht]
\centering
\hspace{-0cm}\subfigure[]{\includegraphics[width=79mm]{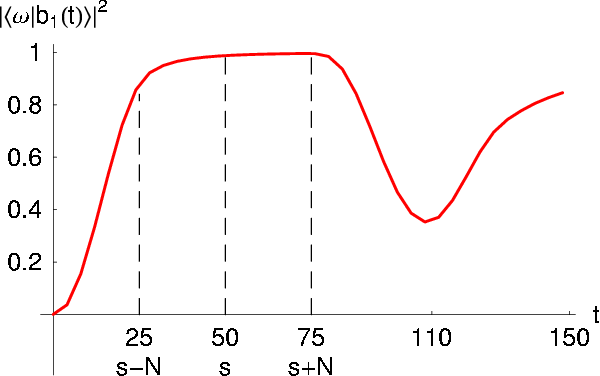}}\\
\subfigure[]{\includegraphics[width=79mm]{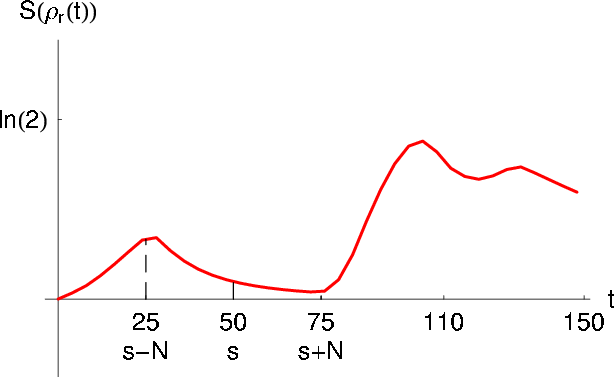}} \\
\subfigure[]{\includegraphics[width=59mm]{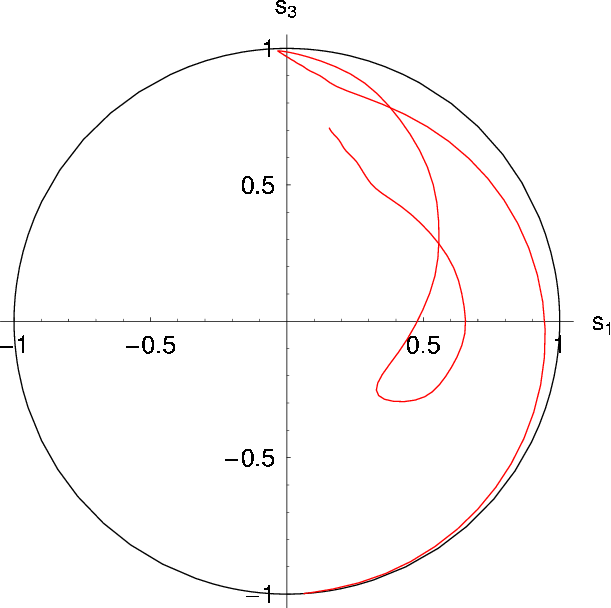}} 
\caption{$s=50,\ 0 \leq t \leq 3s,\ \mu=10,\ N= \left\lfloor \frac{\pi}{4} 2^{\mu/2} \right\rfloor$ (the Grover-optimal number of active steps); $U_1=U_2=\ldots=U_N=\exp(-i \alpha \sigma_2/2)$, with $\alpha$ and $\theta$ chosen as in \eref{eq:theta} and \eref{eq:alpha}; $U_x=I_r$ for $x > N$; initial state $\ket{M_1}= \ket{R(1)} \otimes \ket{C(1)}$.}
\label{fig:figura8}
\end{figure}
\begin{figure}[ht]
\centering
\subfigure[]{\includegraphics[width=79mm]{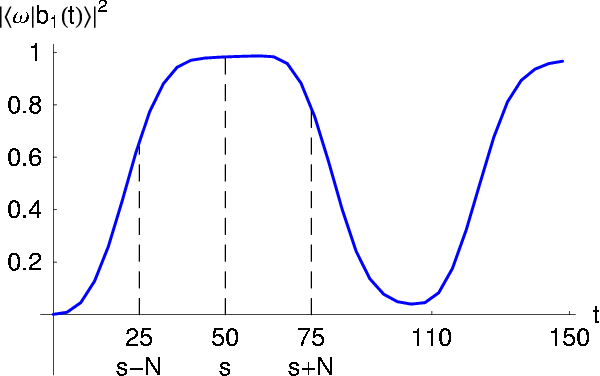}}\\
\subfigure[]{\includegraphics[width=79mm]{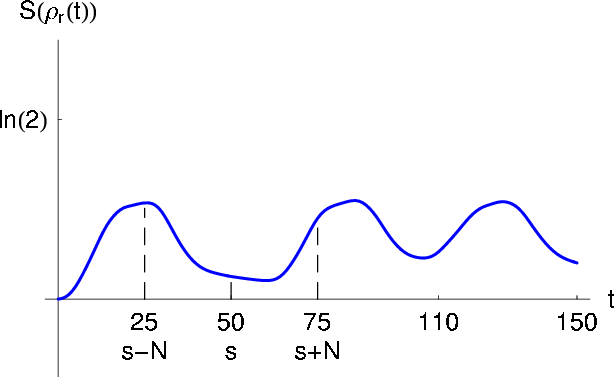}} \\
\subfigure[]{\includegraphics[width=59mm]{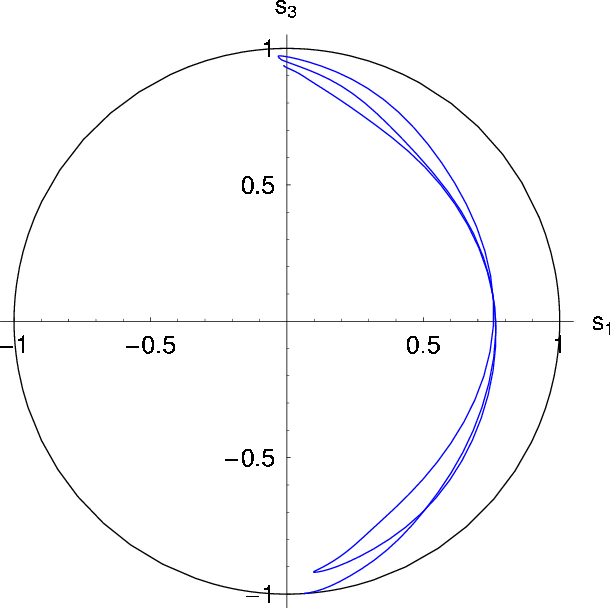}} 
\caption{$s=50,\ 0 \leq t \leq 3s,\ \mu=10,\ N=25,\ \alpha$ and $\theta$ as in \fref{fig:figura8}; $n=5,\ \epsilon=2n-1$; $U_{\epsilon}=U_{\epsilon+1}=\ldots = U_{\epsilon+N-1}= \exp(-i\alpha \sigma_2/2)$; $U_x=I_r$, for $1\leq x < \epsilon$ or $x \geq \epsilon + N$; initial state $\ket{M_n}= \ket{R(1)}\otimes \ket{c_n}$.}
\label{fig:figura9}
\end{figure}
This is done in figures \ref{fig:figura8} and \ref{fig:figura9} in the same \emph{probability-entropy-Bloch } format as in \fref{fig:figura4}. We examine there two different ways of using an additional amount $s-N$ of space, of size comparable with the minimum amount $N$ required by the algorithm. \Fref{fig:figura8} summarizes the experience developed in \cite{apolloni02} on the effect of using all this additional space as a telomeric chain or ``landing strip'': as long as the cursor stays in this region the register remains acted upon by the optimal number of primitives. \Fref{fig:figura9} shows the improvement obtained by investing part of the additional space as a ``launch pad'' on which to prepare a state in which the spreading of the cursor  increases (see \fref{fig:figura7}) at a lower rate than when starting from position 1.\\
Comparison of figures \ref{fig:figura8}.c and \ref{fig:figura9}.c, in particular the improvement of the behavior after  reflections at site $s$, shows that the idealized scenario of reversible computation (the cursor, ``going back and forth'', ``does and undoes'' the reversible computation) is within reach, with, as \eref{eq:varianzavn} shows, a polynomial cost in space.
\begin{figure}[ht]
\centering
\subfigure[]{\includegraphics[width=79mm]{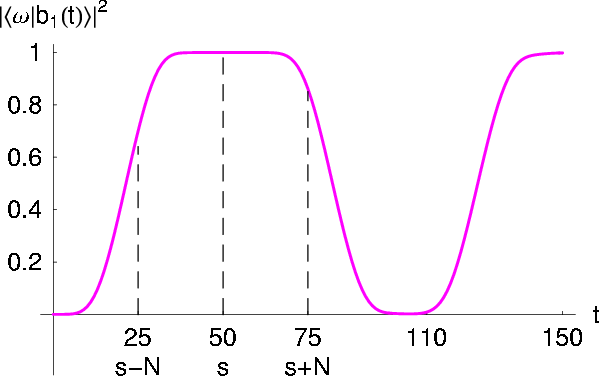}}\\
\subfigure[]{\includegraphics[width=79mm]{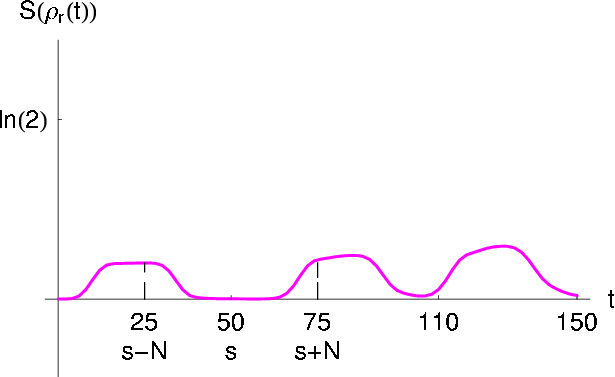}} \\
\subfigure[]{\includegraphics[width=59mm]{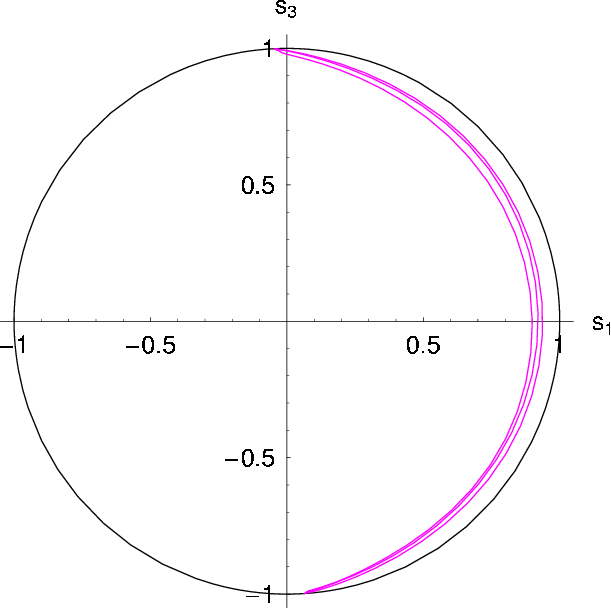}} 
\caption{Same parameters as in \fref{fig:figura9}; initial state  $\ket{R(1)}\otimes \ket{\gamma_n}$, as in \eref{eq:gamman}.}
\label{fig:figura10}
\end{figure}
We note, in \fref{fig:figura10}, that we can do much better than in \fref{fig:figura9}, with the same expenditure of space resources, in approximating the reversible scenario if, instead of the initial state \eref{eq:condiniflat}, we set the cursor in the initial state
\begin{equation} \label{eq:gamman}
	\ket{\gamma_n}=\sqrt{\frac{2}{3n}} \sum_{x=1}^{2n-1} \left( 1+\cos \left( \frac{\pi}{2n} x\right)\right) \sin\left( \frac{\pi}{2} x\right) \ket{C(x)}.
\end{equation}
The state $\ket{\gamma_n}$ emerges quite naturally as a three-mode approximation (a linear combination of \ket{c_n}  and  \ket{c_{n\pm1}}) of the initial condition that maximizes the mean speed of computation for fixed length $\epsilon$ of the \emph{launch pad}.

\section*{Summary}
The quantitative characterization of the random variable \emph{time of flight speed} shows that the, classically obvious, choice of a sharp initial condition $\ket{M_1}$, leads to a fast spreading of the wave packet; this in turn, makes the von Neumann entropy of the register grow fast. We showed that the entropy of the register is related to the probability of finding the desired output state written on the it. We showed, moreover, that we can considerably reduce the spreading rate by means of extra space resources, or \emph{launch pad}, into which we prepare a state \ket{\gamma_n}, which is ``the most'' localized in momentum representation in a finite region.\\
While studying the analytic expression of the entropy of the register in our toy model, $\rho_r(t)$, we gave a simple example of Lindblad dynamics for the state of the system: the clock, which administers the computational primitives to the register at random times, acts as a decohering environment for the latter.\\
In section \ref{sec:energy} we made quantitative assessments on the energetic perturbation induced by a measurement of the system. We observed that, since the Hamiltonian is time independent, the only cost of a computation on a Feynman machine is related to the energy required to reset the system. What we showed is that it depends on the outcome of our measurement.

\chapter{A multi-hand quantum clock} \label{chap:imperfections}
In this chapter we generalize the model of previous chapters by having more than one particle traveling on the program line. We will show that, as soon as the operators acting on the register do not commute, new effects appear; we provide an example where the excitations get confined  in bounded regions, an effect resembling Anderson localization \cite{anderson58,anderson73}.\\
The resolution of the eigenvalue problem for what we call the  \emph{multi-hand quantum clock} in the non-commutative context of the cursor-register coupling in the general subspace $N_3=k$  turns out to be too involved (at least for the author).  The numerical examples provided in this section are frankly heuristic. Analytic justification must wait for future developments.
\section{Number of particles} \label{sec:numberofparticles}
In the previous chapter we have provided examples of the benefit of spreading the initial wave function of the cursor ($N_3=1$) on an initial \emph{launch pad} instead of, as it would be classically ``obvious'', having it strictly localized at site $1$. Equality \eref{eq:varianzavn} is, in this context, a quantitative assessment of the cost, in term of space resources, of implementing Feynman's ballistic mode of computation.\\
In  this section we  abandon, in the same spirit, the classical prejudice of having a single clocking excitation, and present a preliminary analysis of the idea of starting the cursor in an initial state with  $N_3>1$. The idea is to follow the motion of a swarm of several clocking agents (cursor spins in the ``up'' state) acting on the register. Stated otherwise, with  reference for simplicity to the case  $N_3=2$, we allow the \emph{clock} to perform a quantum  walk on the graph having the vertices $(x_1,x_2)$, with  $1 \leq x_1 < x_2 \leq s$, with edges between nearest neighbors \cite{osborne04}.\\
We recall, mainly in order to establish our  notation, a few basic facts \cite{lieb61} about the $XY$ Hamiltonian \eref{eq:freehamiltonian}.\\
The eigenstates of $H_0$ in the subspace $N_3=n$ are labeled by subsets of size $n$ of $\Lambda_s=\{1,2,,\ldots,s\}$; if $\mathsf{K}=\{\ves{k}{n}\}$ is such a subset (where we will always assume $1 \leq k_1 < k_2 < \ldots <k_n \leq s$), an eigenstate of $H_0$ belonging to the eigenvalue
\begin{equation} \label{eq:multieigen}
	E_{\mathsf{K}}=\sum_{j=1}^n e_{k_j}
\end{equation}
is given by
\begin{equation}
	\ket{E_\mathsf{K}}= \sum_{M \subseteq \Lambda_s; |M|=n} V(\mathsf{K},M) \ket{M}.
\end{equation}
For  $M=\{\ves{x}{n}\}$, with $1 \leq x_1 < x_2 < \ldots <x_n \leq s$, we have indicated above by \ket{M} the simultaneous eigenstate of $\ve{\tau_3}{s}$  in which only the spins in $M$ are ``up'', and we have set:
\begin{equation} \label{eq:slaterdet}
	V(\mathsf{K},M)= \det\left(\left\|v_{k_i}(x_j) \right\|_{i,j=1,\ldots,n}\right)
\end{equation}
where the functions $v_k$ have been defined in \eref{eq:autofunzioni}.\\
We set
\begin{equation}
	Q_i \ket{\{\ves{x}{n}\}} = x_i \ket{\ves{x}{n}}.
\end{equation}
It is easy to study, by the techniques of \sref{sec:speed}, the asymptotic (as $s \to +\infty$ and \mbox{$t \to +\infty$}) joint distribution of the observables $Q_i$, and therefore to give quantitative estimates of the correlation between the speeds of different \emph{particles} and its dependence on the initial condition. To quote just one example, in the subspace $N_3=2$ and in the state \ket{\{1,2\}}  the velocities $(V_1,V_2)$ of the two ``up'' spins  (the limits in law of $Q_1/t$ and $Q_2/t$, respectively) have joint probability density
\begin{equation} \label{eq:jointdistr}
	f_{V_1,V_2}(v_1,v_2) = I_{(0,v_2)}(v_1) I_{(0,1)}(v_2)\frac{64 v_1^2 v_2^2 (2-v_1^2-v_2^2)}{\pi^2 \sqrt{(1-v_1^2)(1-v_2^2)}} 
\end{equation}
It is immediate from \eref{eq:jointdistr} to compute the conditional expectation  $E(V_1 | V_2)$ of the velocity of the leftmost \emph{particle} given the one of the rightmost \emph{particle}; it turns out to be:
\begin{equation}
	E(V_1|V_2)= \frac{3 V_2}{4}+ O(V_2^5). 
\end{equation}
In this section we advance the following idea: if the issue of the computation is the application, for a given number $g$ of times, of a given primitive $G$ to the register, initialize the cursor in the $N_3=g$ subspace, in the state, say, \ket{\{1,2,\ldots,g\}}; let then the system evolve according to the Hamiltonian:
\begin{equation} \label{eq:hamiltoniana5}
	H = -\frac{\lambda}{2} \sum_{x=1}^{s-1} U_x \otimes \tau_+(x+1) \tau_-(x)+  U_x^{-1} \otimes \tau_+(x) \tau_-(x+1)
\end{equation}
where
\begin{equation} \label{eq:ux}
	U_x = G^{\delta_{x_0,x}},\mbox{ for a fixed $x_0 \geq g$},\; G^0 = I_r.
\end{equation}
\begin{figure}[h]
\subfigure[]{\includegraphics[width=60mm]{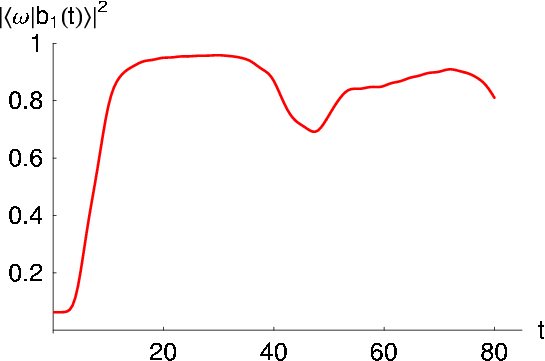}}
\subfigure[]{\includegraphics[width=60mm]{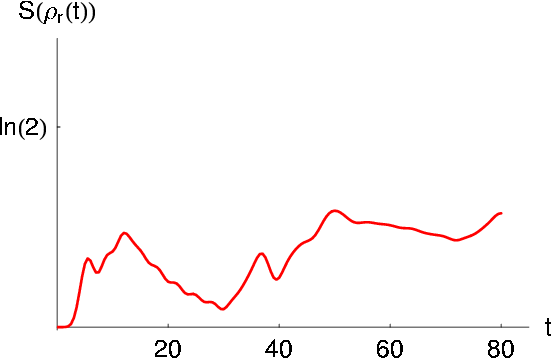}} \\
\centering
\subfigure[]{\includegraphics[width=59mm]{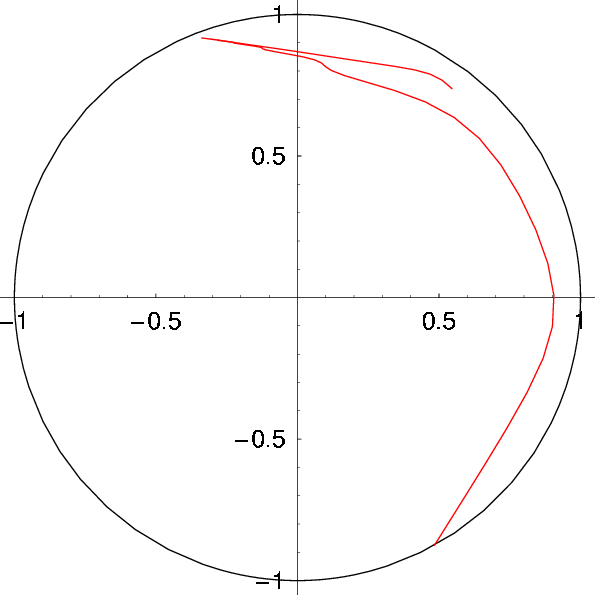}} 
\caption{$\mu=4,\ g=3,\ x_0=6,\ s=20,\ 0 \leq t \leq 4s$; $U_{x_0}=G=\exp(-i \alpha \sigma_2 / 2)$ with $\alpha$ and $\theta$ given by \eref{eq:alpha}  ad \eref{eq:theta}, $U_x=I_r$ for $x \neq x_0$; initial condition $\ket{R(1)} \otimes \ket{\{1,2,3\}}$ with \ket{R(1)} given by \eref{eq:condinigrover}.}
\label{fig:figura11}
\end{figure}
An implementation of this approach is shown by the probability-entropy-Bloch diagram of \fref{fig:figura11}.
Simple expressions for the quantities shown in \fref{fig:figura11} can be obtained by the explicit form of the eigenvectors of the Hamiltonian described by \eref{eq:hamiltoniana5} and \eref{eq:ux} in every eigenspace of $N_3$. For instance in the subspace $N_3=3$ a complete set of eigenstates is given, for $\zeta=\pm 1$ and $1 \leq j <h < k \leq s$, by:
\begin{eqnarray} \label{eq:evolutomulti}
	\ket{\zeta;E_{\{j,h,k\}}} & = & \sum_{1 \leq x_1 <x_2 < x_3 \leq s} V(\{j,h,k\},\{x_1,x_2,x_3\})\cdot \nonumber \\
	& \cdot &  G^{\vartheta(x_1-x_0)+\vartheta(x_2-x_0)+\vartheta(x_3-x_0)} \ket{\sigma_3=\zeta} \otimes \ket{\{x_1,x_2,x_3\}}
\end{eqnarray}
where $\vartheta$ is the unit step function defined by:
\begin{equation} \label{eq:thetax}
	\vartheta(x)= \left \{ 
	\begin{array}{c r}
	1, & \mbox{if $x > 0$}\\
	0, & \mbox{if $x \leq 0$} 
	\end{array}\right.
\end{equation}
\section{Non commuting computational primitives} \label{sec:nonabelian}
The spectral structure \eref{eq:evolutomulti} is peculiar of the extremely simple situation \eref{eq:ux} (just one active link) considered there. As soon as we have more than one active link, say the primitive $A$ acting on link $(a,a+1)$ and the primitive $B$ acting on link $(b,b+1)$, with $b>a+1$, a new phenomenon (that for simplicity we discuss in the $N_3=2$ case) takes place: the energy eigenstates have not anymore the form of a linear combinations of tensors products of the form  $M(x_1,x_2)\ket{\sigma_3=\zeta} \otimes \ket{\{x_1,x_2\}}$, with $M(x_1,x_2)$ a \underline{monomial} in $A$ and $B$; related to this, the coordinates $x_1,\ x_2$ lose, strictly speaking, the meaning of relational time \cite{gambini04b}: given that at a given value of $t$, $Q_1=x_1$ and $Q_2=x_2$ we can only claim that the state of the register has been acted upon by a \underline{polynomial} in $A$ and $B$.\\
This phenomenon is easily understood in terms of the Dyson expansion of the propagator: the probability amplitude for the two excitations being in $x_1,\ x_2$ (both larger than $b$), given that at time $0$ they were in $y_1,\ y_2$ (both $\leq a$), receives contributions not only from Feynman paths along which the rightmost excitation goes past $a$ and $b$ and \underline{then} the leftmost excitation goes past  $a$ and $b$ (along such a computational path the state of the register is modified by $BABA$), but also, among others, from paths along which both excitations go past $a$ before both going past $b$ (along such a computational path the state of the register is modified by  $BBAA$).\\
By means of numerical simulation, it is possible to have some insight into  the effect of the interaction by the register when it is acted upon by non-commuting operators.\\
Let us consider a machine with one register spin in the state \ket{\sigma_3=+1} and a cursor chain of $s$ sites initialized in the state \ket{\{1,2\}}, that is
\begin{equation}
	\ket{\psi(0)}=\ket{M_1}= \ket{\sigma_3=+1}\ket{\{1,2\}}.
\end{equation}
Consider the following family of Hamiltonians
\begin{eqnarray}
H_{a,b} & = & \sum_{x=1}^s B^{\delta_{b,x}} A^{\delta_{a,x}} \tau_+(x+1)\tau_-(x)+ h.c.,
\end{eqnarray}
where $\delta_{j,x}$ is the Kronecker $\delta$ function.\\
Each Hamiltonian of the family is characterized by having only two active links $a \rightarrow a+1$ and $b \rightarrow b+1$, \emph{during which} the operator $A$ and $B$ are respectively administered to the register.\\
Figures \ref{fig:gconfrontoloc} and \ref{fig:gsinistraloc} are suggestive of the following behavior: the motion of the excitations along the chain is unaffected by the interaction with the register as long as the active links are next to each other ($b=a+1$). This reflects the fact that the energy spectrum of Hamiltonians of type $H_{a,a+1}$ is the same as the one of the free Hamiltonian. In fact, it can be easily shown that the states
\begin{eqnarray}
	\ket{\zeta;E_{\{j,h\}}} & = & \sum_{1 \leq x_1 <x_2 \leq s} \mathsf{V}(\{j,h\},\{x_1,x_2\})\cdot  \\
& \cdot &  B^{\vartheta(x_1-a-1)} A^{\vartheta(x_1-a)}B^{\vartheta(x_2-a-1)}A^{\vartheta(x_2-a)}\ket{\sigma_3=\zeta} \otimes \ket{\{x_1,x_2,x_3\}}, \nonumber
\end{eqnarray}
where $\mathsf{V}(\{j,h\})$ are defined as in \eref{eq:slaterdet},
are eigenstates of $H_{a,a+1}$ corresponding to the eigenvalue $E_{\{j,h\}}$ defined in \eref{eq:multieigen} and $\vartheta$ is defined as in \eref{eq:thetax}.\\
As soon as $b>a+1$, namely the active links are separated by inactive links (that is links during which the identity is applied) \emph{localization} appears. This suggests a change in the spectrum of Hamiltonians $H_{a,a+k},\ k > 1$. Whereas in the first case ($H_{a,a+1}$) we can give analytic evidence of the equality of the ``free particles'' spectrum and the interacting one, in this second case we have only numerical evidence of what we stated.\\
In the following examples, we choose the operators $A=\sigma_1$ and $B=\sigma_3$. We observe that 
\begin{eqnarray} \label{eq:kickback}
\sigma_1 \sigma_3 \sigma_1 \sigma_3 \ket{+1} = e^{i\pi} \sigma_1 \sigma_1 \sigma_3 \sigma_3 \ket{+1}  
\end{eqnarray}
\begin{figure}[p] 
	\centering	\includegraphics{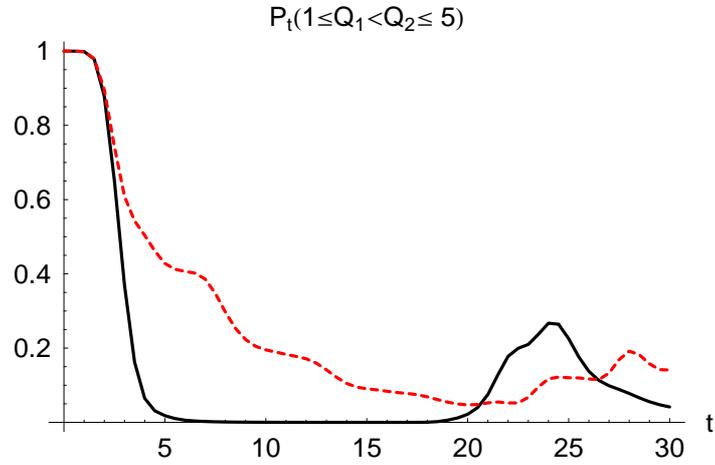}
	\caption{$s=20$, $A=\sigma_1$, $B=\sigma_3$ the solid line represents the probability of observing $1 \leq Q_1 < Q_2 \leq 5$ when $a=2$ and $b=3$; the dashed line the same probability when $a=2$ and $b=4$.}
	\label{fig:gconfrontoloc}
\end{figure}
\begin{figure}[p]
	\centering		\subfigure[]{\includegraphics[width=7cm]{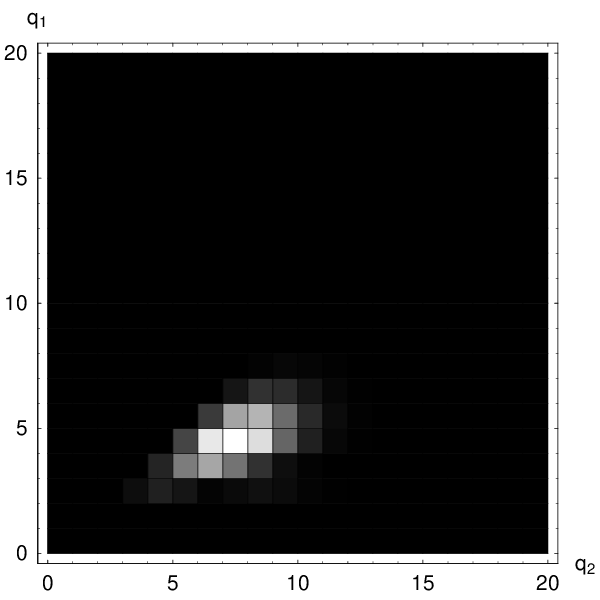} }	\subfigure[]{\includegraphics[width=7cm]{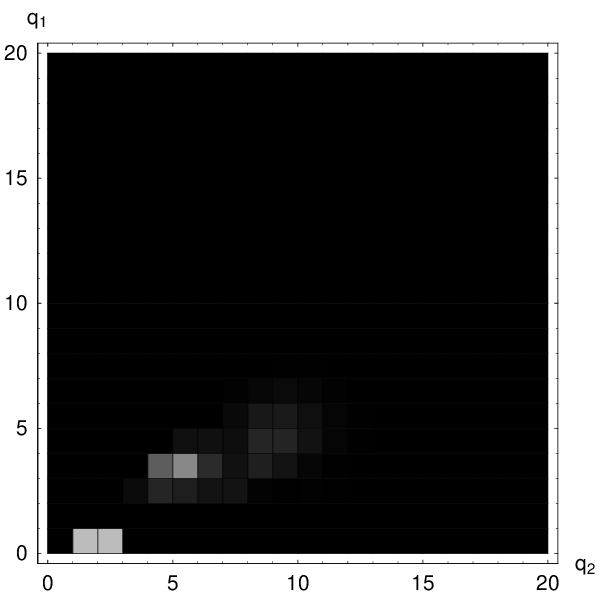} }
	\caption{Same parameters and operators as in \fref{fig:gconfrontoloc}. We represent the probability $P_t(Q_1=q_1, Q_2=q_2)$ for $t=4$; (a) corresponds to the choice $a=2$ and $b=3$; (b) represents the same quantity for $a=2$ and $b=4$.}
	\label{fig:gsinistraloc}
\end{figure}
\begin{figure}[h]
	\centering	\includegraphics[width=70mm]{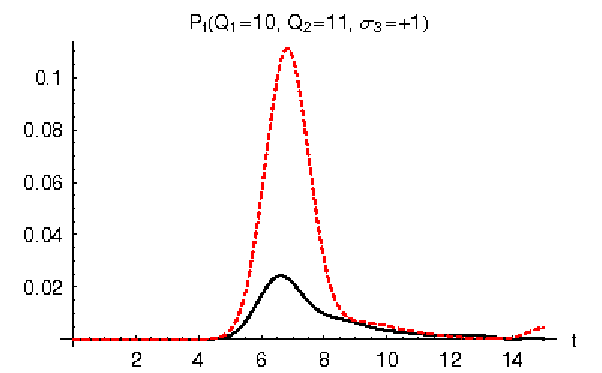}
		\hspace{1cm}	\includegraphics[width=70mm]{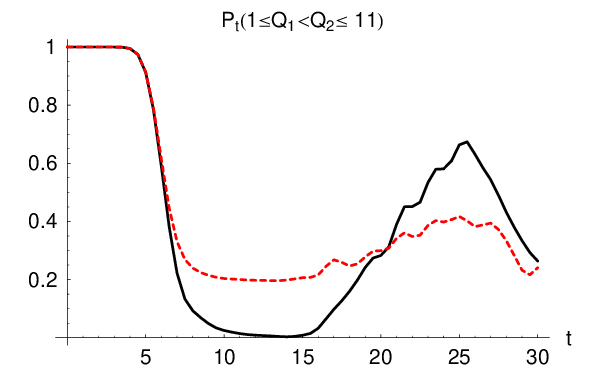}
		\begin{picture}(230,10)
			\put(0,0){(a)}
			\put(230,0){(b)}
		\end{picture}
	\caption{ $s=20$,  $A=\sigma_1$, $B=\sigma_3$; solid (black) line: $a=9,\ b=10$; dashed (red) line: $a=9,\ b=11$. (a) The destructive interference between the $BBAA$ and $BABA$ causes a ``reflection'' on the dynamical boundary. (b) The ``confinement'' effect due to the destructive interference between different computational paths.}
	\label{fig:gconfronto}
\end{figure}
Figures \ref{fig:gconfrontoloc}, \ref{fig:gsinistraloc} and \ref{fig:gconfronto} show the time evolution, in the $N_3=2$ case, of the probability of reaching a determined configuration of the machine when $b=a+1$ (solid black line) and $b=a+2$ (dashed red line) for different positions of the active links. For the first-neighbor case to all the computational paths ending in $\ket{\underline{q}}$ there corresponds the transformation $BABA$ on the register; in the second case, there are different computational paths ending in $\ket{q}$ to which there correspond different transformations. When $b=a+2$ the classes of paths for which the action on the register is $BBAA=1$ and $BABA=-1$ have the same probability amplitude and gather the most of the probability mass. The destructive interference between them is then due to the change of phase (kickback effect, see \eref{eq:kickback}) of the amplitude brought about by the different transformations on the register. The consequent dynamical creation of  boundary conditions causes a ``confinement'' of the traveling excitations on the left of the interference location, here coinciding with the end of the last active link.\\[5pt]
Waiting for an algorithm that might benefit from the above possibility of simultaneously exploring different computational paths (concurrency ?), we explore, in \fref{fig:figura12}, the idea (or classical prejudice?) that this nuisance  can be in part avoided by using suitable initial conditions. The idea, suggested by \eref{eq:densitavk}, is of course to prepare on $\Lambda_a=\{1,2,\ldots,a\}$ an initial $N_3=g$ state such that the excitations travel as spatially well localized wave packets of so different speeds that it is at any time unlikely that they simultaneously hit the region $(a+1,b)$.\\
\begin{figure}[!h]
	\centering	\subfigure[]{\includegraphics[width=7cm]{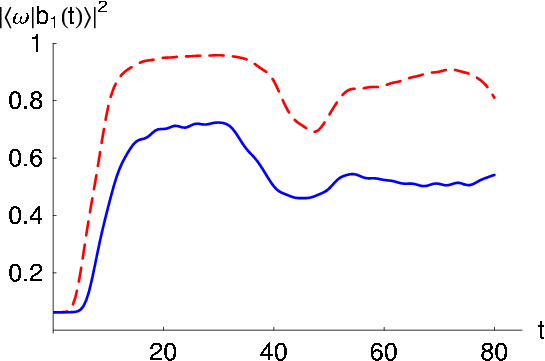}}	\subfigure[]{\includegraphics[width=7cm]{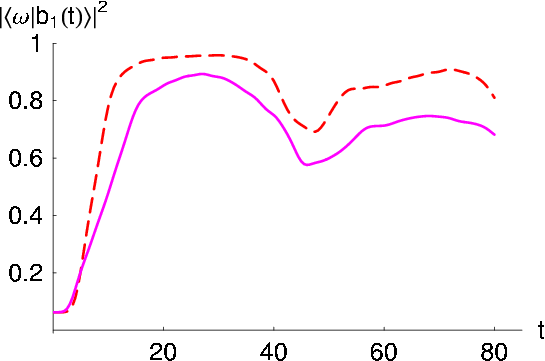}}
	\caption{$\mu=4,\ s=20,\ 0 \leq t \leq 4s$; (a) solid line: $U_a=A,\ U_b=B,\ a=6,\ b=8$, same initial condition as in \fref{fig:figura11}; (b) solid line: the $N_3=3$ state has been prepared by setting the initial chain $\{1,\ldots,6\}$ in the ground state of the Hamiltonian $\sum_{x=1}^5 \tau_+(x+1)\tau_-(x)+h.c$. For comparison purpose \fref{fig:figura11}.a is reproduced in both frames as a dashed line.}
\label{fig:figura12}
\end{figure}
A remark about our insistence, throughout the previous chapter, in gathering experience about the  behavior of the evolution of a state of an initial subchain \mbox{$\Lambda_{\epsilon}= \{1,2,\ldots,\epsilon\}$} is now in order.\\
We observe that an initial state  (not necessarily in the $N_3=1$ subspace) in $\Lambda_{\epsilon}$ of the form
\begin{equation}
	\ket{in}= \frac{1}{2^{\epsilon}} \sum_{M \subseteq \Lambda_{\epsilon}} \left( \sum_{z \in \{-1,1\}^{\epsilon}} f(z) \prod_{j \in M} z_j\right) \ket{\tau_3(x)=(-1)^{I_M(x)}},
\end{equation}
where $I_M$ is the indicator function of the set $M$ and $x=1,\ldots,\epsilon$, can be prepared as a post-kickback state (with respect to an ancilla qubit) after the reversible evaluation of a function $f:\{-1,1\}^{\epsilon} \to \{-1,1\}$. We conjecture that subsequent evolution of \ket{in} under the Hamiltonian \eref{eq:freehamiltonian} on $\Lambda_{s}= \{1,2,\ldots,s\}$, with $s>>\epsilon$, might help in setting tests of hypotheses about the Fourier coefficients
\begin{equation}
	c_M=\frac{1}{2^{\epsilon}} \sum_{z \in \{-1,1\}^{\epsilon}}f(z) \prod_{j \in M} z_j
\end{equation}
of the function $f$ via time-of-flight techniques. There is at least one non trivial case in which the above conjecture works: having prepared all spins in  $\{\epsilon+1,\ldots,s\} $ in the ``up'' state, the Deutsch-Josza alternative \cite{deutsch92} `` $f$ constant ($|c_{\emptyset}|=1$) vs. $f$ balanced ($c_{\emptyset}=0$)'' becomes equivalent to the alternative ``stationary vs. non stationary'' under the Hamiltonian  \eref{eq:freehamiltonian},  about the state of the overall system.\\[10pt]
\chapter{Conclusions and outlook}
The CCNOT gate can be implemented on Feynman's quantum computer  by means of conditional jumps: the transition amplitudes at a vertex of the clocking  graph depend on the state of the register. The Feynman machine, thus, not only anticipates the continuous time quantum walk paradigm but also extends it to what we call \emph{interacting quantum walks}. Since the CCNOT alone constitutes a complete reversible logical basis, the computer model is universal with respect to the class of function reversibly computable by a Turing machine.\\[5pt]
In dealing with specific algorithms, we abandoned the \emph{black-boxes}, or \emph{oracular}, setting and adopted a top-down approach. For instance in the case study of  Grover's algorithm, we wrote down \emph{explicitly} the Hamiltonian generating, at \emph{certain times}, the evolution of the input state into the desired output state. In this specific context, we defined a scheme for the \emph{iteration of quantum subroutines} and analyzed the space cost of it, which turns out to be logarithmic in the number of iteration. Moreover, we showed that the CCNOT \emph{can be implemented by using only three body interactions}; this represents a significant reduction of the complexity of the physical realization of the computing device.\\[5pt]
We have shown that, by means of a generalized Peres basis, it is possible to express the evolution of the system as if the clocking excitation were traveling without interaction along a linear spin chain. As far as the dynamics of the interacting quantum walk is concerned, then, we can restrict our attention to the motion of the clocking excitation on a linear chain.\\[5pt]
The pure $XY$ Hamiltonian $H_0$  given in \eref{eq:freehamiltonian} describes, in the Luther-L\"uscher-Susskind formalism \cite{luther76,luscher76,susskind77}, a massless Dirac quantum field on a $1$-dimensional lattice. The full Hamiltonian \eref{eq:hamfeynman} is suggestive of the minimal coupling of this Fermi field, implementing the \emph{clock}, with additional quantum fields implementing the \emph{register}. We have been trying to contribute to the line of research, that seems to be emerging these days \cite{childs04, strauch05, verstraete05}, devoted to making this connection between quantum computing and relativistic quantum field theory explicit. It is an easy guess that this quantum field theoretical intuition was well present in the original work \cite{feyn86}. Particularly penetrating is, in this respect, Peres' remark that in Feynman's model \emph{calculations run forward and backward in time just as particles and antiparticles in Feynman's classical work on relativistic quantum field theory} (\cite{peres85}, p. 3269). As a further remark, we observe that the three body interactions needed by Feynman's model are hard to conceive out of a field theoretical context.\\[5pt]
It is because of this field theoretical perspective that we have tried to avoid any ``engineering'' \cite{christandl04}(space dependence) of the coupling constant $\lambda$ in \eref{eq:freehamiltonian}, well aware of the fact that, in the $Dirac \to XY$  correspondence, $\lambda$ is related to the spacing adopted in the lattice approximation. In such a context it would be very hard to understand (without a projection mechanism \cite{christandl04}, which seems to have an exponential cost) the implementation of a space dependence such as
\begin{equation} \label{eq:lambdax}
	\lambda(x) = const.\, \sqrt{x(x-s)}
\end{equation}
that leads in \cite{peres85} and \cite{christandl04} to the existence of sharply distinguished instants in which the position of the cursor is certain. Nor would it be easy to understand \eref{eq:lambdax} in a solid state implementation \cite{burkard04}, where $\lambda$ is related to the effective mass of the clocking excitation.\\[5pt]
We have focused our attention on the \emph{clocking field} $\tau(x)$, singled out as the one which, under suitable boundary conditions and for initial conditions localized close to the boundary, exhibits particle-like excitations performing, for long enough intervals of $t$, a quantum walk in a distinguished direction.\\[5pt]
Spatial homogeneity of the chain leads to the existence of the limit in law \mbox{$V=\lim_{t \to  +\infty} Q(t)/t$} for the position of such an excitation on a semi-infinite ($s \to +\infty$) box. In the $N_3=1$ subspace, because of Peres' conservation law \cite{peres85}, the observable $Q$ acquires the meaning of \emph{relational time} (\underline{given} the observed value of $Q$, the state of the register is known with certainty) and, therefore, the random variable $V$ acquires the meaning of number of computational steps per unit $t$. The fact that the variance of $V$ is strictly positive has the effect that in terms of the \emph{parameter time} $t$ (as opposed to \emph{relational time} $Q$) the evolution of the register appears to be dissipative: we have, for a simple model, written the corresponding Lindblad evolution and studied the ensuing buildup of entropy.\\[5pt]
On our simple instance of quantum search we have shown that, in the ``low level'', physical approach that we pursue (in which time runs, for the register, because it is coupled with an additional quantum field) the buildup of entropy imposes an upper bound on the probability of finding the target state which is more severe than the one predicted by the ``high level'', algorithmic approach (in which the successive primitives are applied by an external macroscopic agent).\\[5pt]
In the attempt of decreasing the deficit in the probability of success in a quantum search, due to the decohering effect of the coupling with the clocking field, we have provided examples of the benefit of spreading the initial wave function of the cursor on an initial \emph{launch pad} instead of, as a classical prejudice would suggest, having it strictly localized at one site.\\[5pt]
The master equation of the state of the register satisfies the Lindblad equation; therefore, if on one side the clock makes the register evolve, on the other it acts as a decohering environment.\\[5pt]
We point out, furthermore, that the analysis of the probability cost of an algorithm will be significant only when a optimal preparation-measurement scheme will be found. The $\pi$-\emph{pulse trap} and the \emph{launch-pad} are proposals in this direction. We observe that, with the same amount of space resources, a longer \emph{sojourn time} can be achieved by means of an iteration of quantum subroutines circuit in which to trap the excitation.\\[5pt]
The study of the energy distribution of the system before and after the (non-optimal) measurement allows for the analysis of the resetting cost of the machine which, as long as we consider perfect chains, is the only energetic cost of the computation.\\[5pt]
The role of initial condition must be studied also when the initial state of the machine presents entanglement between the register and the clock; for example let's take the initial condition
\begin{equation}
	\ket{\xi_0}= \sum_{x=1}^{s}\sum_{j=1}^d \xi(j,x) \ket{r_j} \otimes \ket{C(x)},
\end{equation}
with $\{\ket{r_j}\}_{j=1}^d$ is a selected orthonormal basis of \Hr.\\
With the same techniques used in \sref{sec:speed} it is possible to show that
\begin{eqnarray}
	\phi_{V(\xi_0)}(z) & = & \lim_{t \to +\infty} \lim_{s \to +\infty} \bra{\xi_0} e^{itH}\exp \left(iz \frac{Q}{t}\right) e^{-itH}\ket{\xi_0} = \nonumber \\
	& = &\sum_{i=1}^d \int_0^1 e^{i z v} \frac{|\Xi(j,\arcsin(v))|^2+|\Xi(j,\pi-\arcsin(v))|^2}{\sqrt{1-v^2}} dv
\end{eqnarray}
where
\begin{equation}
	\Xi(p)=\sqrt{\frac{2}{\pi}} \sum_{x=1}^s \sin(px)\sum_{i=1}^{d}\bra{r_j} U_1^\dagger U_2^\dagger\ldots U_{x-1}^\dagger \ket{r_i} \xi(i,x)
\end{equation}
is the sine transform of the initial state. The interest of the most general form of the initial vector state relies on the possibility of exploiting the entanglement between the register and the cursor in order to either increase the average speed of computation or reduce the spreading of the wave packet. The analytic description of the interaction between the two subsystems models also some kind of decoherence appearing in spin chains.\\[5pt]
We have, moreover, abandoned the classical prejudice of having a single clocking excitation, providing a preliminary analysis of the idea of starting the cursor in an  initial state with $N_3 > 1$. We have shown, in this context, an efficient way of iterating the application of a single primitive to the register. Numerical examples of section \ref{sec:nonabelian} show that, however, as soon as the unitary operators acting on the register do not commute, the positions of the excitations have no longer the meaning of relational time of the system: the state of the register depends also on the path made by the excitations along the chain. Indeed, we showed numerical examples in which interference between different paths leads to localization \emph{\`a la} Anderson of the excitations. In that case, however, it is the presence of impurities in the lattice, or chain,  which leads to localized states. We refer to \cite{keating06} for a recent discussion of Anderson localization on spin chains.\\[5pt]
The case $N_3>1$ deserves, we think, further research, both from the algorithmic and the physical point of view.\\[5pt]
From the algorithmic point of view we plan to examine other instances (beyond the one cursorily examined at the end of section \ref{sec:nonabelian}) in which time-of-flight spectroscopy (based on the Fourier transform vs. speed relationship recalled in section \ref{sec:speed}) of the post-kickback state can answer Yes/No questions about the algorithm.\\[5pt]
From the physical point of view, the ``obvious'' choice of the ``all down'' reference state made throughout this work is far from being optimal from the point of view of studying the thermodynamic cost of resetting the register. The best reference state for the study of this ultimate cost of reversible computation would of course be the ground state and, for Hamiltonians of the form \eref{eq:freehamiltonian}, with $s$ even, it is an $N_3=s/2$ state. This will require, we think, the formulation of an appropriate Bethe Anzatz for the Hamiltonian \eref{eq:hamfeynman}.\\[5pt]
It is this last point which leads to some intrinsic difficulty in the analytic description of the system. In fact, if scattering theory includes an analytic treatment  of the perturbed Hamiltonian for the XY chain, it seems quite hard to do the same if the coefficients of the matrix belong to a non-Abelian algebra, which is our general case. Nevertheless, we have shown that the interaction between the register and the multi-handed clock shows new effects such as the confinement of the excitation resembling Anderson localization.\\[5pt]
In section \ref{sec:nonabelian} we have related the probability deficit of \fref{fig:figura12}(a) to the simultaneous presence of two of the excitations in the region between the active links $a \rightarrow a+1$ and $b \rightarrow b+1$, which corresponds to the application in the ``wrong'' order of the \emph{oracle} and \emph{estimation} step of Grover's algorithm. If we initialize the initial condition of the cursor in the ground state of the \emph{launch pad} $\{1,2,\ldots,6\}$ we partially fill the probability gap due to the \emph{sojourn} of the excitations in the \emph{critical region}. Similarly, the confinement effect shown in figures \ref{fig:gconfrontoloc} and \ref{fig:gconfronto} is related to the interference pattern between states of the form
\begin{eqnarray} \label{eq:ordine}
BABA\ket{\sigma_3=+1} \ket{\{q_1,q_2\}} & = & e^{i \pi} \ket{\sigma_3=+1} \ket{\{q_1,q_2\}} \label{eq:corretto}
\end{eqnarray}
and states of the form
\begin{eqnarray} \label{eq:disordine}
BBAA \ket{\sigma_3=+1} \ket{\{q_1,q_2\}} & = & \ket{\sigma_3=+1} \ket{\{q_1,q_2\}} \label{eq:sbagliato},
\end{eqnarray}
with $b+1 \leq q_1 < q_2 \leq s$.\\
The effect is similar to scattering on an impurity: the relative phase $e^{i \pi}$ between \eref{eq:ordine} and \eref{eq:disordine} leads, by a familiar kickback mechanism in which the register plays the role of an ancilla qubit, to destructive interference in the distribution of the cursor.\\[5pt]
The transmission amplitude is proportional to $(\alpha-
\beta)$, $\alpha$ being the amplitude associated to states of the form \eref{eq:corretto} and $\beta$ the amplitude of states of the form \eref{eq:sbagliato}.\\[5pt]
The trasmission probability deficit is compatible with the empirical distribution of the random variable \emph{sojourn time} in the critical region $\{(x_1,x_2):a<x_1<x_2\leq b\}$, with $a=9$ and $b=11$, shown in \fref{fig:sojourntime}.
\begin{figure}[h t]
	\centering
		\includegraphics{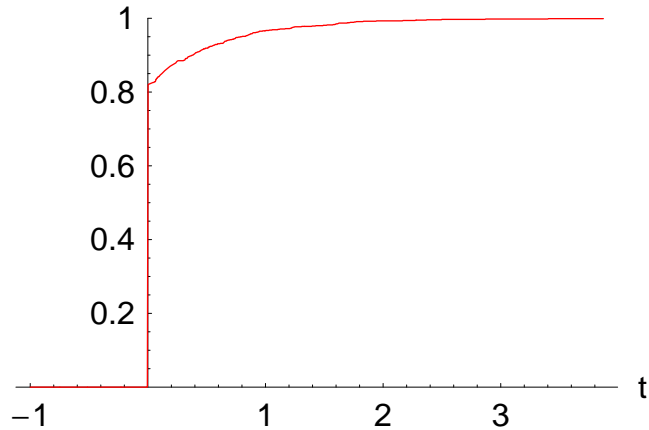}
		\caption{$s=20$; $H=H_0$ as in \eref{eq:freehamiltonian}; $N_3=2$; $a=9$, $b=11$; initial condition \ket{\{1,2\}}; sample of 1000 paths. The empirical  distribution function of the sojourn time in the critical region.}
	\label{fig:sojourntime}
\end{figure}\\
We define (see \cite{carlen}) the \emph{sojourn time} of a \emph{quantum} walk in a given region as the sojourn time of a canonically associated controlled \emph{random} walk. The idea that ``stochastic control theory can provide a very simple model simulating quantum mechanical behavior'' is borrowed from \cite{guerra83}. The actual details of the numerical simulation in our discrete settings implements the prescription of \cite{guerra84}.\\
What \fref{fig:sojourntime} says is that $80\%$ of the trajectories of our sample, of size $1000$, behave as the one shown in \fref{fig:traiet}.(a), never hitting the vertex $(11,10)$; only 20\% have at some time both excitations between $a+1$ and $b$. 
\begin{figure}[h]
	\centering		\subfigure[]{\includegraphics[width=7cm]{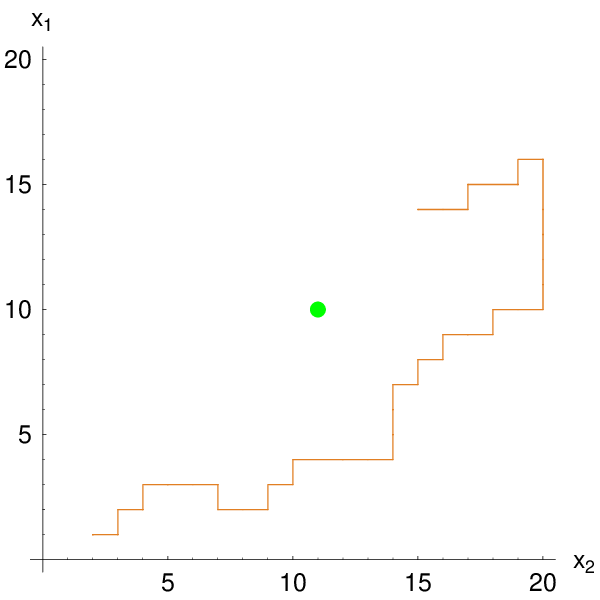} }	\subfigure[]{\includegraphics[width=7cm]{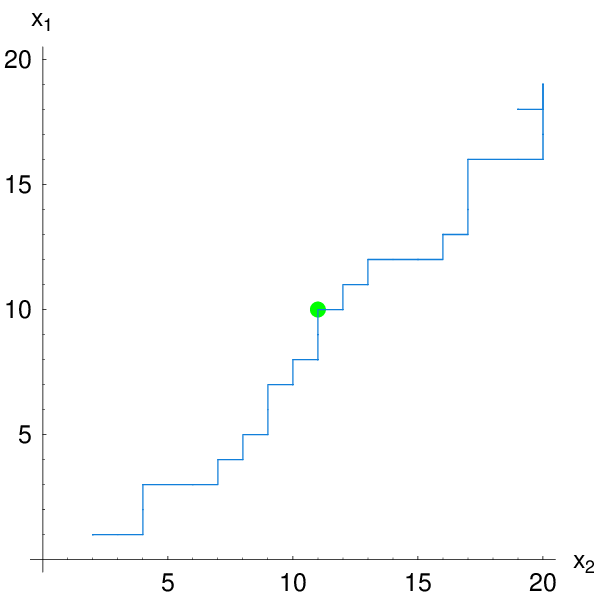} }
	\caption{Same parameters as in \fref{fig:sojourntime}.(a) A trajectory not passing through the critical region. (b) A trajectory traversing the critical region. The critical region, that is the configuration $(11,10)$, is highlighted by a (green) point.}
	\label{fig:traiet}
\end{figure}
As a sketch of the body of ideas relating \emph{quantum} walks to \emph{random} walks \cite{defa08} we discuss the example referring to the case $N_3=1$ represented in \fref{fig:graficotraj1}.
\begin{figure}[h]
	\centering	\includegraphics[width=14cm]{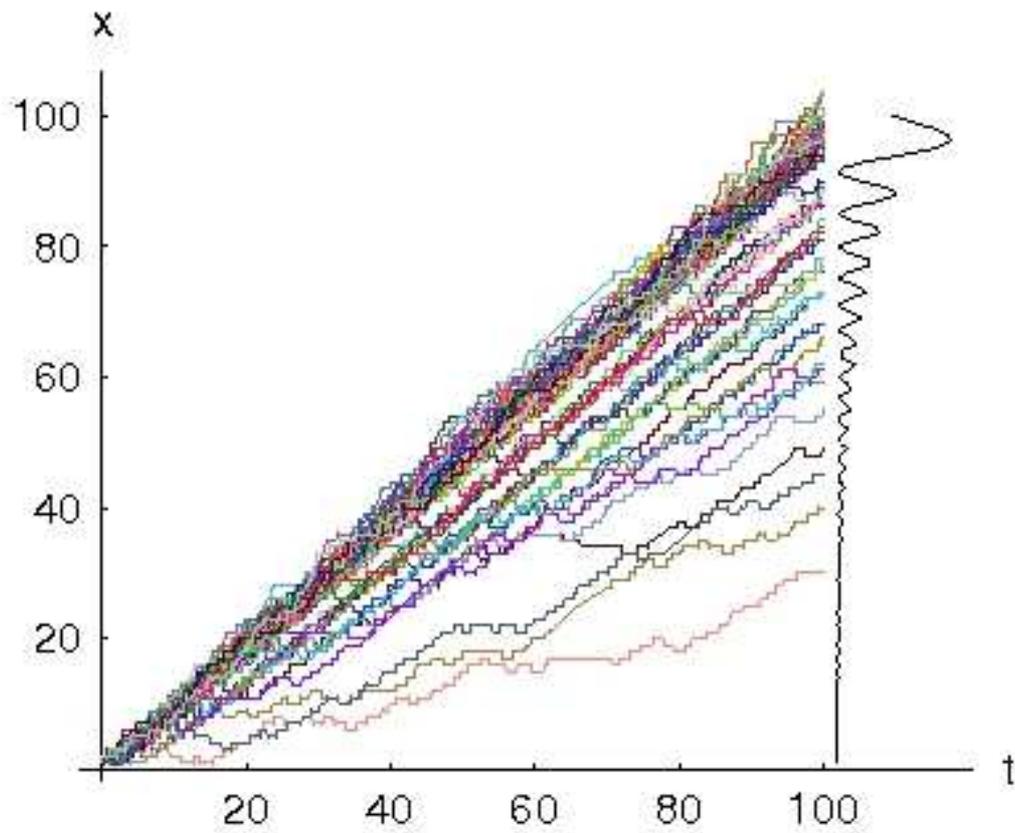}
	\caption{$s=100$, $H=H_0$; $N_3=1$; initial condition \ket{Q=1}. The probability density $\rho(t,x)=\frac{4x^2}{t^2}J_x^2(t)$ at the final instant $t=100$ is shown for comparison purposes.}
	\label{fig:graficotraj1}
\end{figure}
The whole point behind the algorithm leading to \fref{fig:graficotraj1} is that the probability density
\begin{equation}
	\rho(t,x)=\frac{4x^2}{t^2}J_x^2(t)
\end{equation}
corresponding to the amplitude \eref{eq:bj} satisfies the equation
\begin{eqnarray} \label{eq:nm}
	\frac{d}{dt}\rho(t,x) & = & -sign(J_x(t) J_{x+1}(t)) \sqrt{\rho(t,x)\rho(t,x+1)}+ \nonumber \\
	& + &  sign(J_x(t) J_{x-1}(t)) \sqrt{\rho(t,x-1)\rho(t,x)}.
\end{eqnarray}
Equation \eref{eq:nm} can, in turn, be read as the continuity equation for a \emph{birth and death process} on  $\Lambda_{\infty} = \{1,2,\ldots\}$. The sample paths $(t,q(t))$ of this process are shown in \fref{fig:graficotraj1}.\\[5pt]
The stochastic simulation of quantum phenomena discussed in \cite{carlen,guerra83,guerra84} raises extremely interesting questions about the foundations of quantum mechanics itself \cite{nelson85}.\\[5pt]
We wish to stress as a final remark that stochastic mechanics poses also interesting questions in the theory of quantum computation: is the exponential speedup offered by \emph{quantum} walks in crossing a graph or decision tree \cite{farhi97,childs02}  attainable by a \emph{classical} stochastic algorithm?
\backmatter
%
\bibliography{mybib}
\bibliographystyle{unsrt}
\end{document}